\documentclass[10pt]{iopart}
\usepackage{graphicx}
\graphicspath{{}}
\expandafter\let\csname equation*\endcsname\relax
\expandafter\let\csname endequation*\endcsname\relax
\usepackage[tbtags]{amsmath}%amsthm is included in amsmath package centertags
\usepackage{amsfonts,amssymb}
\usepackage{mathabx}
\usepackage{mathtools}
\mathtoolsset{showonlyrefs=fasle,showmanualtags}
\usepackage{amsthm} % use amsthm envirenments
\theoremstyle{plain}% default
\newtheorem{theorem}{Theorem}
\newtheorem{lem}[theorem]{Lemma}
\newtheorem{prop}[theorem]{Proposition}

\newtheorem{as}{Assumption}
\theoremstyle{definition}

\theoremstyle{remark}
\newtheorem*{remark}{Remark}
\usepackage{amssymb}
\usepackage[english]{babel}
\usepackage{color}
\usepackage{latexsym,bm}
\usepackage{algorithm}
\usepackage{algorithmic}
\usepackage{float}

\usepackage{comment}
\usepackage{subfigure}
\newcommand{\norm}[1]{\left\lVert#1\right\rVert}

\newcommand{\rd}{\,\mathrm{d}}

\begin{document}

\title[On Iterative Algorithms for QPAT in the Radiative Transport Regime]{On Iterative Algorithms for Quantitative Photoacoustic Tomography in the Radiative Transport Regime}

\author{Chao Wang$^1$ \& Tie Zhou$^2$}

\address{$^1$$^2$ School of Mathematical Sciences, Peking University, China}
\ead{chaowyww@pku.edu.cn}
\ead{tzhou@math.pku.edu.cn}
\begin{abstract}
In this paper, we present the numerical reconstruction method for quantitative photoacoustic tomography (QPAT) based on the radiative transfer equation (RTE), which models light propagation more accurately than diffusion approximation (DA). We investigate the reconstruction of absorption coefficient and scattering coefficient of biological tissues. Given the scattering coefficient, an improved fixed-point iterative method is proposed to retrieve the absorption coefficient for its cheap computational cost. And the convergence of this method is also proved. Barzilai-Borwein (BB) method is applied to retrieve two coefficients simultaneously. Since the reconstruction of optical coefficients involves the solutions of original and adjoint RTEs in the framework of optimization, an efficient solver with high accuracy is developed from~\cite{Gao}. Simulation experiments illustrate that the improved fixed-point iterative method and the BB method are the competitive methods for QPAT in two cases.
\end{abstract}
%\tableofcontents
\section{Introduction}
\label{sec:1}
 Photoacoustic tomography (PAT) is a developing medical imaging technique using non-ionizing waves. It combines ultrasound imaging and optical imaging \cite{Xu2006}, so it can achieve high contrast and high resolution simultaneously. Photoacoustic imaging is based on photoacoustic effect. A short pulse of electromagnetic energy propagates through the tissue rapidly at speed of higher magnitude than sound after the tissue is illuminated. Then the electromagnetic energy is absorbed by the tissue and converted to heat. There quickly follows the temperature of tissue increasing which leads to  expansion spatially, and the expansion induces a pressure field. Acoustic waves generated by pressure distribution will be measured by detectors placed at the surface. They can be used to extract physical information about tissue inside. In medical imaging science, this technique is mostly applied to tomography imaging of skin, breast cancer and small animal imaging \cite{Xu2006}.

In the physical view, there are two inverse problems associated with two stages: the first stage of generating pressure after absorbing electromagnetic energy and second stage of initial pressure generating acoustic waves. First problem (aka PAT) is to reconstruct initial pressure distribution from temporal measurements of photoacoustic waves, and the second one, called quantitative photoacoustic tomography (QPAT), is to estimate optical coefficients from initial pressure distribution. Although initial pressure reflects some information inside, it is the product of the absorption of energy relied on properties of the object indirectly. Only the intrinsic optical coefficients provide direct information. Essentially, spatially varying optical coefficients lead to spatially varying initial pressure distribution. Moreover, reference \cite{Mamonov2014} mentions that initial pressure decreases quickly in the tissue area with small absorption coefficient, so that we can not distinguish visually in the finial image of initial pressure. In a word, the original information in the investigated object can not be obtained completely until the optical coefficients are recovered. Generally, estimated optical coefficients comprise of absorption coefficient, scattering coefficient and Gr\"{u}eneisen coefficient, where absorption coefficient is usually of the major clinic interest \cite{Gao2012}. Nowadays, algorithms on reconstruction of initial pressure have been developed maturely~\cite{Kuchment,Xu2006,Xu2002,Agranovsky2007,Lv2014}, so a lot of researchers are focusing on QPAT.

In terms of mathematical model of QPAT, light propagation inside the tissue can be modeled by radiative transfer equation (RTE)~\cite{Saratoon2013,Mamonov2014,Haltmeier2015} or diffusion approximation (DA)~\cite{Bal2012,Ding2015}. However, DA approximates to RTE only in the situation where light behaves diffusely, which is invalid for small object or in regions close to light source, so that the optical coefficients estimated from DA model are not accurate. The reconstruction theories and some algorithms based on DA model are presented in~\cite{Pulkkinen2014,Pulkkinen2016,Bal2010,Feng2016,Zhang2014,Beretta2015}. The two models are compared numerically in~\cite{Tarvainen2012} and the derivation from RTE model to DA model can be found in~\cite{Cox2012}.

In terms of measurement, there are mainly three kinds of methodologies: single illumination, multi-source illumination and multi-spectral illumination. It has been proved that single illumination can recover any one of absorption and Gr\"{u}eneisen coefficient stably given scattering coefficient, and even there is an explicit analytic formula for the non-scattering media \cite{Mamonov2014}. Multi-source can recover two of the three coefficients mentioned above uniquely and simultaneously~\cite{Mamonov2014,Haltmeier2015,Bal,Shao2011}. Furthermore, multi-spectral illumination can achieve unique reconstruction of three coefficients \cite{Bal2009}. We focus on the RTE model with multi-source illumination and aim to investigate more efficient reconstruction method in this setting.

In term of reconstruction algorithm, assuming the knowledge of scattering coefficient, fixed-point iteration is used to recover absorption coefficient in~\cite{Harrison2013}. However, the requirement of close enough initial value and the lack of theoretical results are obstacles of its application. When the scattering coefficient is unknown, QPAT can be formulated as a nonlinear optimization problem~\cite{Haltmeier2015,Gao2012} and furthermore can be simplified to a linear optimization problem by Born linearization~\cite{Mamonov2014}. Optimization approaches such as Jacobian-based method and gradient-based method are investigated in reference~\cite{Saratoon2013}, in which Limited-memory BFGS (LBFGS), as a state of art gradient-based method, has been applied to the problem~\cite{Gao2012}.

In the optimization framework, it is inevitable to solve original RTE in forward problem and adjoint RTE in the process of calculating gradient of objective function. Finite element method combined with streamline diffusion modification is applied to solve stationary RTE~\cite{Haltmeier2015,Saratoon2013,Yao2010}. In this way, a large sparse linear system needs to be solved, which is still difficult. Inspired by \cite{Gao}, we apply Discontinuous Galerkin (DG) method combined with multigrid method to solve the two RTEs.

\textbf{Our contributions} in this paper are in several folds.
\begin{enumerate}
 \item  We propose an improved fixed-point iterative algorithm and prove its convergence given scattering coefficient. Numerically, the improved fixed-point iteration converges faster than the optimization-based Barzilai-Borwein (BB) method.
\item We apply BB method to reconstruct absorption coefficient and scattering coefficient simultaneously. To the best of our knowledge, it is the first time to apply BB method to this problem.
\item We revise the numerical scheme for RTE proposed in \cite{Gao}, which fails to solve adjoint RTE. We revise it to make it suitable for the adjoint RTE and prove its convergence. Then the RTE can be reduced to a sparse block diagonal linear system finally.
\end{enumerate}

The rest of the paper is organized as follows. We introduce mathematical model of forward problem and inverse problem as well as notation conventions in Section \ref{sec:2}. Given scattering coefficient, an improved fixed-point iterative algorithm and relevant proofs for convergence are presented in Section \ref{sec:3.1}. The BB method for QPAT and the deduction of the gradient of objective function are detailed in Section \ref{sec:3.2}. And we discuss solvers for original RTE and adjoint RTE respectively in Section \ref{sec:4.1}. The simulation results are shown in Section \ref{sec:4.3}. The conclusion comes in Section \ref{sec:5}.

\section{The mathematical description}
\label{sec:2}

In the following, we consider the convex bounded object region $\Omega\in \mathcal{R}^{n} (n=2,3)$ with Lipschitz boundary $\partial \Omega$ and angular space $\mathcal{S}^{n-1}$. Describing boundary condition, we divide boundary $\partial \Omega$ into inflow boundary and outflow boundary
 \begin{equation}
\label{eq:1}
    \begin{aligned}
      \Gamma_-:=&\left\{(x,\theta)\in\partial \Omega\times \mathcal{S}^{n-1}:\nu(x)\cdot \theta<0\right\},\\
\Gamma_+:=&\left\{(x,\theta)\in\partial \Omega\times \mathcal{S}^{n-1}:\nu(x)\cdot \theta>0\right\},
    \end{aligned}
  \end{equation}
where $\nu(x)$ is the outward normal vector at the position $x$ of the boundary $\partial \Omega$.

\subsection{QPAT model}
\label{sec:2.1}

We use stationary RTE to describe the light propagation in the form of
\begin{equation}
  \label{eq:2}
  \left[\theta\cdot \nabla +\mu_a(x)+\mu_s(x)\right]\phi(x,\theta)-\mu_s(x)\oint_{\mathcal{S}^{n-1}}f(\theta,\theta')\phi(x,\theta')\rd\theta'=q(x,\theta),\quad (x,\theta)\in \Omega\times \mathcal{S}^{n-1},
\end{equation}
where $\theta$ and $x$ are the direction and the position of interest respectively, $\phi(x,\theta)$ is the density of energy at the position $x$ in the direction $\theta$, $\mu_a(x)$ and $\mu_s(x)$ are spatially varying absorption and scattering coefficients respectively, and $q(x,\theta)$ is the source term. Function $f(\theta,\theta')$ describing the probability of photon traveling from direction $\theta'$ to $\theta$ is called scattering phase function and usually is characterized by anisotropic factor $g$ of the form
\begin{equation}
  \label{eq:3}f(\theta,\theta')=\left\{
  \begin{aligned}
& \frac{1-g^2}{2\pi(1+g^2-2g\theta\cdot\theta')},\quad && n=2,\\
   & \frac{1-g^2}{4\pi(1+g^2-2g\theta\cdot\theta')^{3/2}},\quad  && n=3,
  \end{aligned}
\right.
\end{equation}
which is the well-known Henyey-Greenstein (H-G) scattering function. For the sake of simplicity, we define scattering operator $\bm{K}$ by
\begin{equation}
  \label{eq:03}
  (\bm{K}\phi)(x,\theta):=\oint_{\mathcal{S}^{n-1}}f(\theta,\theta')\phi(x,\theta')\rd \theta'.
\end{equation}
We assume inflow boundary condition
\begin{equation}
  \label{eq:4}
  \phi(x,\theta)=q_b(x,\theta), \quad(x,\theta)\in \Gamma_-,
\end{equation}
that is no photon getting in $\Omega$ except for boundary source $q_b(x,\theta)$. In QPAT, source term $q(x,\theta)$ is usually regarded as zero. Absorbed energy density is presented by
\begin{equation}
  \label{eq:5}
  h(x):=\mu_a(x)\Phi(x),
\end{equation}
where
\begin{equation}
  \label{eq:6}
  \Phi(x):=(\bm{A}\phi)(x)=\oint_{\mathcal{S}^{n-1}}\phi(x,\theta)\rd \theta,
\end{equation}
in which $\bm{A}$ is the accumulation operator over all directions.

Owing to absorbed energy, initial acoustic pressure distribution is generated of the form
\begin{equation*}
  p_0(x):=\gamma(x)h(x),
\end{equation*}
where $\gamma(x)$ is spatially varying Gr\"{u}eneisen coefficient. $\gamma(x)$ is the efficiency of conversion of absorbed energy to pressure and is assumed to be known. Throughout this work, it is rescaled to be one, i.e. $\gamma(x)=1$.

\subsection{PAT model}
\label{sec:2.2}
On account of initial pressure distribution, acoustic waves propagate through object region $\Omega$ and a series of temporal acoustic signals are detected by ultrasonic detectors at the surface $\partial \Omega$. The behavior is described by
\begin{equation}
  \label{eq:8}
  \left\{
\begin{aligned}
  &\frac{1}{c^2(x)}\frac{\partial^2}{\partial t^2}p(x,t)-\Delta p(x,t)=0, \quad &&(x,t)\in \mathcal{R}^n\times(0,T], \\
&p(x,0)=p_0(x),&& x\in \Omega,\\
&\frac{\partial p}{\partial t}(x,0)=0,&& x\in \Omega,
\end{aligned}
\right.
\end{equation}
where $c(x)$ is the acoustic speed, $p(x,t)$ is the acoustic pressure at the position $x$ and time $t$, and $p_0(x)$ is the initial pressure distribution.

The goal of PAT is to recover optical coefficients $\mu_a$ and $\mu_s$ given temporal data $p(x,t)\ ((x,t)\in \partial\Omega\times (0,T])$. This is often accomplished in two stages: first one recovers $p_0(x)$ from temporal data $p(x,t)\ ((x,t)\in \partial\Omega\times (0,T])$, and next one reconstructs $\mu_a$ and $\mu_s$ from initial pressure $p_0$ (or $h$) given illumination conditions. We assume the first stage has been done, and the second stage QPAT is our concern. We assume tissue is illuminated for $M$ times with known boundary sources $q_{b_m}\ (m=0,1,\dots,M-1)$, then inverse problem QPAT is to determine $\mu_a$ and $\mu_s$ from data $h_m$, that is $h(x;\mu_a,\mu_s,q_{b_m})\ (m=0,1,\dots,M-1)$.

\subsection{Notation conventions}
\label{sec:2.3}
We denote the absorption coefficient and scattering coefficient by $\mu_a^i$ and $\mu_s^i$ in the $i$th iteration respectively, and corresponding solution of RTE and heat with $m$th illumination of boundary source $q_{b_{m}}\ (m=0,1,\dots,M-1)$ is denoted by $\phi_m^i(x,\theta)$ and $h^i_m(x)$. The exact coefficients and data are denoted by $\mu_a^*$, $\mu_s^*$, $\phi^*_m$, and $h^*_m$.

\section{Reconstruction for optical coefficients}
\label{sec:3}
In this section, we propose an improved fixed-point iterative method to recover absorption coefficient given scattering coefficient and then apply Barzilai-Borwein (BB) method to simultaneously recover scattering coefficient as well. First, it is necessary to make several assumptions on the two optical coefficients.

\begin{as}
\label{as:1}
  Assume optical coefficients $\mu_a$, $\mu_s$ and energy density $\phi_m(x,\theta;\mu_a,\mu_s)$, which is the solution of RTE with $\mu_a$ and $\mu_s$, satisfy
  \begin{enumerate}
  \item $\mu_a(x)\in \mathcal{D}_a:=\{\mu_a\in C(\overline{\Omega}):0<\mu_a^0\leq\mu_a(x)\leq\mu_a^{\text{upper}}\}$ for fixed $\mu_a^{\text{upper}}$;
  \item$\mu_s(x)\in \mathcal{D}_s:=\{\mu_s\in C(\overline{\Omega}):0<\mu_s^0\leq \mu_s(x)\leq \mu_s^{\text{upper}}\}$ for fixed $\mu_s^{\text{upper}}$;
  \item Boundary condition $q_{b_m}(m=0,1,\dots,M-1)\in L^\infty(\Gamma_-,|\theta \cdot \nu|)$;
  \item Scattering kernel $f(\theta,\theta')\in L^1(\mathcal{S}^{n-1}\times \mathcal{S}^{n-1})$;
  \item If $\mu_a$ and $\mu_s$ satisfy (i) and (ii), there exist boundary sources $q_{b_{m}}(m=0,1,\dots,M-1)$ such that $\phi_m(x,\theta)\geq\epsilon>0$ for some positive $\epsilon$.
  \end{enumerate}

\end{as}

\subsection{The case of given scattering coefficient}
\label{sec:3.1}

Under the Assumption \ref{as:1}, we provide an effective fixed-point iterative method. Since one measurement can recover absorption coefficient, subscripts `$m$' of symbols mentioned in Section~\ref{sec:2.3} are omitted for brevity. The algorithm is detailed as follows.

\begin{algorithm}[H]
\caption{Improved fixed-point iteration\label{alg:1}}
\begin{algorithmic}[1]
 \REQUIRE Given initialization $\mu_a^0$ mentioned in Assumption~\ref{as:1}, data $h^*$, scattering coefficient $\mu_s$, boundary source $q_b$, and tolerance $\epsilon$.
\FOR {$i=0,1,\dots$}
\STATE Solve stationary RTE with absorption $\mu_a^i$ and scattering coefficients $\mu_s$ to obtain the solution $\phi^i$. Then let $h^i(x)=\mu_a^i(x)(\bm{A}\phi^i)(x)$;
\STATE If $\norm{h^i-h^*}_1<\epsilon$, break with $\mu_a=\mu_a^i$; or implement next step;
\STATE Calculate $\widetilde{\mu_a}^{i+1}(x)=\frac{h^*(x)}{(\bm{A}\phi^i)(x)}$;
\STATE Calculate $\mu_a^{i+1}(x)=\max\{\mu_a^i(x),\widetilde{\mu_a}^{i+1}(x)\}$.
\ENDFOR
\end{algorithmic}
\end{algorithm}

\begin{remark}
The Assumption \ref{as:1} is important because $\phi$ is a denominator in the Algorithm \ref{alg:1}. In practice, we can also use $\phi+\varepsilon$ instead with the small $\varepsilon$ to avoid the situation where $\phi=0$.
\end{remark}

Fixed-point iterative method has been mentioned in many references~\cite{Harrison2013,Cox2009,Gao2012,Cox2006}, which is usually unstable and sensitive to initial value. However, the modification of fifth step in Algorithm~\ref{alg:1} can guarantee convergence given scattering coefficient.

Carrying out Algorithm \ref{alg:1}, we can obtain function sequences $\{\widetilde{\mu_a}^i(x)\}_{i=0}^{\infty}$, $\{\mu_a^i(x)\}_{i=0}^{\infty}$, $\{h^i(x)\}_{i=0}^{\infty}$ and $\{\phi^i(x)\}_{i=0}^{\infty}$. We claim that the sequence $\norm{h^i-h^*}_1$ converges to zero. The proof is presented as follows.

Firstly, we state a theorem on stationary RTE.
\begin{theorem}[\cite{Case1963}]
\label{th:1}
  For RTE \eqref{eq:2} satisfying Assumption \ref{as:1}, if source $q\geq 0$ and boundary source $q_b\geq 0$, there exists a unique non-negative and continuous solution $\phi(x,\theta)$.
\end{theorem}

In the view of physics, the intuitive explanation behind Theorem \ref{th:1} is straight-forward. For non-negative initial energy, the energy will not become negative while it was transferred and absorbed. Lemma~\ref{lem:2} reveals the monotonicity of $\phi$ with respect to $\mu_a$.

\begin{lem}
\label{lem:2}
  Let $\phi^1(x)$ and $\phi^2(x)$ be the solutions of RTEs
  \begin{equation}
    \label{eq:9}\left\{
    \begin{aligned}
      &\left[\theta\cdot \nabla +\mu_a(x)+\mu_s(x)\right]\phi(x,\theta)-\mu_s(x)(\bm{K}\phi)(x,\theta)=0,\ && (x,\theta)\in \Omega\times \mathcal{S}^{n-1},\\
&\phi(x,\theta)=q_b(x,\theta), &&(x,\theta)\in \Gamma_-,
    \end{aligned}
\right.
  \end{equation}
with $\mu_a$ being $\mu_a^1$ and $\mu_a^2$ respectively.
Then $\phi^1(x,\theta)\geq \phi^2(x,\theta)$ provided $\mu_a^1(x)\leq \mu_a^2(x)(\forall x\in \Omega)$. Note that the superscripts of $\mu_a^1$ and $\mu_a^2$ are used to distinguish the different absorption coefficients.
\end{lem}

\begin{proof}
 We see that $\phi:=\phi^1-\phi^2$ solves the RTE
\begin{equation}
  \label{eq:11}
\left\{
\begin{aligned}
  &\left[\theta\cdot \nabla+\mu_a^2+\mu_s\right]\phi(x,\theta)-\mu_s(\bm{K}\phi)(x,\theta)=(\mu_a^2-\mu_a^1)\phi^1(x,\theta),\ &&(x,\theta)\in \Omega\times \mathcal{S}^{n-1}\\
&\phi(x,\theta)=0,\ && (x,\theta)\in \Gamma_-
\end{aligned}
\right.
\end{equation}
with absorption coefficient $\mu_a^2$ and scattering coefficient $\mu_s$. Since source term $(\mu_a^2-\mu_a^1)\phi^1$ and boundary source are both non-negative, we can derive
\begin{equation}
  \label{eq:12}
  \phi^1\geq \phi^2
\end{equation}
using Theorem \ref{th:1}.
\end{proof}

Using Lemma~\ref{lem:2}, we investigate the properties of function sequences from Algorithm \ref{alg:1}.

\begin{lem}
\label{lem:3}
For any $x\in \Omega$, the sequence $\{\mu_a^i(x)\}_{i=0}^\infty$ obtained in the Algorithm \ref{alg:1} satisfies
  \begin{equation}
    \label{eq:13}
    \mu_a^0(x)\leq \mu_a^1(x)\leq \dots \leq \mu_a^i(x)\leq \dots \leq \mu_a^*(x).
  \end{equation}
\end{lem}
  \begin{proof}

Given $\mu_a^0(x):=\mu_a^0\leq \mu_a^*(x)$, let us assume
\begin{equation}
  \label{eq:14}
  \mu_a^0(x)\leq \mu_a^1(x)\leq \dots \leq \mu_a^i(x) \leq \mu_a^*(x).
\end{equation}
Then it suffices to establish \eqref{eq:14} for $(i+1)$.

Obviously, $\mu_a^i(x)\leq \mu_a^*(x)$ indicates $\phi^i(x)\geq \phi^*(x)$. Since $\widetilde{\mu_a}^{i+1}(x)=\frac{h^*(x)}{(\bm{A}\phi^i)(x)}$ and $\phi^i\geq \phi^*$, combining $h^*=\mu_a^*(\bm{A}\phi^*)$, we infer 
\begin{equation*}
  \widetilde{\mu_a}^{i+1}(x)\leq \mu_a^*(x).
\end{equation*}
According to Algorithm \ref{alg:1},
\begin{equation*}
  \mu_a^i(x)\leq\mu_a^{i+1}(x)\leq \mu_a^*(x).
\end{equation*}
An easy induction completes the proof.
\end{proof}

Monotonicity of $\mu_a^i$ can easily deduce the monotonicity of $\phi^i$. 
\begin{lem}
\label{lem:4}
For any $(x,\theta)\in \Omega\times \mathcal{S}^{n-1}$, the sequence $\{\phi^i(x,\theta)\}_{i=0}^\infty$ obtained in the Algorithm \ref{alg:1} satisfies
  \begin{equation}
    \label{eq:17}
\phi^0(x,\theta)\geq \phi^1(x,\theta)\geq \dots \geq \phi^i(x,\theta)\geq \dots \geq \phi^*(x,\theta).
  \end{equation}
\end{lem}
\begin{proof}
  It is obvious using Lemma \ref{lem:2} and Lemma \ref{lem:3}.
\end{proof}

We analyze the update process for $\mu_a$ by introducing new definitions. For any integer $i$, we divide region $\Omega$ into two parts in the form of
\begin{equation*}
  \left\{
    \begin{aligned}
      \Omega_+^i&:=\{x\in \Omega:h^*(x)\geq h^i(x)\},\\
\Omega_-^i&:=\{x\in \Omega:h^*(x)< h^i(x)\}.
    \end{aligned}
\right.
\end{equation*}
Then we have
\begin{equation*}
\begin{aligned}
  \widetilde{\mu_a}^{i+1}&=\frac{h^*}{\bm{A}\phi^i}\\
&=\left\{
  \begin{aligned}
    &\geq \mu_a^i,\ && x\in \Omega_+^i,\\
    &<\mu_a^i,\  && x\in \Omega_-^i.
  \end{aligned}
\right.
\end{aligned}
\end{equation*}
According Algorithm \ref{alg:1},
\begin{equation*}
  \mu_a^{i+1}(x)=\left\{
  \begin{aligned}
    &\widetilde{\mu_a}^{i+1}(x), \ &&x\in \Omega_+^i, \\
    & \mu_a^i, && x\in \Omega_-^i.
  \end{aligned}
\right.
\end{equation*}
Actually, it is a process of keeping and updating from $\mu_a^i$ to $\mu_a^{i+1}$, keeping the value of $\mu_a^i$ of $\Omega^i_-$, updating the value of $\Omega^i_+$.

Eventually, we prove the convergence of function sequence $\{h^i(x)\}_{i=0}^\infty$.

\begin{theorem}
  \label{the:5}
For any $x\in \Omega$, the sequence $\{h^i(x)\}_{i=0}^\infty$ obtained in the Algorithm \ref{alg:1} satisfies
\begin{equation}
  \label{eq:18}
  \lim_{i\rightarrow \infty}\|h^*-h^i\|_1=0.
\end{equation}

\end{theorem}
\begin{proof}
We divide our proof into four claims.

(i): Function sequences $\{\mu_a^i(x)\}_{i=0}^\infty$, $\{\phi^i(x,\theta)\}_{i=0}^\infty$ and $\{h^i(x)\}_{i=0}^\infty$ converge pointwisely.

Since monotonicity and boundedness of $\{\mu_a^i(x)\}_{i=0}^\infty$ and $\{\phi^i(x,\theta)\}_{i=0}^\infty$ for any $(x,\theta)\in \Omega\times\mathcal{S}^{n-1}$, it follows that they converge pointwisely, whose limits are denoted by $\overline{\mu_a}(x)$ and $\overline{\phi}(x,\theta)$ respectively. Obviously, $\{h^i(x)\}_{i=0}^\infty$ converges pointwisely and its limit is denoted by $\overline{h}(x)$. And obviously, $\overline{h}(x)=\overline{\mu_a}(x)(\bm{A}\overline{\phi})(x)$.

(ii): For any $x\in \Omega$, if there is some $I$ such that if $x\in \Omega^I_+$, then $x\in \Omega^{i}_+(i>I)$.

It is obvious that $\mu_a^{I+1}(x)=\widetilde{\mu_a}^{I+1}(x)\geq \mu_a^I(x)$. According to Lemma \ref{lem:4}, we can obtain
\begin{equation*}
  h^{I+1}(x)=\mu_a^{I+1}(\bm{A}\phi^{I+1})(x)\leq \mu_a^{I+1}(x)(\bm{A}\phi^I)(x)=h^*(x).
\end{equation*}
And an easy induction shows that $x\in \Omega^{i}_+(i>I)$.

(iii): $\overline{h}(x)=h^*(x)\ (x\in \Omega_+^0)$, $\|h^*-h^i\|_{L^1(\Omega_+^0)}\rightarrow 0\ (i\rightarrow \infty)$.

Thanks to $x\in \Omega^i_+(i>0)$ for any $x\in \Omega^0_+$, it follows that
\begin{equation}
\label{eq:19}
\begin{aligned}
  \mu_a^{i+1}(x)-\mu_a^i(x)&=\widetilde{\mu_a}^{i+1}(x)-\mu_a^i(x)\\
&=\frac{h^*(x)-h^i(x)}{\bm{A}\phi^i(x)}\\
&\geq \frac{h^*(x)-h^i(x)}{2\pi\overline{\phi^0}},
\end{aligned}
\end{equation}
where $\overline{\phi^0}=\sup_{x\in \Omega}\phi^0(x)$. According to Lemma \ref{lem:4},
\begin{equation*}
  \phi^i(x)\leq \phi^0(x)\leq \overline{\phi^0}
\end{equation*}
gives last inequality in \eqref{eq:19}. Then
\begin{equation*}
  \overline{\mu_a}(x)-\mu_a^0(x)=\sum_{i=0}^\infty(\mu_a^{i+1}(x)-\mu_a^i(x))\geq\frac{1}{2\pi\overline{\phi^0}} \sum_{i=0}^\infty (h^*(x)-h^i(x)).
\end{equation*}
Owing to $ 0\leq\overline{\mu_a}-\mu_a^0<\infty$, we can obtain that
\begin{equation}
  \label{eq:20}
h^*(x)-h^i(x)\xrightarrow{i\rightarrow \infty} 0\ (\forall x \in \Omega^0_+),
\end{equation}
that is, $h^*(x)-h^i(x)$ converges to 0 pointwisely. So $\overline{h}(x)=h^*(x)\ (x\in \Omega_+^0)$.
Integrating \eqref{eq:19} with respect $x$ over $\Omega_+^0$ gives
\begin{equation*}
  \|\mu_a^{i+1}-\mu_a^i\|_{L^1(\Omega_+^0)}\geq \frac{1}{2\pi\overline{\phi^0}}\|h^*-h^i\|_{L^1(\Omega_+^0)}.
\end{equation*}
Since
\begin{equation*}
\begin{aligned}
  \|\overline{\mu_a}-\mu_a^0\|_{L^1(\Omega_+^0)}&=\int_{\Omega_+^0}(\overline{\mu_a}(x)-\mu_a^0(x))\rd x\\
&=\int_{\Omega_+^0}\sum_{i=0}^\infty(\mu_a^{i+1}(x)-\mu_a^i(x))\rd x\\
&=\sum_{i=0}^\infty\int_{\Omega_+^0}(\mu_a^{i+1}(x)-\mu_a^i(x))\rd x\\
&\geq\frac{1}{2\pi\overline{\phi^0}} \sum_{i=0}^\infty\int_{\Omega_+^0}(h^*(x)-h^i(x))\rd x\\
&=\frac{1}{2\pi\overline{\phi^0}} \sum_{i=0}^\infty\|h^*-h^i\|_{L^1(\Omega_+^0)}
\end{aligned}
\end{equation*}
and $ 0\leq\|\overline{\mu_a}-\mu_a^0\|_{L^1(\Omega_+^0)}<\infty$, directly,
\begin{equation}
  \label{eq:21}
  \|h^*-h^i\|_{L^1(\Omega_+^0)}\rightarrow 0\quad (i\rightarrow \infty).
\end{equation}
(iv): $\|h^*-h^i\|_{L^1(\Omega_-^{0})}\rightarrow 0\ (i\rightarrow \infty$).

We divide region $\Omega_-^0$ into two parts of form
\begin{equation*}
  \left\{
    \begin{aligned}
      \Omega_-^{01}:&=\{x\in \Omega_-^0:\text{Exist some }i \text{ such that } x\in \Omega_+^i\},\\
\Omega_-^{02}:&=\{x\in \Omega_-^0:\text{For any }i, x\in \Omega_-^i\}.
    \end{aligned}
\right.
.
\end{equation*}
According the proof of (iii) and the definition of $\Omega^i_-$, it can easily be seen that
\begin{equation}
\label{eq:22}
 \overline{h}(x)
  \left\{
    \begin{aligned}
    = h^*(x),\quad x\in \Omega_-^{01},\\
\geq h^*(x),\quad x\in \Omega_-^{02},
    \end{aligned}
\right.
\end{equation}
and
\begin{equation}
\label{eq:23}
  \|h^i-h^*\|_{L^1(\Omega_-^{01})}\rightarrow 0.
\end{equation}
From \eqref{eq:22}, it is clear that
\begin{equation}
\|\overline{h}\|_{L^1(\Omega_-^{02})}\geq\|h^*\|_{L^1(\Omega_-^{02})}.
  \label{eq:24}
\end{equation}
Considering RTEs with absorption coefficient $\mu_a^i$ and $\mu_a^*$, it is easy to obtain
\begin{equation}
  \label{eq:15}
\theta\cdot \nabla(\phi^i-\phi^*)+\mu_a^i\phi^i-\mu_a^*\phi^*=-(\mu_s\bm{I}-\mu_s\bm{K})(\phi^i-\phi^*).
\end{equation}
Since phase function $f(\theta, \theta')$ satisfies
\begin{equation*}
  \oint_{\mathcal{S}^{n-1}}f(\theta,\theta')\rd \theta=1,
\end{equation*}
it follows that
\begin{equation*}
\begin{aligned}
  \oint_{\mathcal{S}^{n-1}}\mu_s(x)(\bm{K}\phi)(x,\theta)\rd \theta&=\oint_{\mathcal{S}^{n-1}}\mu_s(x)\oint_{\mathcal{S}^{n-1}}f(\theta,\theta')\phi(x,\theta')\rd \theta' \rd \theta\\
&=\oint_{\mathcal{S}^{n-1}}\mu_s(x)\phi(x,\theta')\oint_{\mathcal{S}^{n-1}} f(\theta,\theta')\rd \theta'\\
&=\oint_{\mathcal{S}^{n-1}}\mu_s(x)\phi(x,\theta)\rd \theta.
\end{aligned}
\end{equation*}
Integrating \eqref{eq:15} with respect to $\theta$ gives
\begin{align}
  \label{eq:16}
h^i-h^*+\oint_{\mathcal{S}^{n-1}}\theta\cdot \nabla(\phi^i(x,\theta)-\phi^*(x,\theta))\rd \theta=0.
\end{align}
Integrating \eqref{eq:16} with respect to $x$ over $\Omega$ gives
\begin{equation}
  \label{eq:25}
\|h^i\|_1-\norm{h^*}_1+\int_{\Omega}\oint_{\mathcal{S}^{n-1}}\theta\cdot \nabla(\phi^i(x,\theta)-\phi^*(x,\theta))\rd \theta \rd x=0.
\end{equation}
Applying Green's formula to the second term of \eqref{eq:25}, we have
\begin{equation*}
\begin{aligned}
  &\int_{\Omega}\oint_{\mathcal{S}^{n-1}}\theta\cdot \nabla(\phi^i(x,\theta)-\phi^*(x,\theta))\rd \theta \rd x\\
=&\oint_{\mathcal{S}^{n-1}}\int_{\Omega}\theta\cdot \nabla(\phi^i(x,\theta)-\phi^*(x,\theta))\rd \theta \rd x\\
=&\oint_{\mathcal{S}^{n-1}}\int_{\Gamma}(\theta\cdot\nu)(\phi^i(x,\theta)-\phi^*(x,\theta)) \rd x \rd \theta\\
=&\norm{\phi^i(x,\theta)-\phi^*(x,\theta)}_{L^1(\Gamma_+,|\theta\cdot\nu|)}\\
\geq & 0.
\end{aligned}
\end{equation*}
Apparently, we have
\begin{equation*}
  \|h^i\|_1\leq \norm{h^*}_1.
\end{equation*}
Naturally, when $i$ is sufficiently large, we have
\begin{equation}
\label{eq:26}
\begin{aligned}
  \|h^*\|_1-\|h^i\|_1&=\int_{\Omega}(h^*(x)-h^i(x))\rd x\\
&=\int_{\Omega^0_+\cup \Omega_-^{01}}(h^*(x)-h^i(x))\rd x-\int_{\Omega_-^{02}}(h^i(x)-h^*(x))\rd x\\
&=\|h^*-h^i\|_{L^1(\Omega^0_+\cup \Omega_-^{01})}-\|h^*-h^i\|_{L^1(\Omega_-^{02})}.
\end{aligned}
\end{equation}
Hence, when $i$ tends to infinity, we have
\begin{equation*}
  0\leq \lim_{i\rightarrow \infty}\|h^*-h^i\|_{L^1(\Omega_+\cup \Omega_-^{01})}-\|h^*-h^i\|_{L^1(\Omega_-^{02})}=0-\lim_{i\rightarrow \infty}\|h^*-h^i\|_{L^1(\Omega_-^{02})}\leq 0.
\end{equation*}
Thus
\begin{equation*}
  \lim_{i\rightarrow \infty}\|h^*-h^i\|_{L^1(\Omega_-^{02})}=0.
\end{equation*}
This completes the proof of (iv).

Therefore,
\begin{equation*}
  \label{eq:27}
  \lim_{i\rightarrow \infty}\|h^*-h^i\|_1=0.
\end{equation*}

\end{proof}

Given scattering coefficient, by Algorithm \ref{alg:1}, we can get monotonically increasing sequence $\mu_a^i(x)$ and monotonically decreasing sequence $\phi^i(x,\theta)$, which converge pointwisely to $\bar{\mu_a}\leq \mu^*$ and $\bar{\phi}\geq \phi^*$ respectively. Sequence $h^i(x)$ converges to exact function $h^*(x)$ in the $L^1$-norm.

\subsection{Reconstruction of $\mu_a$ and $\mu_s$ simultaneously}
\label{sec:3.2}

In practice, the scattering coefficient is also unknown, so the improved fixed-point iteration method is not applicable for the two unknown coefficients case. It is necessary to establish a more general method to recover two coefficients simultaneously. From now on, we follow optimization approach to estimate $\mu_a$ and $\mu_s$. First, we define error functional
\begin{equation}
  \label{eq:28}
  \mathcal{F}(\mu_a,\mu_s):=\sum_{m=0}^{M-1}\frac{1}{2}\norm{\log (h_m(\ \cdot\ ;\mu_a,\mu_s))-\log(h^*_m)}_2^2
\end{equation}
to measure the distance between measurement data $h_m^*$ and estimated data $h_m$, where $h_m(x;\mu_a,\mu_s)$ equals $\mu_a(\bm{A}\phi_m)(x;\mu_a,\mu_s)$ for estimated $\mu_a$ and $\mu_s$. Then the reconstruction of QPAT can be reformulated as
\begin{equation}
  \label{eq:modeling problem}
  \min_{\mu_a,\mu_s}\mathcal{F}(\mu_a,\mu_s).
\end{equation}

\begin{remark}
We replace $\norm{h_m(\ \cdot\ ;\mu_a,\mu_s)-h_m^*}_2^2/2$ by $\norm{\log (h_m(\ \cdot\ ;\mu_a,\mu_s))-\log(h_m^*)}_2^2/2$ as data-fidelity term is due to that the latter boosts the convergence of minimization method according to \cite{Tarvainen2012}. We discuss its advantage and properties in Appendix C.
\end{remark}

\begin{remark}
  According to inverse problem theory, regularization is useful for handling ill-conditioned problem. Usually, in image science, ill-conditionedness results in the edges of image blur and noise amplifying. Even though QPAT is a typical nonlinear problem, as pointed out in~\cite{Bal2009,Mamonov2014}, it is relatively stable with the help of multiple measurements. We focus on the minimization of data-fidelity term for the reason that multiple measurements alleviate the ill-posedness numerically. The algorithms on the minimization of objective function with regularization term, such as some a priori information, can be easily derived by our framework. We do not explore this in this paper.
\end{remark}

Using the definition of Fr\'echet derivative $\nabla \mathcal{F}$ of $\mathcal{F}$ in feasible direction $(h_{\mu_a},h_{\mu_s})\in \mathcal{D}_a\times \mathcal{D}_s$, we can write
\begin{equation}
  \label{eq:29}
  \mathcal{F}'(\mu_a,\mu_s)(h_{\mu_a},h_{\mu_s}):=\left<\nabla\mathcal{F}, (h_{\mu_a},h_{\mu_s})\right>_{L^2(\Omega)}.
\end{equation}
The solution of RTE with coefficients $\mu_a$ and $\mu_s$ can be regarded as a function with respect to $\mu_a$ and $\mu_s$, that is, $\phi=\phi(x,\theta;\mu_a,\mu_s)$. Fortunately, the directional derivative of $\phi(x,\theta;\mu_a,\mu_s)$ with respect to $\mu_a$ and $\mu_s$ in any feasible direction exists according to \cite{Haltmeier2015}. Furthermore, the directional derivative of data-fidelity with respect to $\mu_a$ and $\mu_s$ can be expressed analytically by $\mu_a$, $\phi(x,\theta;\mu_a,\mu_s)$ and its solution of adjoint RTE. This concludes in Proposition~\ref{prop:1}.

Similar to the proof of Proposition 3.3 in \cite{Haltmeier2015}, we can express the gradient of objective functional as follows.
\begin{prop}
\label{prop:1}
  For any pairs $(\mu_a,\mu_s)\in \mathcal{D}_a\times\mathcal{D}_s$ and feasible direction $(h_{\mu_a},h_{\mu_s})\in \mathcal{D}_a\times\mathcal{D}_s$, we have
  \begin{equation}
    \label{eq:30}
\begin{aligned}
 \phantom{=} &\mathcal{F}'(\mu_a,\mu_s)(h_{\mu_a},h_{\mu_s})\\
&=\sum_{m=0}^{M-1} \left<\frac{\log(\mu_a \bm{A}\phi_m)-\log(h^*_m)}{\mu_a}-\bm{A}(\phi_m\phi_m^*),h_{\mu_a}\right>+\sum_{m=0}^{M-1} \left<-\phi_m \phi_m^*+(\bm{K}\phi_m)\phi_m^*,h_{\mu_s}\right>,
\end{aligned}
  \end{equation}
where $\phi_m^*$ solves following adjoint RTE
\begin{equation}
  \label{eq:31}
\left\{
\begin{aligned}
  (-\theta\cdot \nabla_x+(\mu_a+\mu_s-\mu_s\bm{K}))\phi_m^*&=\bm{A}^*\left(\frac{\log(\mu_a \bm{A}\phi_m)-\log(h^*_m)}{\bm{A}\phi_m}\right),\\
\phi_m^*|_{\Gamma_+}&=0.
\end{aligned}
\right.
\end{equation}
Notice that $\bm{A}^*$ is the adjoint operator of $\bm{A}$ and
\begin{equation*}
  (\bm{A}^*f)(x,\theta)=f(x),\quad \forall f\in L^2(\Omega).
\end{equation*}

\end{prop}

\begin{proof}
Following the proof in \cite{Haltmeier2015} (Proposition 3.3), we have
  \begin{equation}
    \label{eq:32}
\begin{aligned}
 \phantom{=} &\mathcal{F}'(\mu_a,\mu_s)(h_{\mu_a},h_{\mu_s})\\
&=\sum_{m=0}^{M-1}\left<\log (\mu_a \bm{A}\phi_m(\ \cdot\ ;\mu_a,\mu_s))-\log (h^*_m), \frac{h_{\mu_a} \bm{A}\phi_m(\ \cdot\ ;\mu_a,\mu_s)+\mu_a \bm{A}\phi_m'(\ \cdot\ ;\mu_a,\mu_s)(h_{\mu_a},h_{\mu_s})}{\mu_a \bm{A}\phi_m(\ \cdot\ ;\mu_a,\mu_s)}\right>\\
&=\sum_{m=0}^{M-1}\left<\frac{\log(\mu_a \bm{A}\phi_m)-\log(h^*_m)}{\mu_a},h_{\mu_a}\right>+\sum_{m=0}^{M-1} \left<\bm{A}^*\left(\frac{\log(\mu_a \bm{A}\phi_m)-\log(h^*_m)}{\bm{A}\phi_m}\right),\phi_m'(h_{\mu_a},h_{\mu_s})\right>\\
&=\sum_{m=0}^{M-1} \left<\frac{\log(\mu_a \bm{A}\phi_m)-\log(h^*_m)}{\mu_a}-\bm{A}(\phi_m\phi_m^*),h_{\mu_a}\right>+\sum_{m=0}^{M-1} \left<-\phi_m \phi_m^*+(\bm{K}\phi_m)\phi_m^*,h_{\mu_s}\right>.
\end{aligned}
  \end{equation}
Notice that $\phi_m'(\ \cdot\ ;\mu_a,\mu_s)(h_{\mu_a},h_{\mu_s})$ is the directional derivative of $\phi_m(\ \cdot\ ;\mu_a,\mu_s)$ with respect to $(\mu_a,\mu_s)$ in direction $(h_{\mu_a},h_{\mu_s})$.

 It completes the proof.
\end{proof}

After deducing the gradient of error functional~\eqref{eq:28}, all we need to do is to find appropriate stepsize in the negative gradient direction to decrease the functional. Many optimization methods can achieve this goal, including steepest descent, Quasi-Newton, and so on. However, these gradient-based methods usually involve the linesearch process. A linesearch step needs to solve original RTE or adjoint RTE for several times. Since solving RTE dominates the computational cost, linesearch is computationally intensive. To mitigate the heavy computational cost, the well-known BB gradient method is applied to compute stepsize. It is derived from a two-point approximation to the scant equation underlying Quasi-Newton methods~\cite{Barzilai1988,Optimization2002}. And it is R-superlinearly convergent in the two-dimensional quadratic case~\cite{Barzilai1988}. Without loss of generality, we denote $\mu_a$ or $\mu_s$ by $\mu$ , then the update formula at $k$th step is
\begin{equation}
  \label{eq:032}
  \mu^{k+1}=\mu^k-\alpha_k\nabla\mathcal{F}_k.
\end{equation}
Generally, there are two choices about the stepsize $\alpha_k$:
\begin{equation*}
  \alpha_{k1}=\frac{s_k^\top y_k}{\norm{y_k}^2},
\end{equation*}
and
\begin{equation*}
  \alpha_{k2}=\frac{\norm{s_k}^2}{s_k^\top y_k},
\end{equation*}
where $s_k=\mu^k-\mu^{k-1}$ and $y_k=\nabla\mathcal{F}_k-\nabla\mathcal{F}_{k-1}$, and $\mathcal{F}_k$ is the iterative sequence of error functional $\mathcal{F}$, see~\eqref{eq:28}. BB algorithm is detailed in Algorithm \ref{alg:2}.

\begin{algorithm}[h]
\caption{BB method reconstruction\label{alg:2}}
\begin{algorithmic}[1]
 \REQUIRE Given initialization $\mu_a^0$, $\mu_s^0$, data $h_m^*$, boundary source $q_{b_m}\ (m=0,1,\dots,M)$, $\epsilon_1,\epsilon_2,\epsilon_3>0$, maximum number of iterations $N$, $\text{flag}_a=1$ and $\text{flag}_s=1$.
\FOR {$i=0,1,\dots, N$}
\STATE If $\text{flag}_a=0$ and $\text{flag}_s=0$, end up with $\mu_a=\mu_a^i$ and $\mu_s=\mu_s^i$;
\STATE Solve stationary RTEs~\eqref{eq:2} with absorption and scattering coefficients $\mu_a^i$ and $\mu_s^i$ respectively to obtain the solution $\phi_m^i$. Then let $h_m^i(x)=\mu_a^i(x)(\bm{A}\phi_m^i)(x)$;
\STATE  If $\mathcal{F}_i<\epsilon_1$, end up with $\mu_a=\mu_a^i$ and $\mu_s=\mu_s^i$;
\STATE Solve adjoint RTEs \eqref{eq:31};
\STATE Calculate the gradient of $\mathcal{F}_i$ : $\nabla_{\mu_a}\mathcal{F}_i$ and $\nabla_{\mu_s}\mathcal{F}_i$ with respect to $\mu_a$ and $\mu_s$. If $\norm{\nabla_{\mu_a}\mathcal{F}_i}\leq \epsilon_2$ and/or $\norm{\nabla_{\mu_s}\mathcal{F}_i}\leq\epsilon_3$, let $\text{flag}_a=0$ and/or $\text{flag}_s=0$;
\STATE If $i=0$ or $1$, updating $\mu_a$ and/or $\mu_s$ in negative gradient direction with small step size such that $\mathcal{F}_i$ decreases; otherwise, if $\text{flag}_a=1$ and/or $\text{flag}_s=1$, update $\mu_a$ and/or $\mu_s$ by BB stepsize;
\ENDFOR
\end{algorithmic}
\end{algorithm}

\section{Numerical simulations}
\label{sec:4}
The reconstructions of absorption and scattering coefficients are investigated with simulations in two cases: the reconstruction of $\mu_a$ given $\mu_s$ and the reconstruction of $\mu_a$ and $\mu_s$ simultaneously.
\subsection{Solver for RTE}
\label{sec:4.1}

We consider the numerical solver for RTE in 2-D. And it can be extended to 3-D with little effort. Finite element method combined with streamline diffusion modification is applied to solve stationary RTE~\cite{Yao2010,Haltmeier2015,Saratoon2013}, where $P_1$ Lagrangian elements in spatial and angular space are used. In this way, a large sparse linear system needs to be solved which is still difficult. By improving the algorithm proposed in \cite{Gao}, we can solve original RTE as well as adjoint RTE on 2D and 3D unstructured mesh by Discontinuous Galerkin (DG) method combined with multigrid method, which reduces the problem to a sparse block diagonal linear system.

We divide angular space into $P$ equal intervals, and the corresponding directions and angles are denoted by $\theta_0,\theta_1,\dots, \theta_{P-1}$ and $\beta_0,\beta_1,\dots,\beta_{P-1}$ with interval $\Delta \beta$. The spatial domain is discretized into unstructured triangular mesh. Suppose spatial domain is divided into N triangles and each triangle contains $n_d$ nodes, where $n_d=3$ for 2-D spatial domain. Lagrangian elements and piecewise linear DG elements are used to discretize RTE, that is, numerical discrete scheme of RTE is

where $L_k(\theta)$ is piecewise linear basis function in angular space which takes value of 1 in direction $\theta_k$ and 0 in other directions, $\varphi_{ij}$ is the linear basis function in spatial domain in the direction $\theta_k$ which takes value of 1 in the $j$th nodes of $i$th triangle and 0 in other nodes and triangles, and $\phi_{i,j,k}$ is the value of $\phi(x,\theta)$ in direction $\theta_k$ in $j$th node of $i$th element. Such spatial basis function can approximate discontinuous solution which is more suitable for actual situation, such as some edges of inclusions in object region. On account of scattering term in RTE, angular Gauss-Seidel iteration is applied in~\cite{Gao}. It iteratively solves RTE in every fixed direction in the form of
\begin{equation}
  \label{eq:34}
  \theta_k\cdot \nabla \phi_k+(\mu_a+\mu_s)\phi_k=\mu_s\sum_{k'=1}^P\omega_{kk'}\phi_{k'}+q_k, 1\leq k\leq P,
\end{equation}
where $\omega_{kk'}=\frac{\omega^0_{kk'}}{\sum_{k'}\omega^0_{kk'}}$ with $\omega^0_{kk'}=f(\theta_k,\theta_{k'})$. Obviously, equation \eqref{eq:34} can be solved by a lot of numerical methods, and we use DG method.

A multigrid scheme is applied to solve \eqref{eq:34}, and it is reduced to solve a sparse block diagonal system. Multiplying test function $\varphi_{ij}$ in both sides of \eqref{eq:34}, integrating it over $i$th triangle $\tau_i$ with respect to $x$, we have
\begin{equation}
  \label{eq:35}
  \begin{aligned}
  &-\int_{\tau_{i}}\phi_k(\theta_k\cdot \nabla\varphi_{ij})\rd x+ \int_{\Gamma_+}\phi_k\varphi_{ij}(\theta_k\cdot \nu)\rd S+\int_{\tau_i}(\mu_a+\mu_s-\mu_s\omega_{kk})\phi_k\varphi_{ij}\\
&=-\int_{\Gamma_-}\hat{\phi}_k\varphi_{ij}(\theta_k\cdot \nu)\rd S+\int_{\tau_i}(\mu_s\sum_{k'\neq k}\omega_{kk'}\phi_{k'})\varphi_{ij}\rd x+\int_{\tau_i}q_k\varphi_{ij}\rd x
\end{aligned}
\end{equation}
from Green's formula, where subscripts $i$ and $j$ of $\phi_{i,j,k}$ are omitted. In \eqref{eq:35}, $\hat{\phi}$ is the value of neighboring element in inflow direction, which is the product of upwind scheme. The specific explanation of why the value inflow direction is used is as follows.

For an intuitive explanation, we use rectangular mesh to illustrate in Figure~\ref{fig:rectmesh}. For rectangular mesh, in direction $\theta_k$ and node $x$, RTE can be discretized into
\begin{equation}
  \label{eq:36}
\cos \beta_k\frac{\partial \phi_k}{\partial x}+\sin \beta_k\frac{\partial \phi_k}{\partial y}+(\mu_a+\mu_s)\phi_k=\mu_s\sum_{k'=1}^M\omega_{kk'}\phi_{k'}+q_k,\, k=0,1,\dots,P-1,
\end{equation}
where $\phi_k:=\phi(x,\theta_k)$. Backward difference is used to approximate one order derivative. Then when $0\leq \beta_k<\frac{\pi}{2}$, its upwind scheme is
\begin{equation}
  \label{eq:37}
   (a+b+\mu_a+\mu_s-\mu_s\omega_{kk})\phi_{i,j,k}-(a\phi_{i-1,j,k}+b\phi_{i,j-1,k})-\mu_s\sum_{k'\neq k}\omega_{kk'}\phi_{i,j,k'}=q_{i,j,k},
\end{equation}
where
\begin{equation*}
  a=\frac{\cos\beta_k}{\Delta x}\geq 0,\, b=\frac{\sin \beta_k}{\Delta y}\geq 0.
\end{equation*}

\begin{figure}[H]
    \centering
\includegraphics[width=0.3\textwidth]{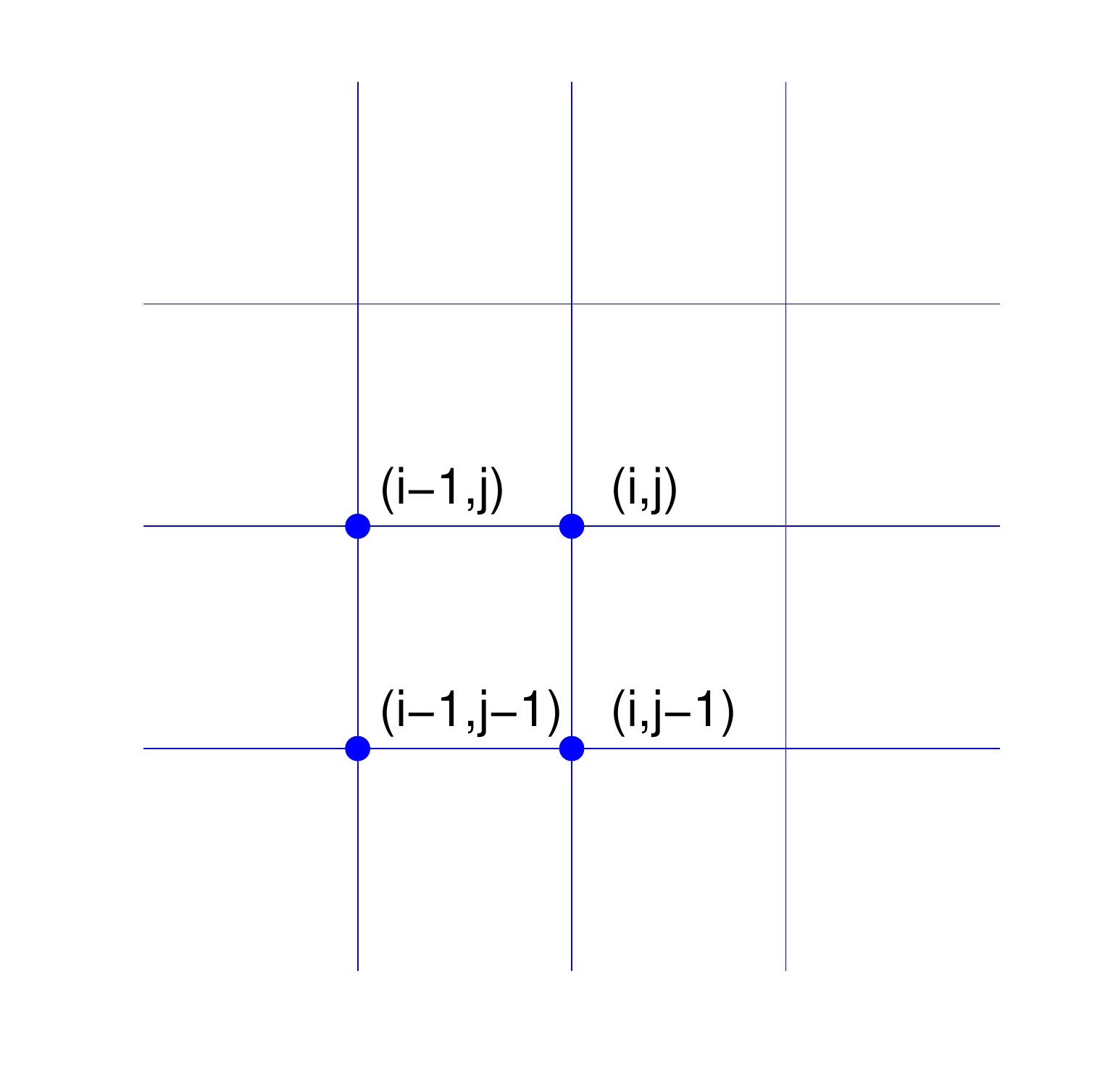}
 \caption{\label{fig:rectmesh}Rectangular mesh.}
\end{figure}

Obviously, the scheme converges because the equation \eqref{eq:37} is diagonally dominant, and Gauss-Seidel scheme will accelerate the convergence. In fact, from \eqref{eq:37}, $\phi_{i,j,k}$ is updated using
\begin{equation}
\label{eq:037}
  \phi_{i,j,k}=\frac{a\phi_{i-1,j,k}+b\phi_{i,j-1,k}+\mu_s\sum_{k'\neq k}\omega_{kk'}\phi_{i,j,k'}+q_{i,j,k}}{a+b+\mu_a+\mu_s-\mu_s\omega_{kk}}.
\end{equation}
Apparently, $\phi_{0,j,k}$ and $\phi_{i,0,k}$ are known as inflow boundary condition. Therefore, \eqref{eq:037} is an explicit scheme which updates $\phi_{i,j,k}$, $\phi_{i,1,k}$, $\phi_{2,j,k}$, $\phi_{i,2,k},\cdots\ (i,j,=1,2,\cdots)$ successively. We can see that \eqref{eq:37} and \eqref{eq:35} both use inflow information to update outflow information. This may be a kind of intuitive explanation. Indeed, the convergence of \eqref{eq:35} for vacuum boundary condition is proved in \cite{Gao2013}, see Appendix A. 

However, the above numerical scheme fails to solve the adjoint RTE with same updating order, that is to use inflow information to update outflow information. We provide a reverse updating order for elements to guarantee the convergence of solver for adjoint RTE. The solver for adjoint RTE is detailed as follows.

For adjoint RTE \eqref{eq:31}, the corresponding discrete scheme in direction $\theta_k$ for rectangular mesh is
\begin{equation}
  \label{eq:38}
  -\cos \beta_k\frac{\partial \phi_k}{\partial x}-\sin \beta_k\frac{\partial \phi_k}{\partial y}+(\mu_a+\mu_s)\phi_k=\mu_s\sum_{k'=1}^M\omega_{kk'}\phi_{k'}+q_k,\, k=0,1,\dots,P-1.
\end{equation}
Differently, forward difference is used to approximate the one order derivative, and when $0\leq \beta_k<\frac{\pi}{2}$, its upwind scheme is
\begin{equation}
  \label{eq:39}
  (a+b+\mu_a+\mu_s-\mu_s\omega_{kk})\phi_{i,j,k}-(a\phi_{i+1,j,k}+b\phi_{i,j+1,k})-\mu_s\sum_{k'\neq k}\omega_{kk'}\phi_{i,j,k'}=q_{i,j,k},
\end{equation}
where
\begin{equation*}
  a=\frac{\cos\beta_k}{\Delta x}\geq 0,\, b=\frac{\sin \beta_k}{\Delta y}\geq 0.
\end{equation*}
Then $\phi_{i,j,k}$ is updated by
\begin{equation}
  \label{eq:039}
  \phi_{i,j,k}=\frac{q_{i,j,k}+a\phi_{i+1,j,k}+b\phi_{i,j+1,k}+\mu_s\sum_{k'\neq k}\omega_{kk'}\phi_{i,j,k'}}{a+b+\mu_a+\mu_s-\mu_s\omega_{kk}}.
\end{equation}

The use of forward difference makes \eqref{eq:39} diagonally dominant again. Besides, outflow information is used to update inflow information in \eqref{eq:039}.
Heuristically, for unstructured mesh, we have
\begin{equation}
  \label{eq:40}
  \begin{aligned}
  &\int_{\tau_{i}}\phi_k(\theta_k\cdot \nabla\varphi_{ij})\rd x- \int_{\Gamma_-}\phi_k\varphi_{ij}(\theta_k\cdot \nu)\rd S+\int_{\tau_i}(\mu_a+\mu_s-\mu_s\omega_{kk})\phi_k\varphi_{ij}\\
&=\int_{\Gamma_+}\hat{\phi}_k\varphi_{ij}(\theta_k\cdot \nu)\rd S+\int_{\tau_i}(\mu_s\sum_{k'\neq k}\omega_{kk'}\phi_{k'})\varphi_{ij}\rd x+\int_{\tau_i}q_k\varphi_{ij}\rd x,
\end{aligned}
\end{equation}
where $\hat{\phi}$ is the value of neighboring element in outflow direction which is used to update inflow information. Using similar discussion, the convergence of \eqref{eq:40} is proved in Appendix B. Furthermore, we apply multigrid method to solve RTE to accelerate convergence.

From above discussion, no matter which mesh is used, the key of convergence is to update outflow information using inflow information for the original RTE \eqref{eq:2} and update the latter using the former for the adjoint RTE \eqref{eq:40}. The convergence in rectangular mesh is obvious. As for triangular mesh, the update scheme results from variation analysis, so the convergence proofs are obtained by discussing corresponding variations, see Appendix A and B.

Therefore, for original RTE \eqref{eq:2}, there are three layers of loops. We apply the multigrid scheme at the outermost loop. Since it involves two coordinate systems, i.e., angular-coordinate and spatial-coordinate, there are variable updating schemes in terms of the iteration order, such as the angle-prior and space-prior. The second-layer loop is about direction, that is, the corresponding spatial equation \eqref{eq:34} is solved in turn for each direction $\theta_k(k=0,1,\dots, P-1)$. In third-layer loop, for each element in the direction $\theta_k$, the value of $\phi_{i,j,k}$ is updated iteratively through solving a $3\times 3$ linear system~\eqref{eq:35}. Note that the updating order is in the following order: consider the each element successively from the boundary along the direction of $\theta_k$. We refer the interested reader to~\cite{Gao} for its algorithm details and to~\cite{Gao2013} for its theory.  On the contrast, for adjoint RTE, in third-layer loop, inspired by \eqref{eq:39} we propose to consider each element successively from the boundary along the direction of $-\theta_k$, which just reverse the updating order of original RTE. And linear system \eqref{eq:40} needs to be solved. The corresponding pseudo-code is presented in Algorithm \ref{alg:3}.
\begin{algorithm}[h]
\caption{Original and adjoint RTE solver\label{alg:3}}
\begin{algorithmic}[1]
\FOR {each loop of multigrid iteration}
\FOR {each direction $\theta_k(k=0,1,\dots,P-1)$}
\FOR {each element $\tau_i$}
\STATE Updating $\phi_{i,j,k}(j=0,1,2)$ by solving a $3\times 3$ linear system \eqref{eq:35} for original RTE or \eqref{eq:40} for adjoint RTE. For original RTE, the updating order is from boundary along $\theta_k$ through each elements to other side of region. For adjoint RTE, the updating order is reverse.
\ENDFOR
\ENDFOR
\ENDFOR
\end{algorithmic}
\end{algorithm}

\subsection{Numerical results}
\label{sec:4.3}

In this subsection, we present some numerical results to demonstrate the numerical performance of improved fixed-point iteration and BB method. For the sake of simplicity, only 2D reconstruction is investigated. The anisotropic factor $g$ equals 0.9. Applying RTE solver described in~\ref{sec:4.1}, we can get discrete energy density $\bm{H}$ of the same length as mesh. In order to explore the stability to noise, noisy data are generated by
\begin{equation}
  \label{eq:41}
  \tilde{\bm{H}}=\bm{H}(1+\epsilon\bm{N}),
\end{equation}
where $\bm{N}$ is a random vector that follows the normal distribution with mean 0 and variance 1. We consider 4 object regions:
\begin{enumerate}
\item Rectangle $\Omega_0=[-20,20]\times [-20,20]$ and four inclusions $\Omega_1=\{(x,y)\in\Omega_0:(x+10)^2+(y-10)^2=6^2\}$, $\Omega_2=\{(x,y)\in \Omega_2:(x-10)^2+(y-10)^2=4^2\}$, $\Omega_3=[-17,-5]\times[-17,-5]$, and $\Omega_4=[5,17]\times[-17,-5]$;
\item Circle $\Omega_0$ with center $(0,0)$ and radius 20;
\item Circle $\Omega_0$ with center $(0,0)$ and radius 20, and four inclusions $\Omega_1=[-12,-8]\times[-12,12]$, $\Omega_2=[-8,-2]\times[-12,12]$, $\Omega_3=[-2,12]\times[6,12]$, and $\Omega_4=[-2,12]\times[-12,6]$;
\item Circle $\Omega_0$ with center $(0,0)$ and radius 20, and three inclusions $\Omega_1=\{(x,y)\in \Omega_0:\frac{(y-3)^2}{9^2}+\frac{(x-7)^2}{6.2^2}=1\}$, $\Omega_2=[-14,-4]\times[-10,8]$, and $\Omega_3=\{(x,y)\in \Omega_0:\frac{(\frac{\sqrt{2}x}{2}+\frac{\sqrt{2}y}{2}-8.4)}{8^2}+\frac{(-\frac{\sqrt{2}x}{2}+\frac{\sqrt{2}y}{2}+8)^2}{5^2}=1\}$.
\end{enumerate}
 Their absorption and scattering coefficients respectively are
 \begin{enumerate}
 \item $\mu_a$ is 0.02 in $\Omega_1$ and $\Omega_4$ with background 0.01; $\mu_s$ is 3 in $\Omega_2$ and $\Omega_3$ with background 1,
 \item $\mu_a(x,y)=0.02+0.01\sin (\frac{\pi x}{8})$ and $\mu_s(x,y)=2+\sin (\frac{\pi y}{8})$,
 \item $\mu_a$ is 0.03 in $\Omega_1$, 0.02 in $\Omega_2$, 0.04 in $\Omega_3$, and 0.015 in $\Omega_4$ with background 0.01; $\mu_s$ is 2.5 in $\Omega_1$, 1.5 in $\Omega_2$, 3 in $\Omega_3$, and 2 in $\Omega_4$ with background 1,
 \item $\mu_a$ is 0.015 in $(\Omega_1\cup \Omega_2)\backslash(\Omega_1\cap\Omega_2)\cup \Omega_3$ and 0.03 in $\Omega_1\cap\Omega_2$ with background 0.01; $\mu_s$ is 3 in $\Omega_1\cup\Omega_3$ with background 1.
 \end{enumerate}
These regions are depicted in Figure \ref{fig:2}, where the optical coefficients of 1st, 3rd, and 4st templates are piecewise constant, and the second one is smooth in object region $\Omega_0$. These templates contain discontinuous borders as well as continuous borders, whose corners are both sharp and smooth. To avoid inverse crime, we generate data by solving RTE in finer unstructured mesh than inverse problem. Original data are generated on 21376, 16352, 17376, and 16576 unstructured mesh respectively, and corresponding inverse problem are solved on 9600, 7392, 7392, 7392 unstructured mesh. Four point sources are placed in (-20,0), (0,20), (20,0) and (0,-20). Iterative relative errors are defined by
\begin{equation*}
  \mathcal{\epsilon}_{\mu_a}=\frac{\norm{\mu_a-\mu_a^*}_2}{\norm{\mu_a^*}_2},\qquad \mathcal{\epsilon}_{\mu_s}=\frac{\norm{\mu_s-\mu_s^*}_2}{\norm{\mu_s^*}_2}.
\end{equation*}

\begin{figure}[H]
    \centering
\begin{minipage}{0.24\linewidth}
\includegraphics[width=\textwidth]{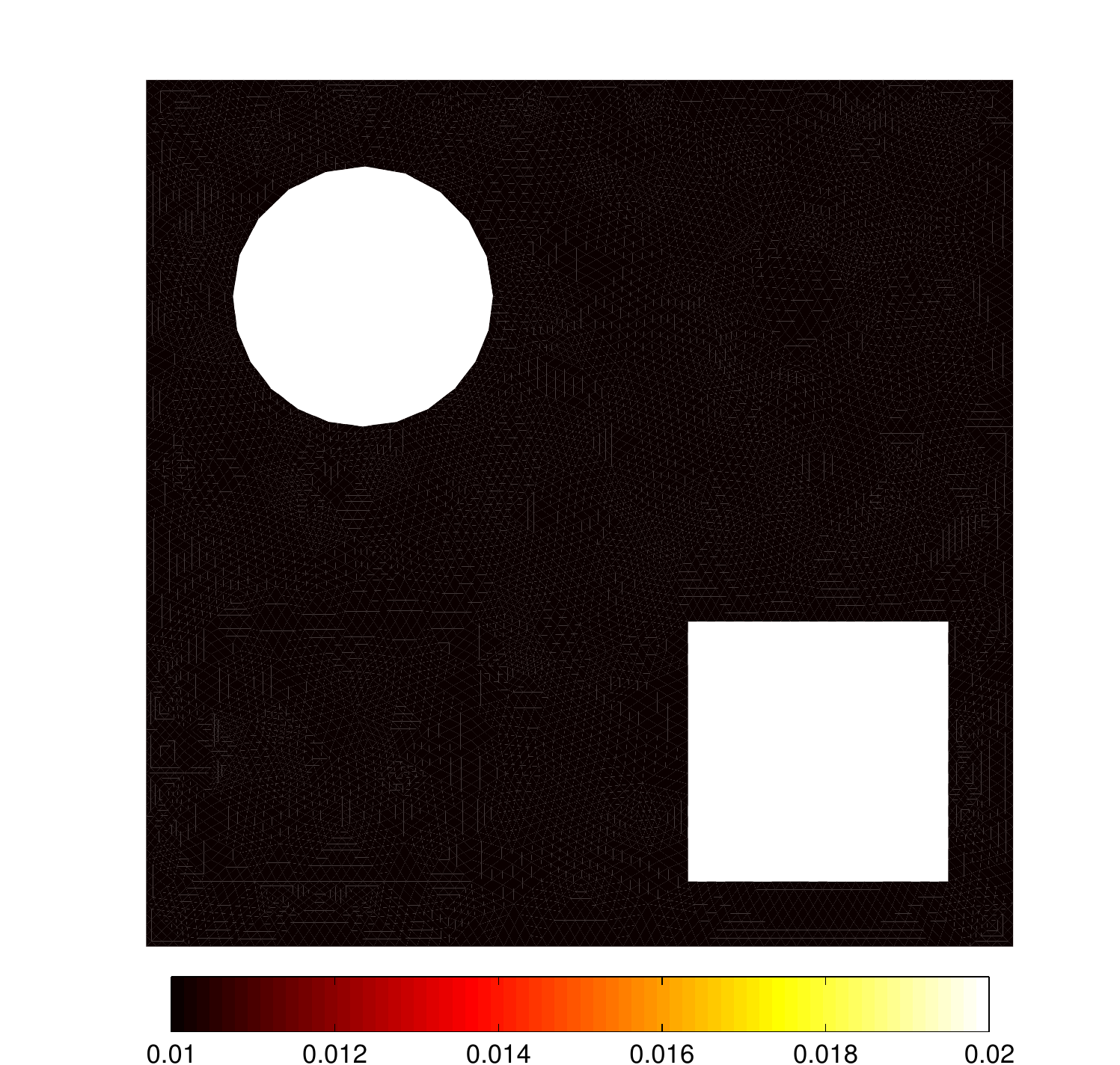}
\end{minipage}
\begin{minipage}{0.24\linewidth}
  \includegraphics[width=\textwidth]{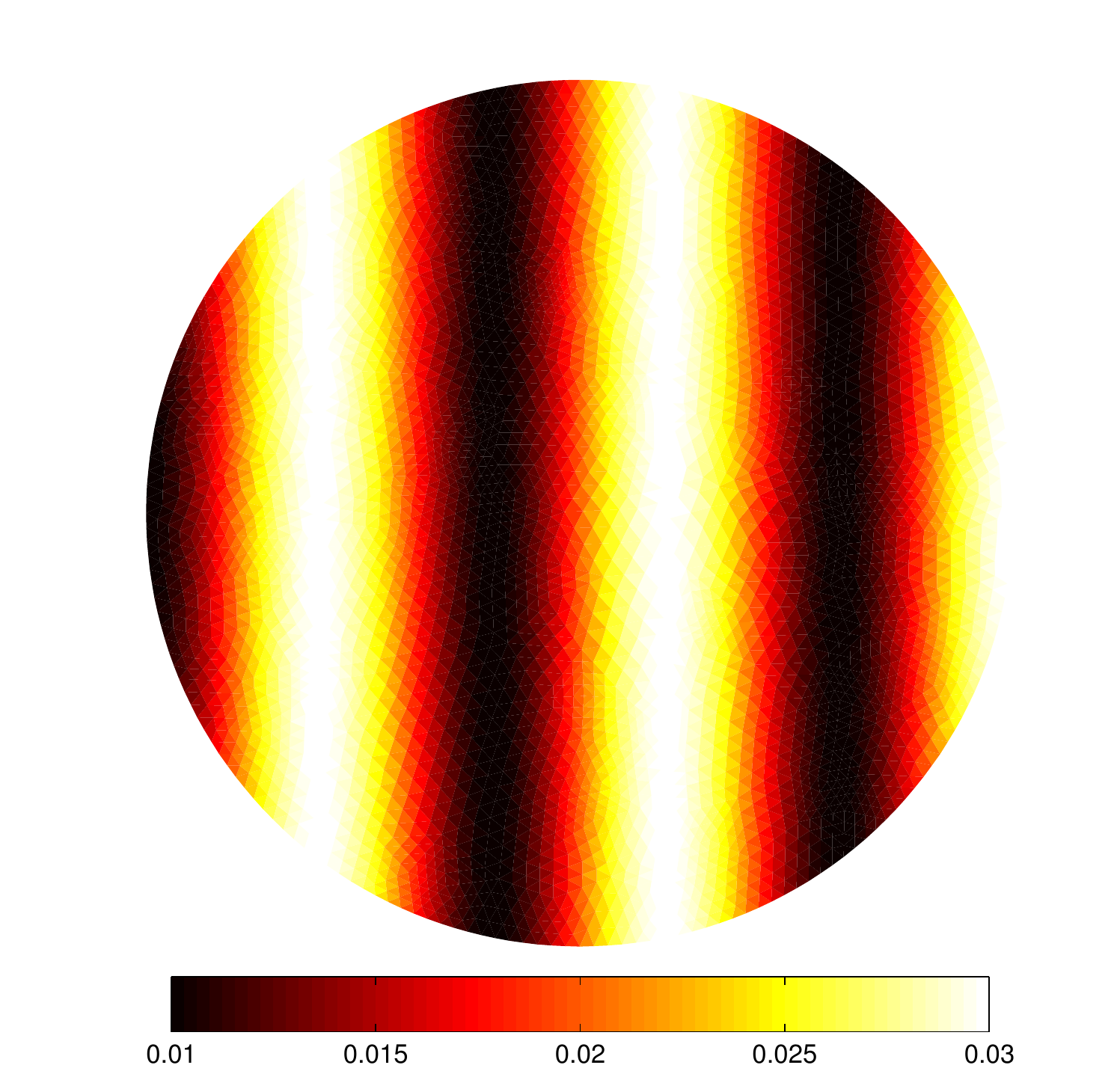}
\end{minipage}
\begin{minipage}{0.24\linewidth}
  \includegraphics[width=\textwidth]{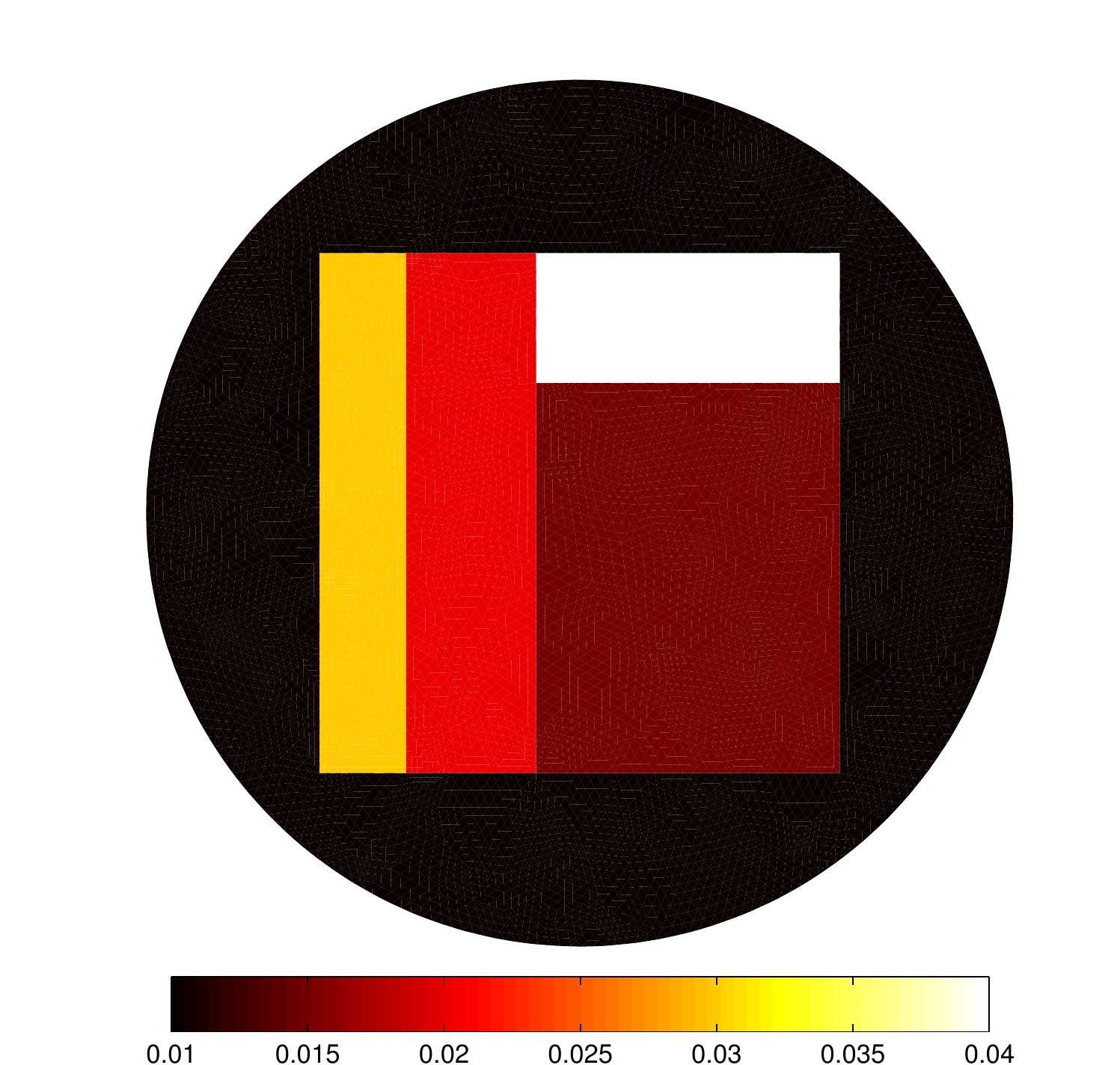}
\end{minipage}
\begin{minipage}{0.24\linewidth}
  \includegraphics[width=\textwidth]{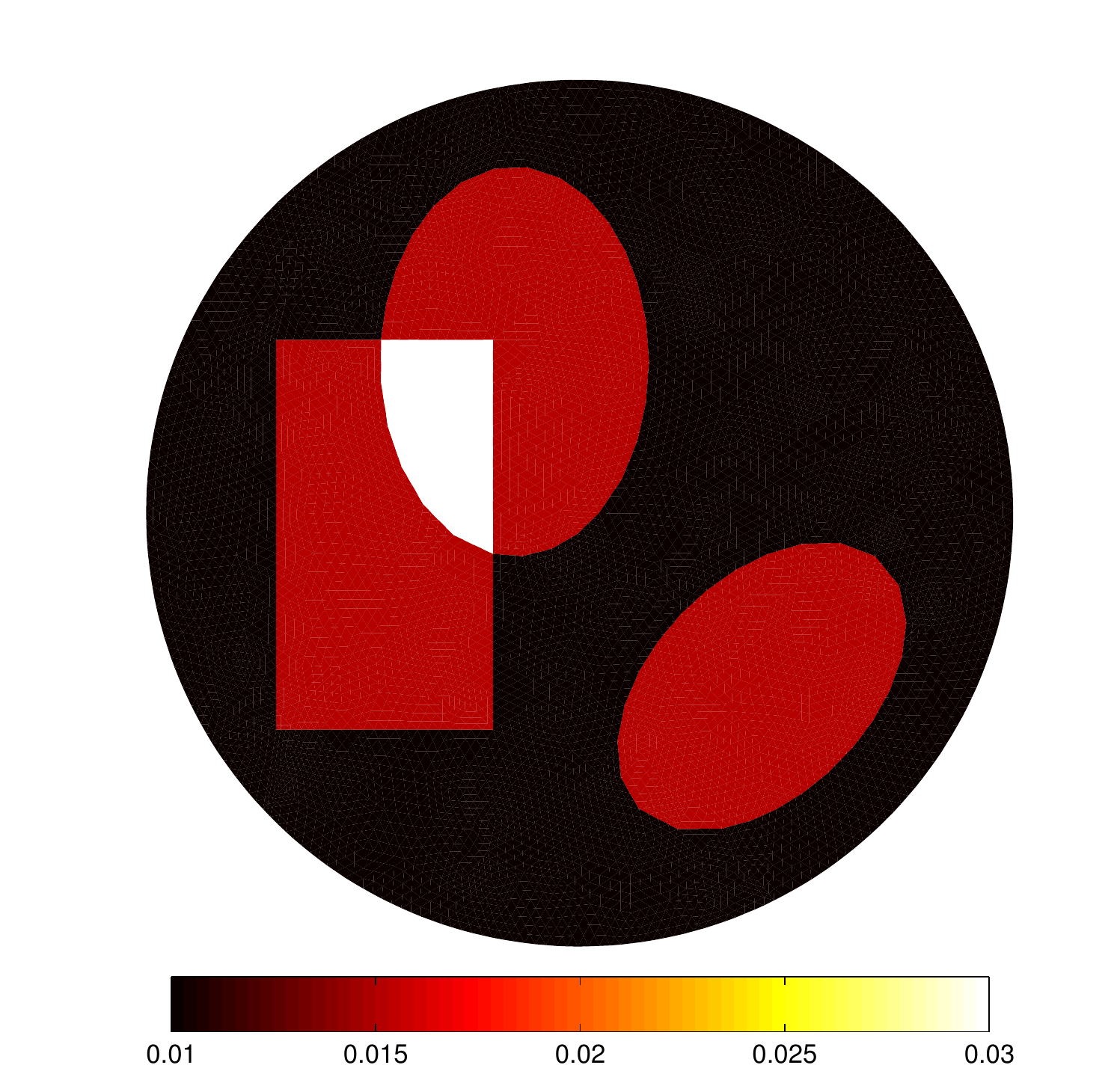}
\end{minipage}
\\
\begin{minipage}{0.24\linewidth}
  \includegraphics[width=\textwidth]{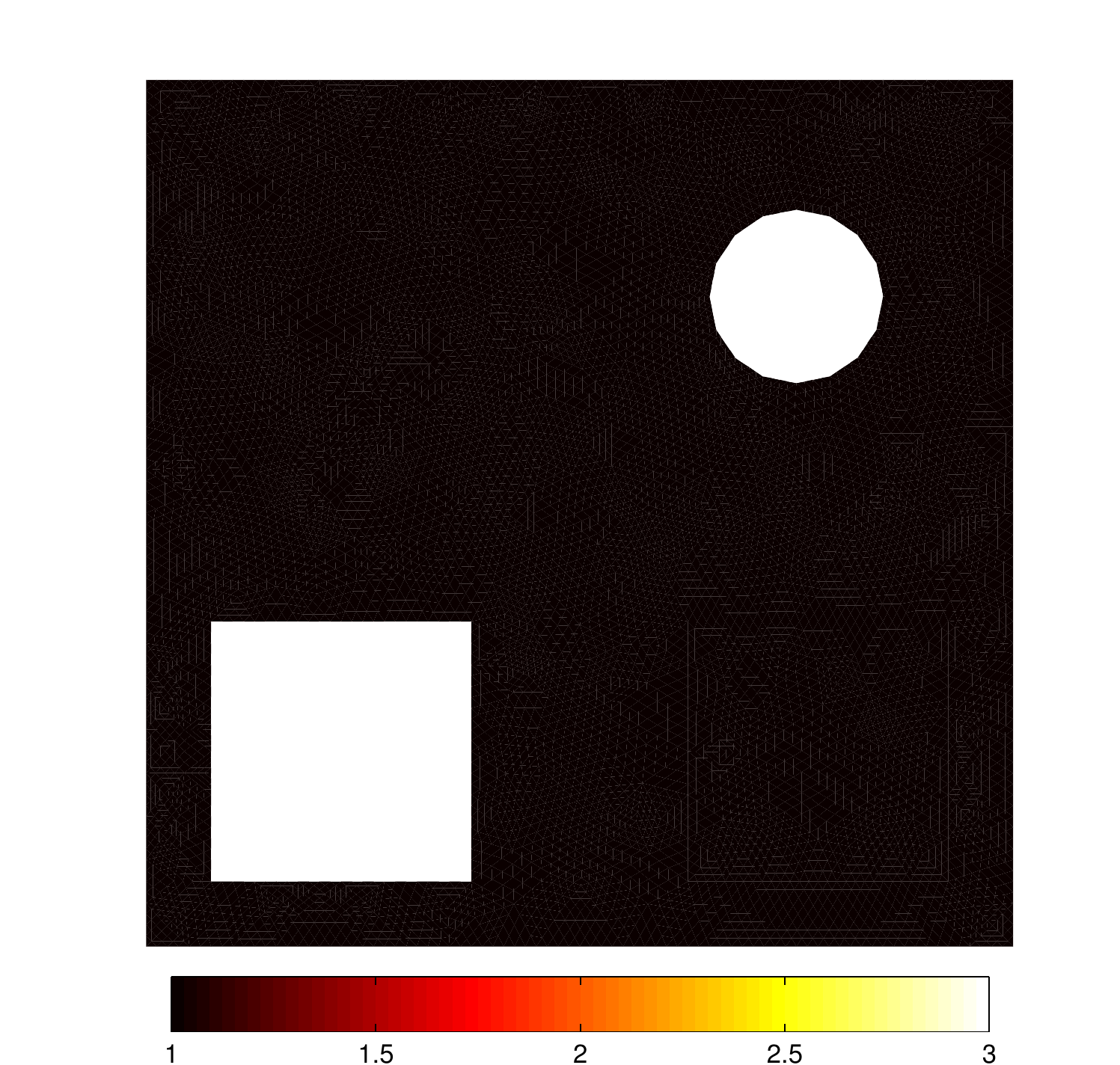}
\end{minipage}
\begin{minipage}{0.24\linewidth}
  \includegraphics[width=\textwidth]{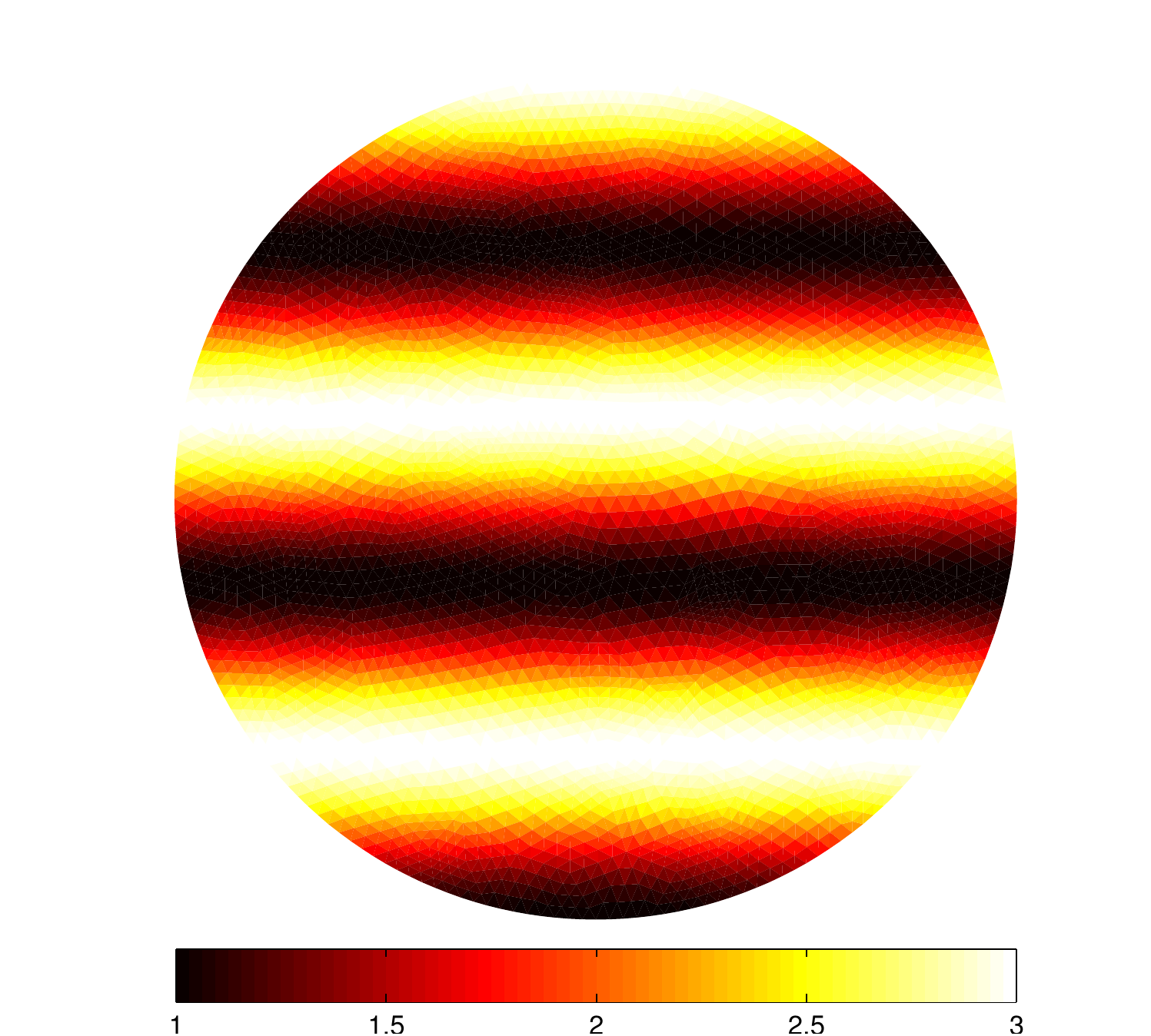}
\end{minipage}
\begin{minipage}{0.24\linewidth}
  \includegraphics[width=\textwidth]{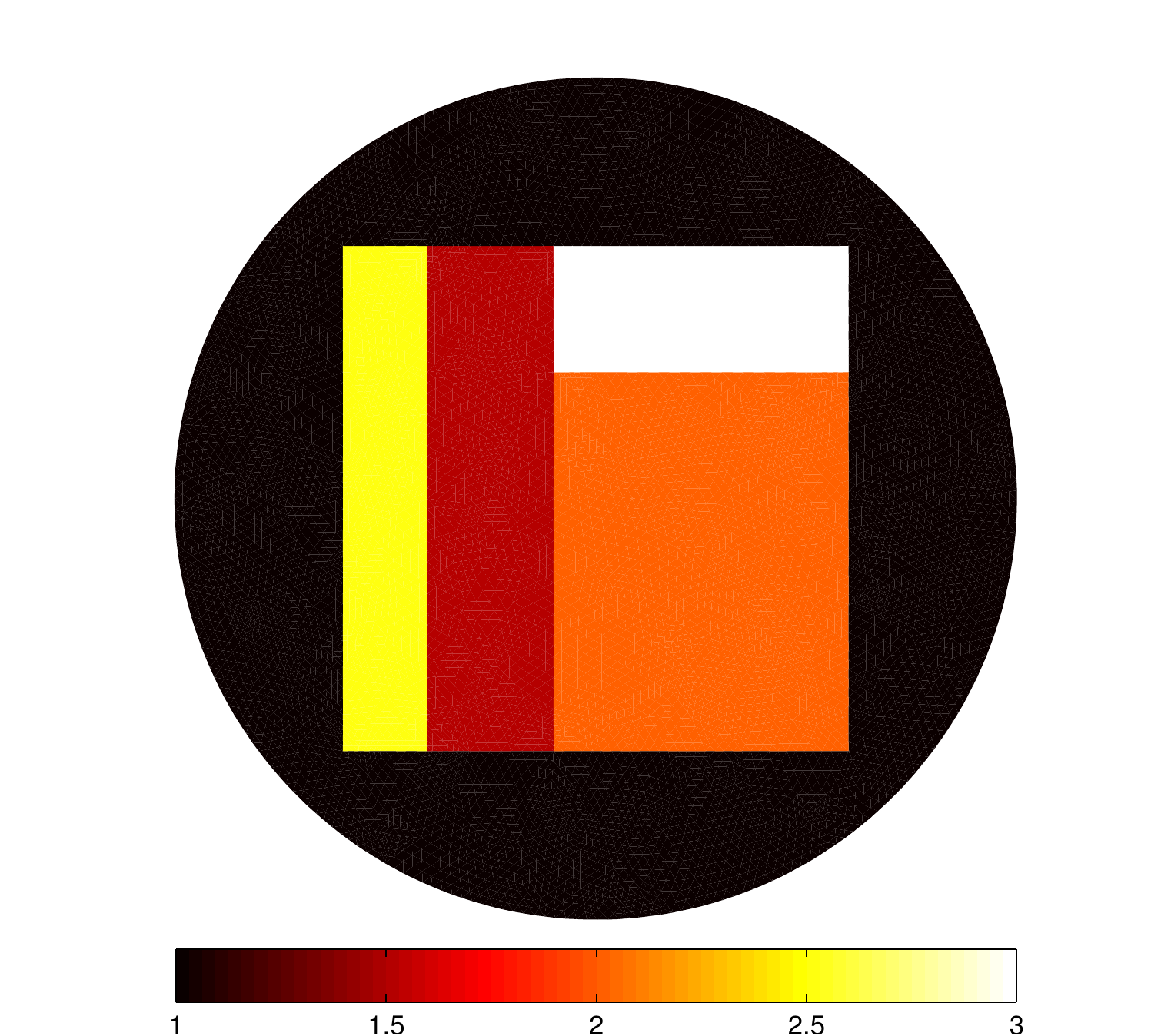}
\end{minipage}
\begin{minipage}{0.24\linewidth}
  \includegraphics[width=\textwidth]{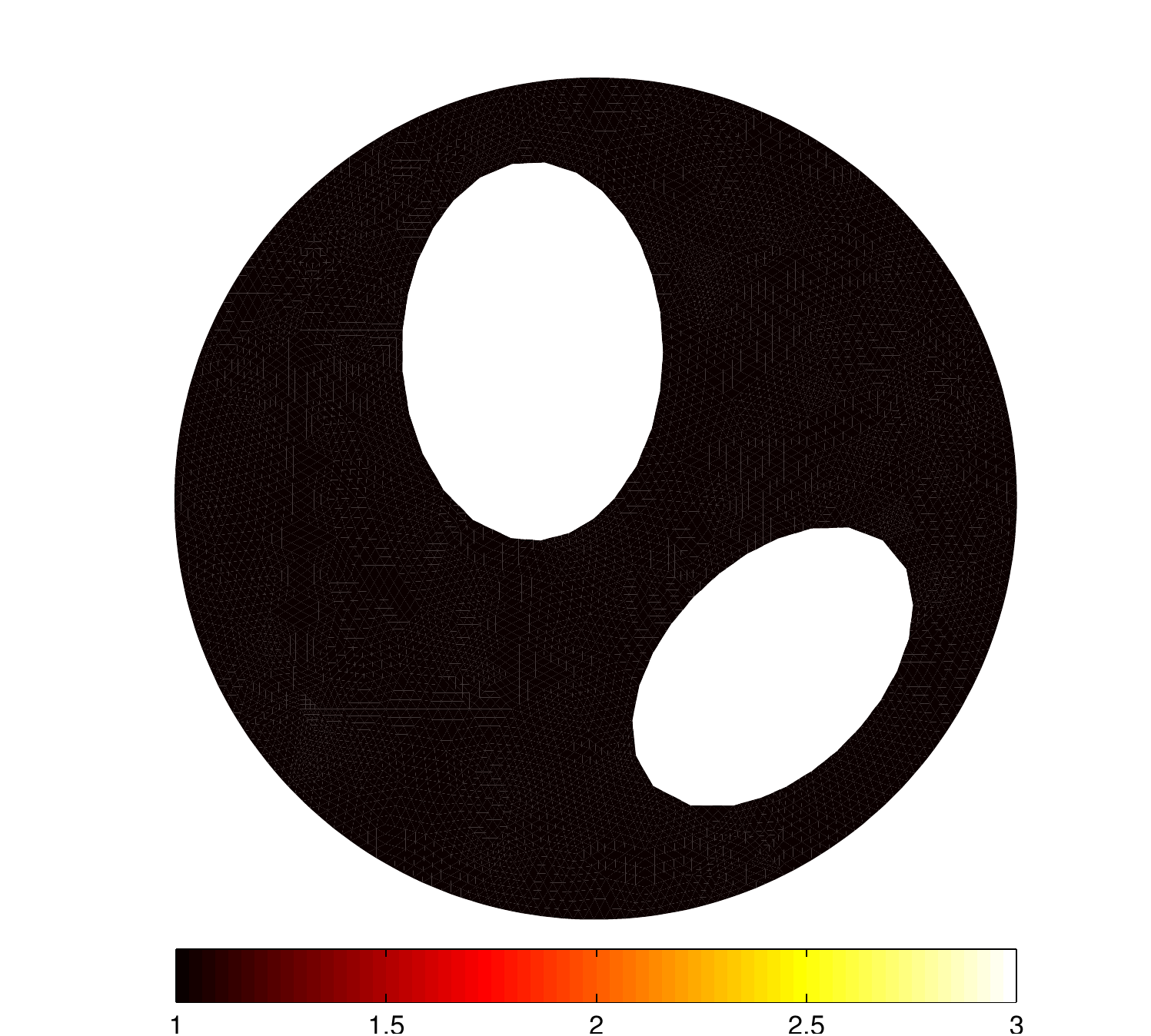}
\end{minipage}
 \caption{\label{fig:2}Original absorption and scattering coefficient maps of four templates. Top row: absorption coefficient map $\mu_a$. Bottom row: scattering coefficient map $\mu_s$.}
\end{figure}

\subsubsection{Reconstruction of $\mu_a$ given $\mu_s$ }
\label{sec:4.3.1}
Using one measurement with boundary source in (20,0), we apply improved fixed-point iterative method and BB method to retrieve absorption coefficients of templates given scattering coefficients. Our initial guesses are set to be the same as background. Figure~\ref{fig:3} shows results after 50 iterations for noiseless data. Then $5\%$ Gaussian noise (i.e. $\epsilon=5\%$) is added to data according to \eqref{eq:41}, and the final reconstructed images are showed in Figure~\ref{fig:3n}. Specific iterative relative errors are showed in Figure~\ref{fig:4}. From Figure~\ref{fig:3} and Figure~\ref{fig:3n}, it seems that both methods can retrieve absorption coefficient with almost the same accurate solutions in the case of noise-free and noisy measurements. However, from Figure \ref{fig:4}, it is obvious that improved fixed-point iteration converges more rapidly than BB method. Improved fixed-point iteration achieve the critical point only after a few iterations. To attain the same accuracy, the number of iteration of improved fixed-point iteration is about the half of the one of BB method. Owing to solving the adjoint RTEs in BB method, improved fixed-point iteration can achieve the accuracy as the BB method with about quarter computational time of the latter.

\begin{figure}[H]
    \centering
\begin{minipage}{0.24\linewidth}
\includegraphics[width=\textwidth]{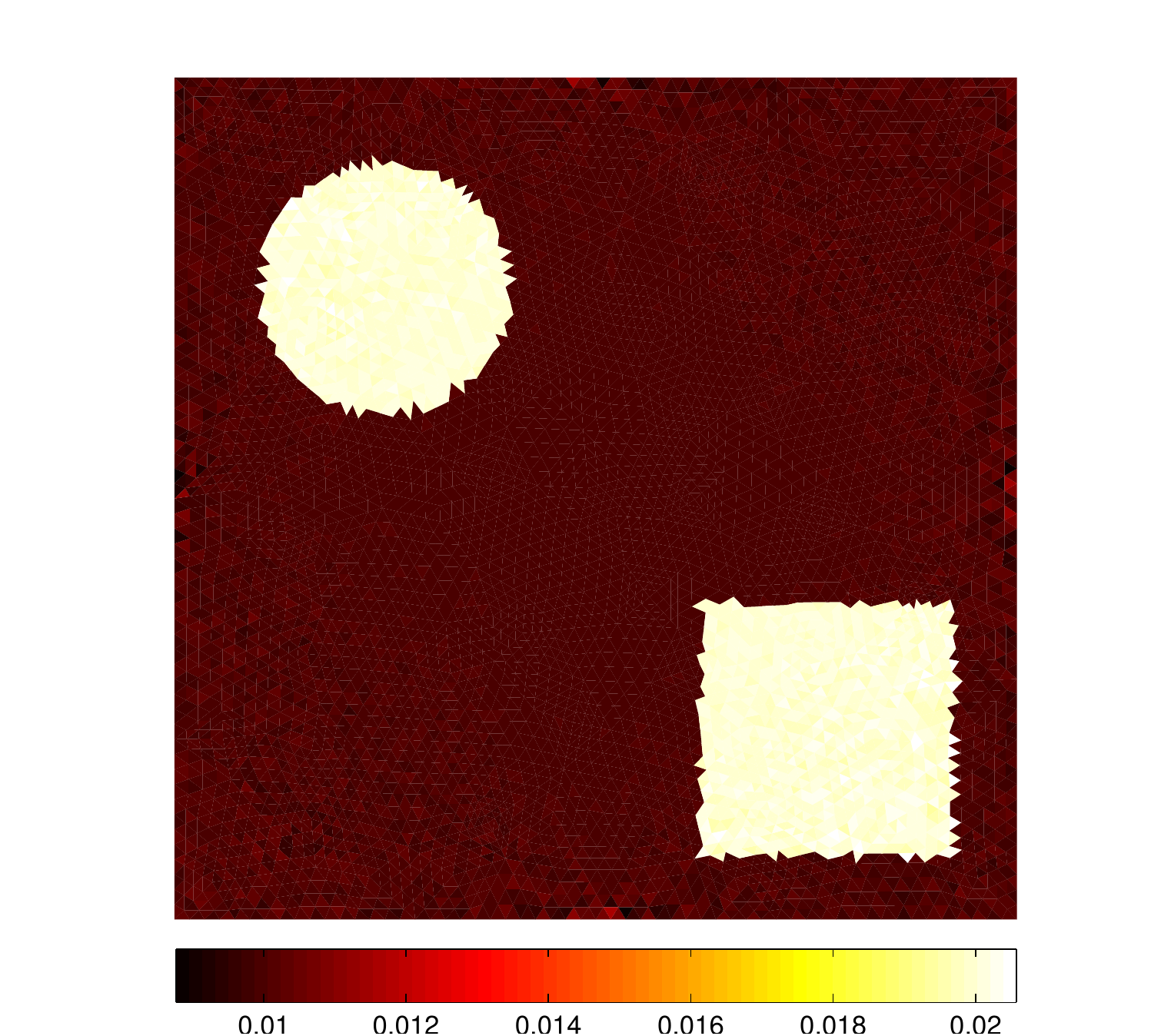}
\end{minipage}
\begin{minipage}{0.24\linewidth}
  \includegraphics[width=\textwidth]{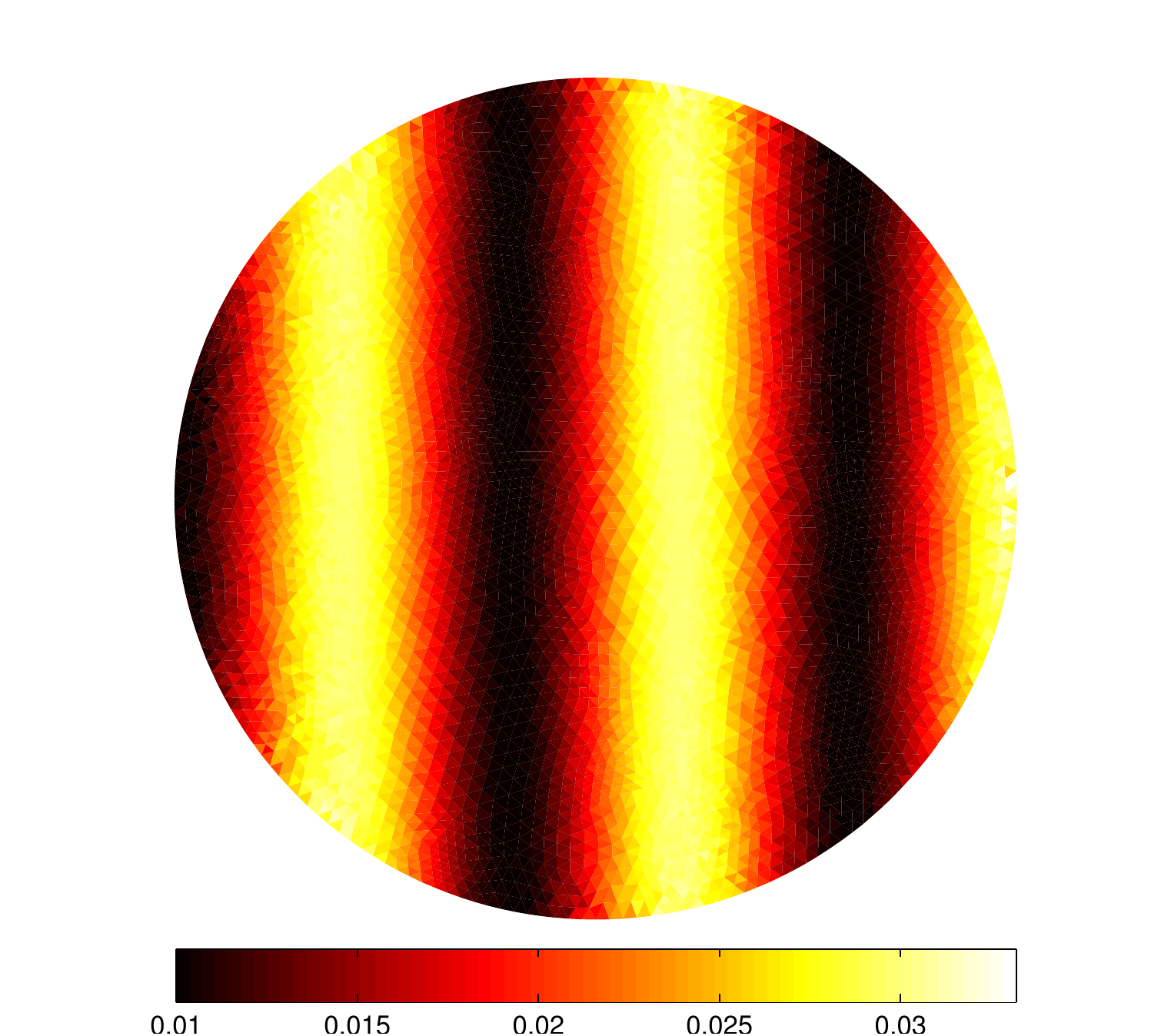}
\end{minipage}
\begin{minipage}{0.24\linewidth}
  \includegraphics[width=\textwidth]{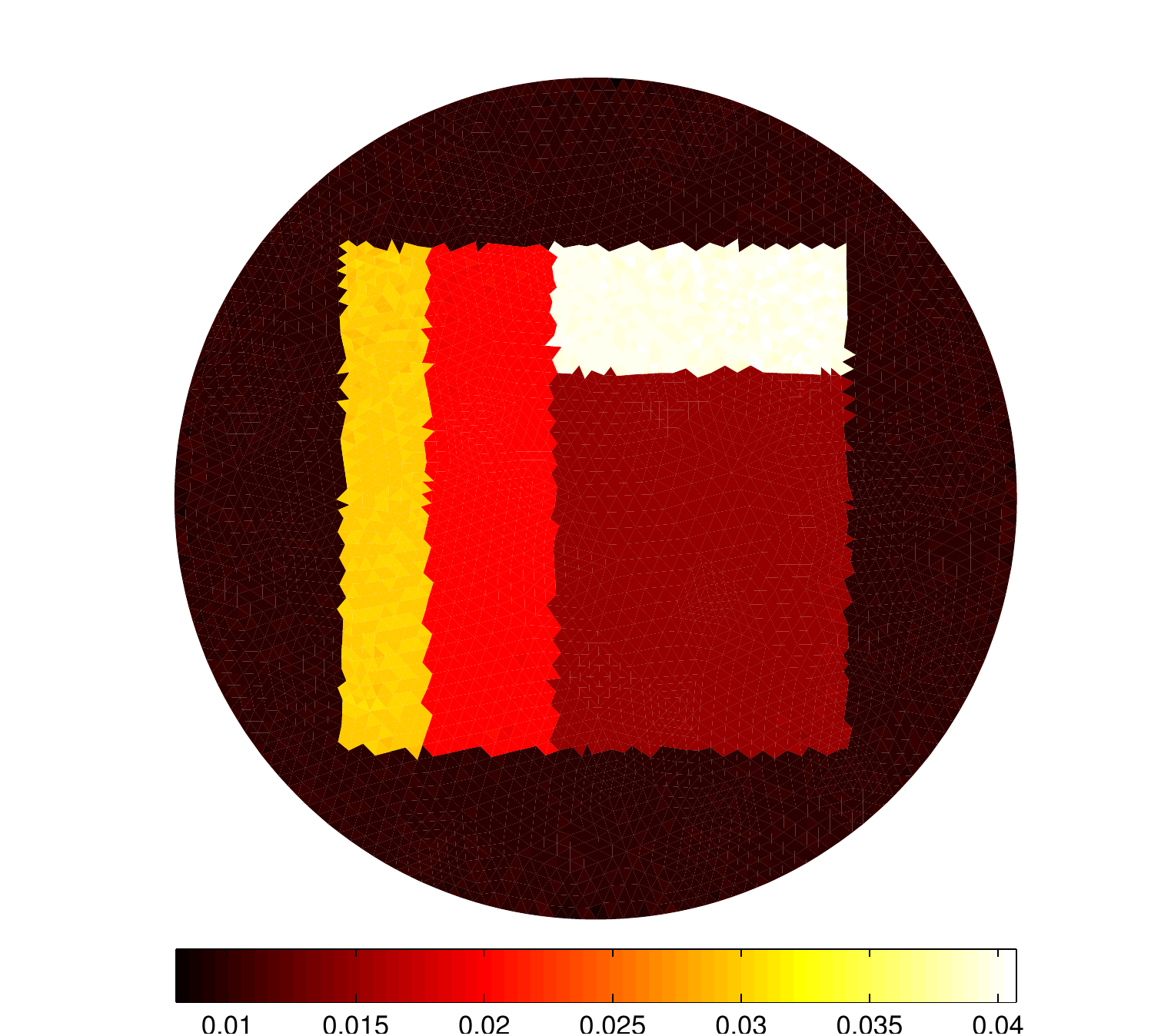}
\end{minipage}
\begin{minipage}{0.24\linewidth}
  \includegraphics[width=\textwidth]{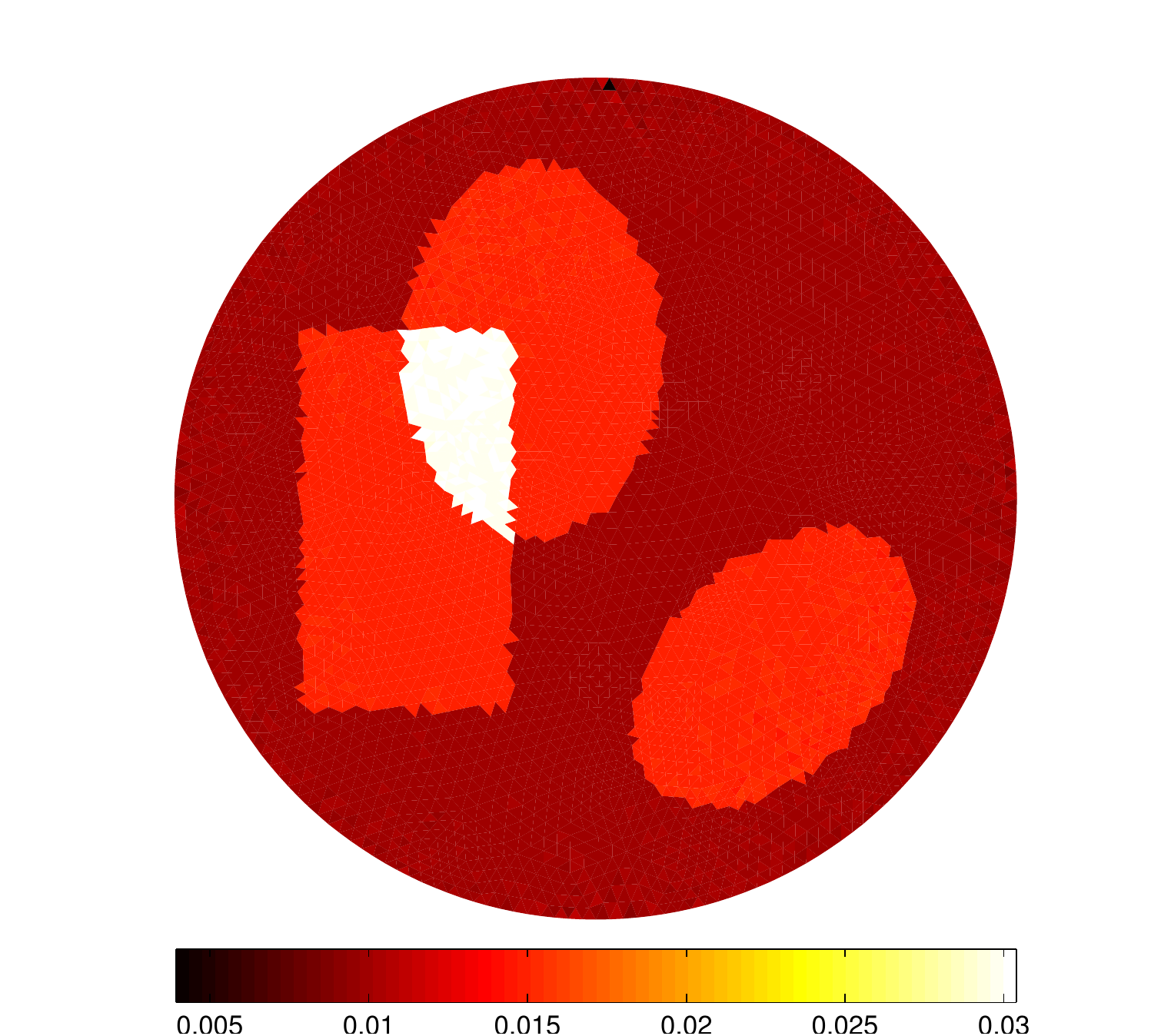}
\end{minipage}
\\
\begin{minipage}{0.24\linewidth}
  \includegraphics[width=\textwidth]{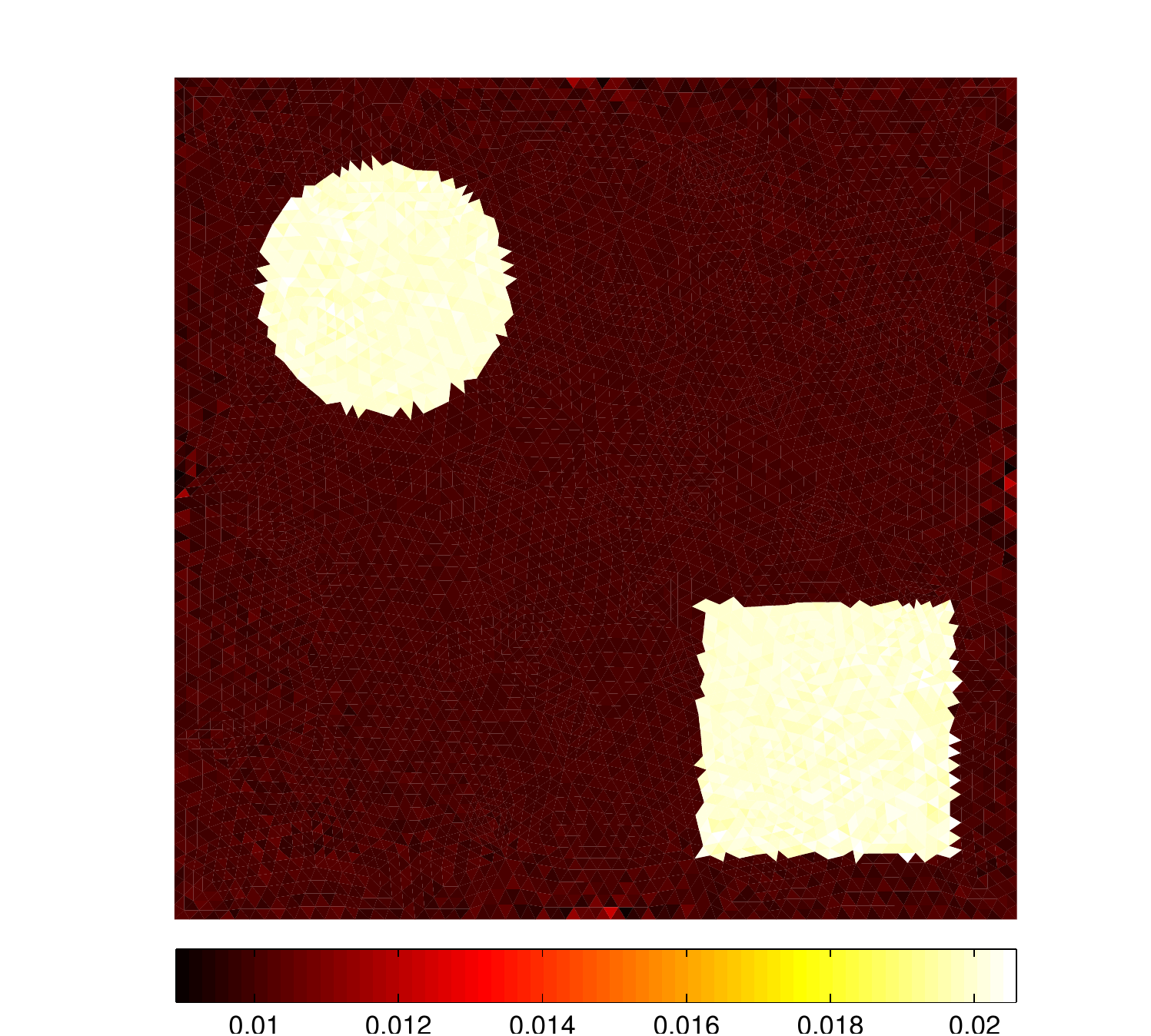}
\end{minipage}
\begin{minipage}{0.24\linewidth}
  \includegraphics[width=\textwidth]{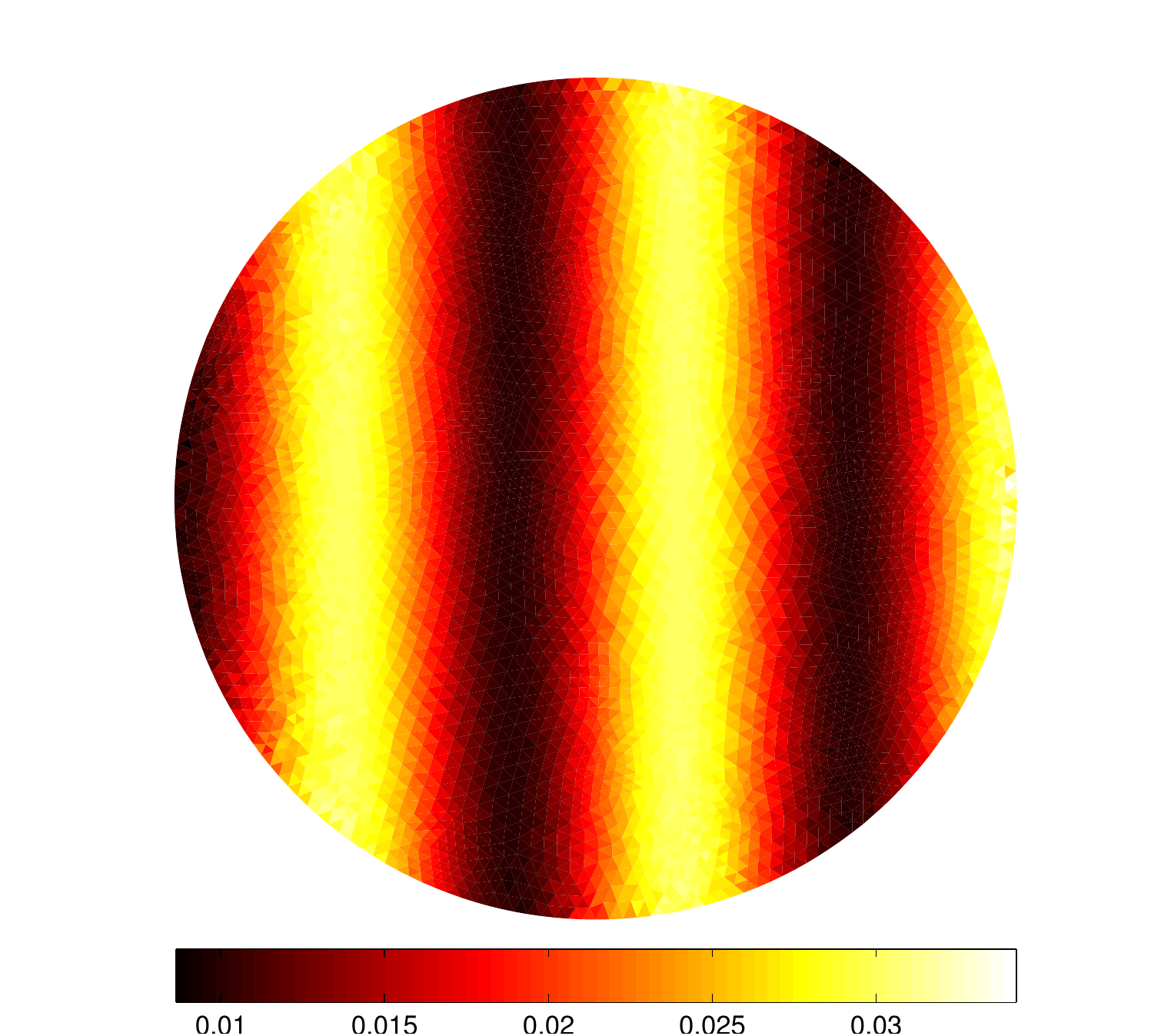}
\end{minipage}
\begin{minipage}{0.24\linewidth}
  \includegraphics[width=\textwidth]{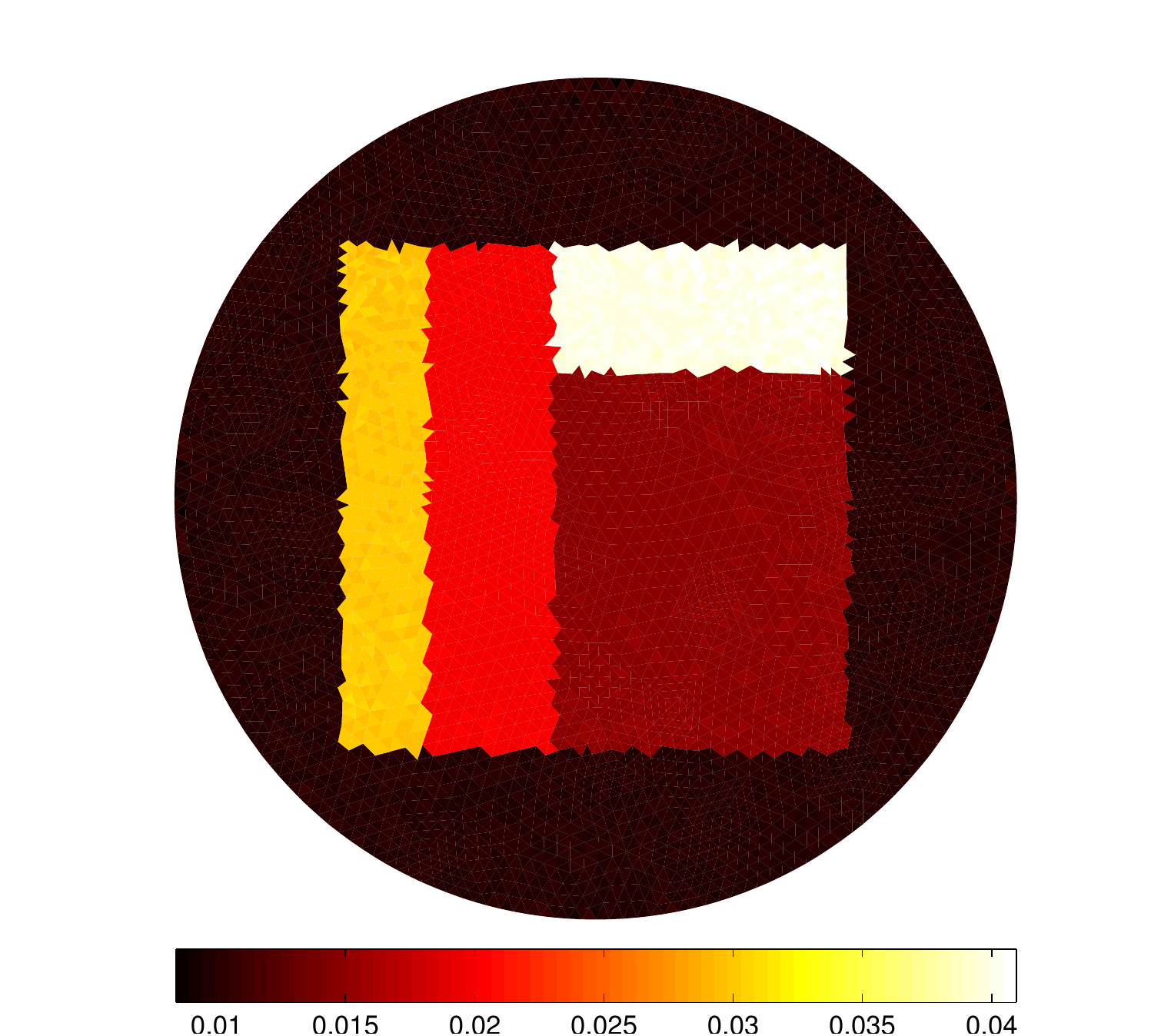}
\end{minipage}
\begin{minipage}{0.24\linewidth}
  \includegraphics[width=\textwidth]{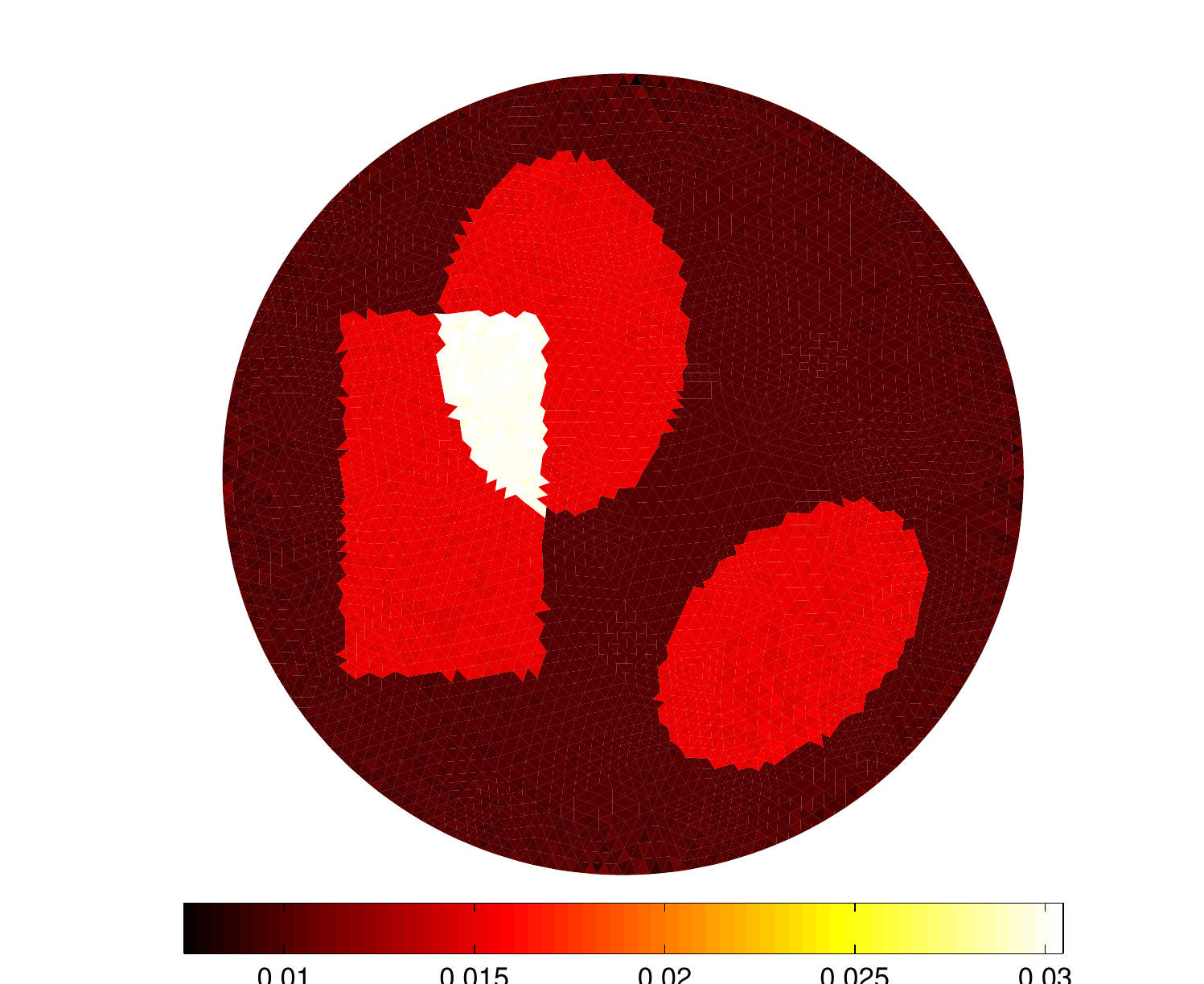}
\end{minipage}
 \caption{\label{fig:3}Given scattering coefficient, reconstructions of absorption coefficient from noiseless data. Top row: BB method. Bottom row: fixed-point iteration.}
\end{figure}

\begin{figure}[H]
    \centering
\begin{minipage}{0.24\linewidth}
\includegraphics[width=\textwidth]{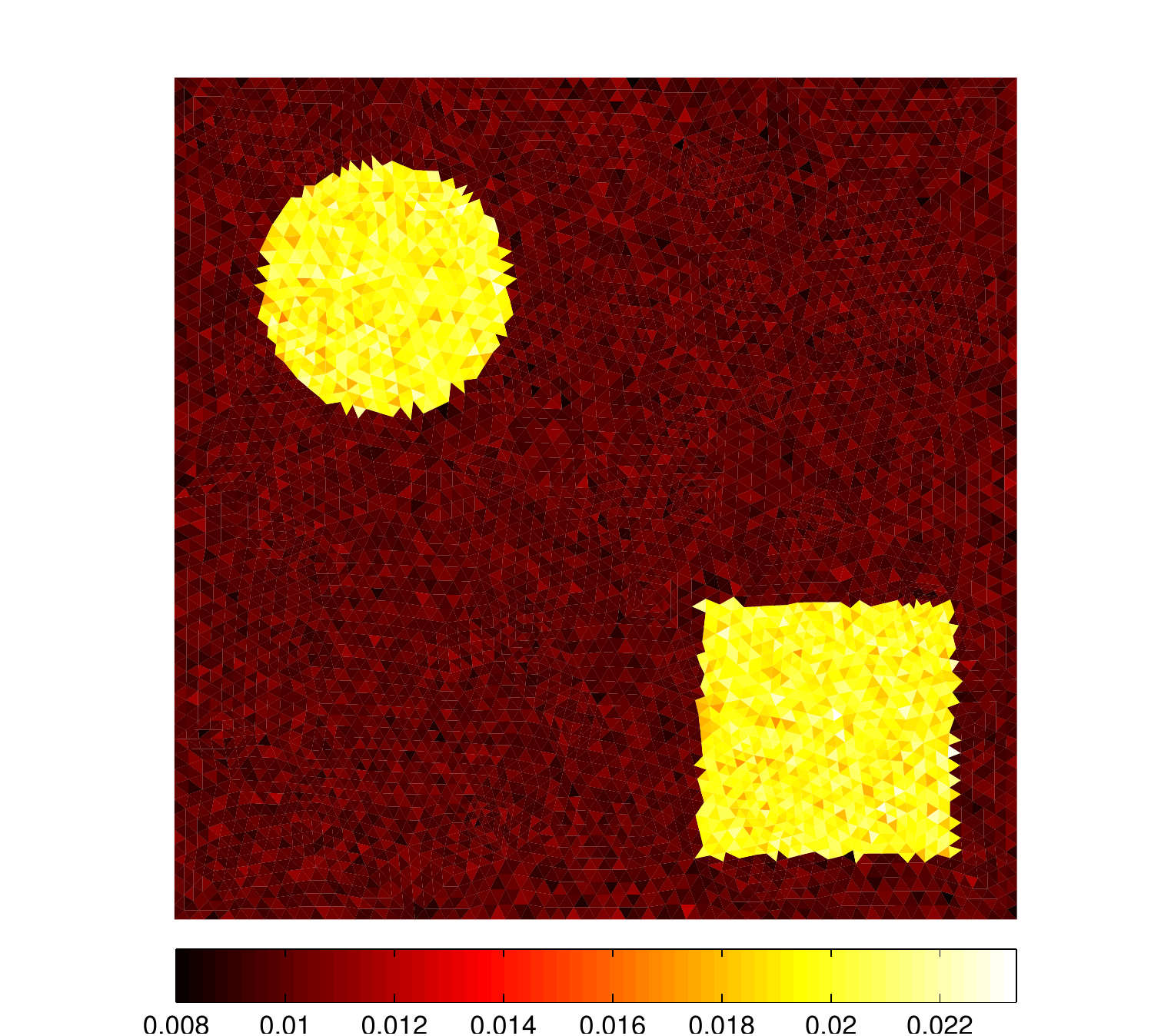}
\end{minipage}
\begin{minipage}{0.24\linewidth}
  \includegraphics[width=\textwidth]{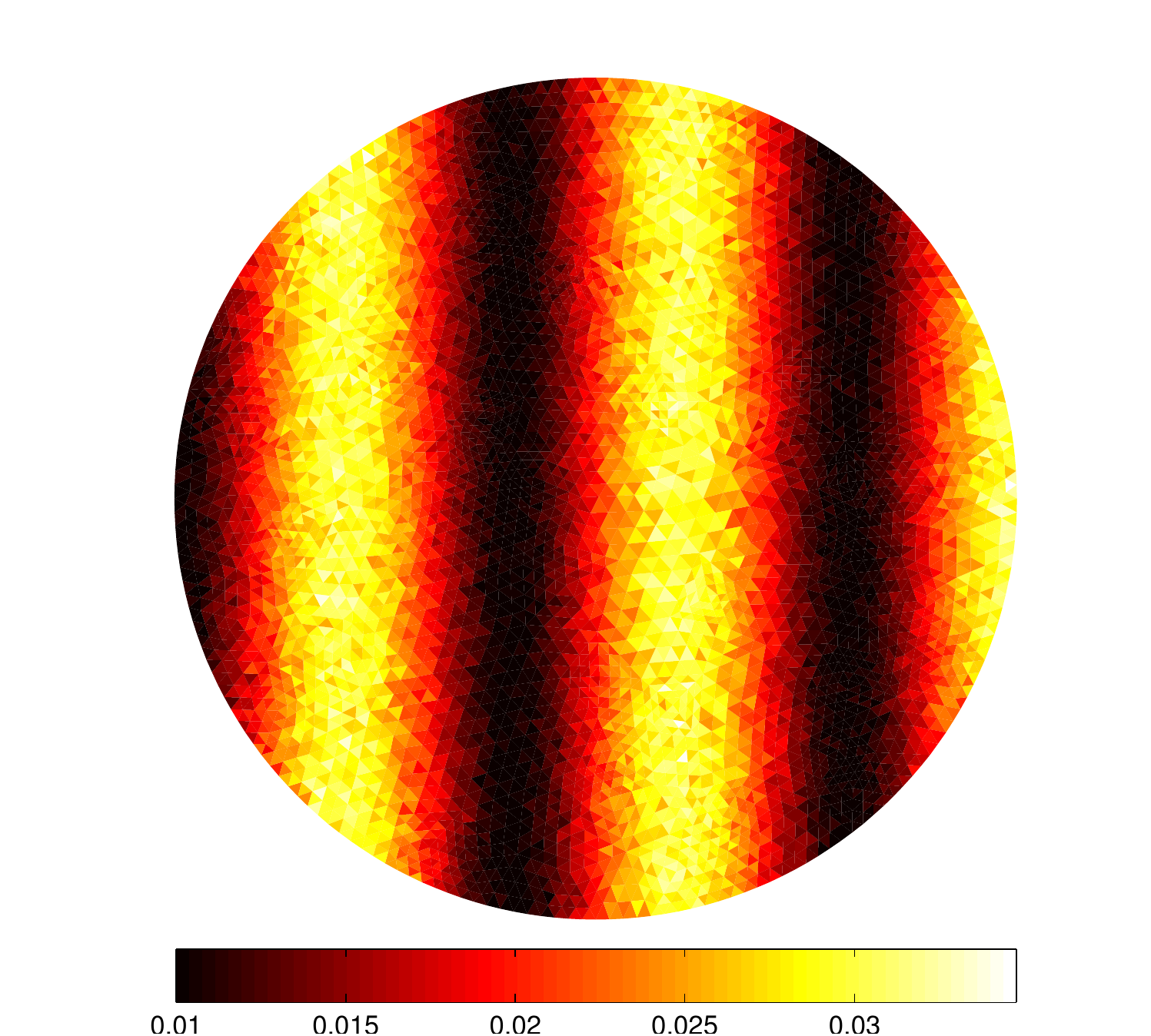}
\end{minipage}
\begin{minipage}{0.24\linewidth}
  \includegraphics[width=\textwidth]{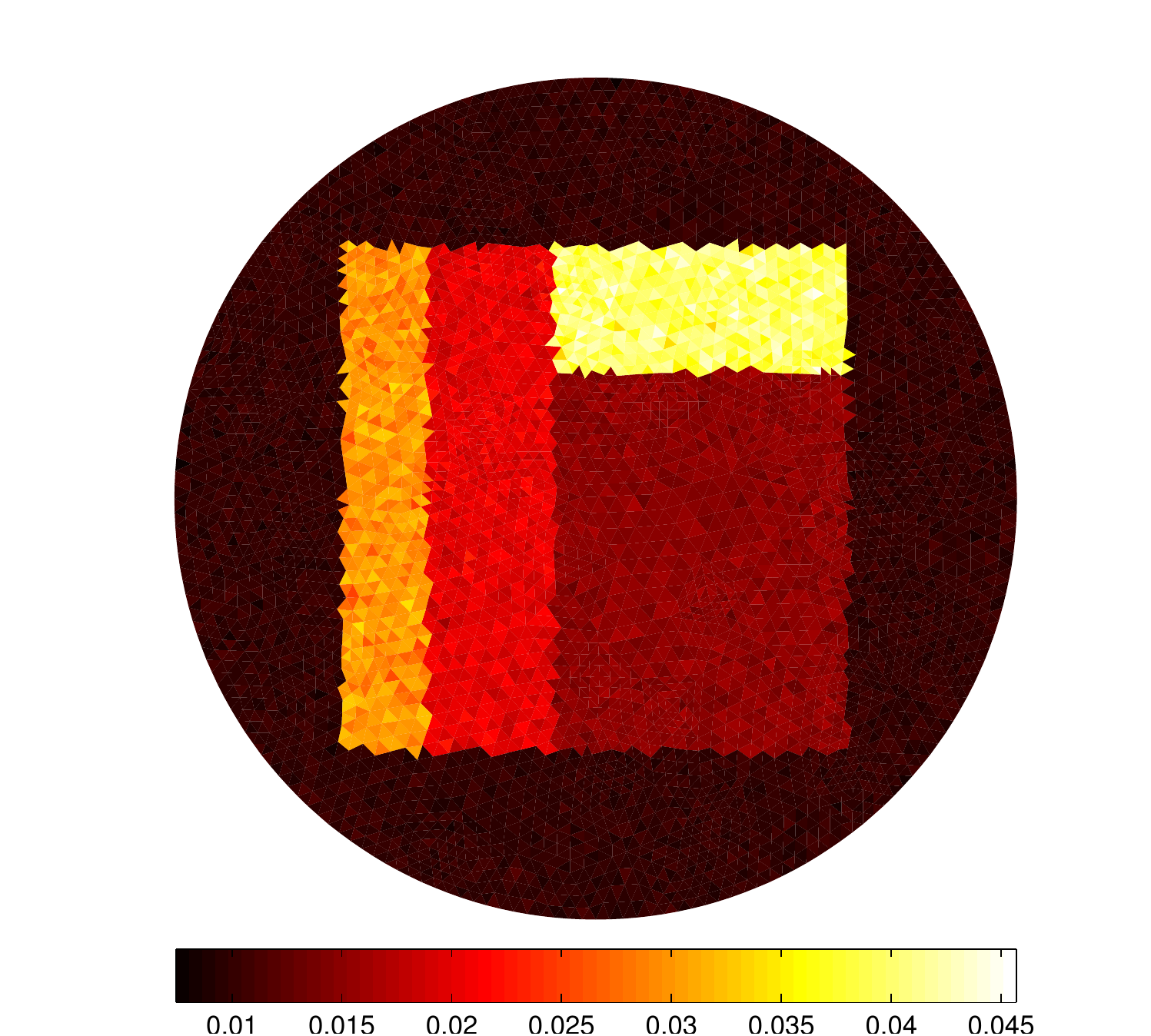}
\end{minipage}
\begin{minipage}{0.24\linewidth}
  \includegraphics[width=\textwidth]{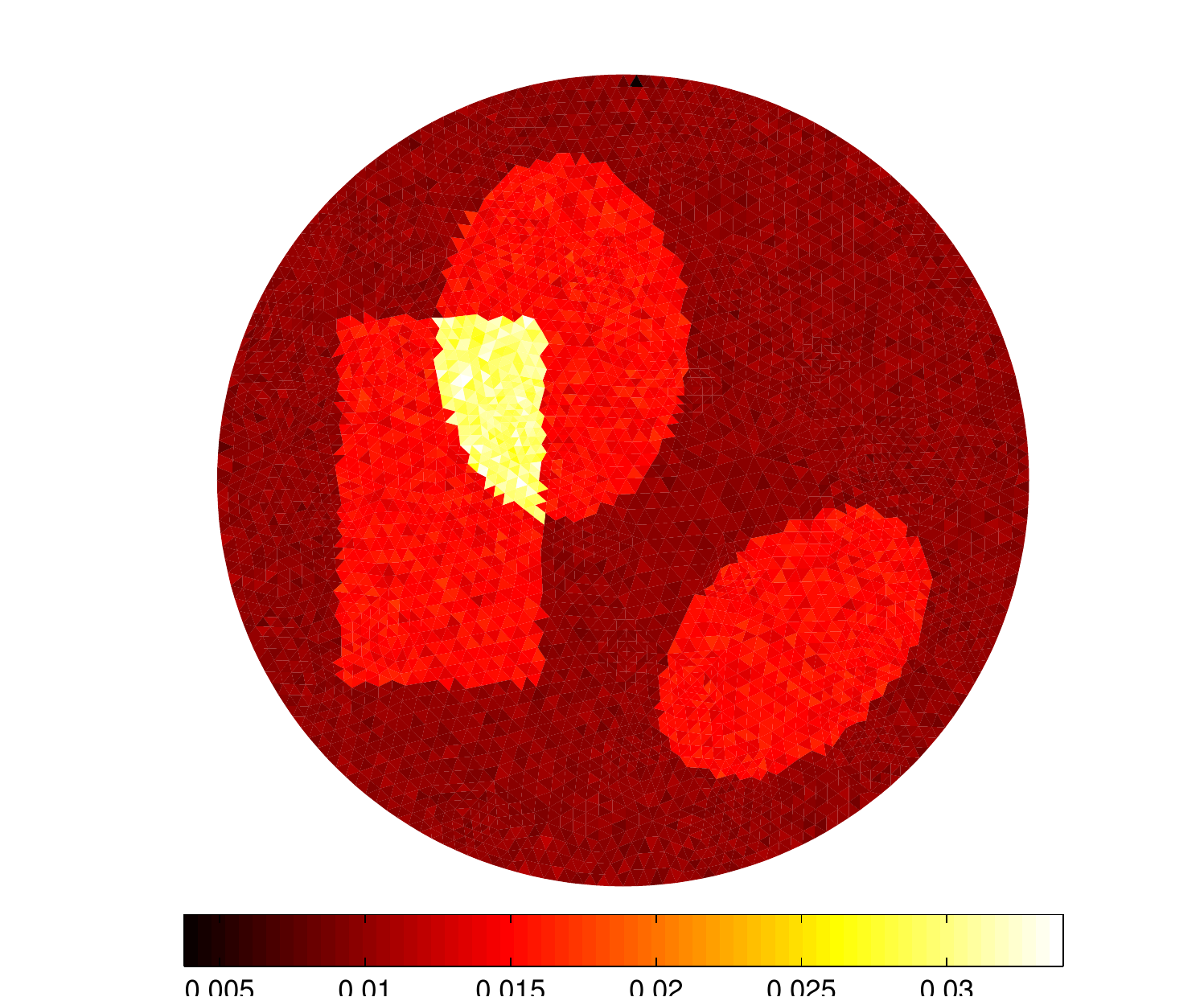}
\end{minipage}
\\
\begin{minipage}{0.24\linewidth}
  \includegraphics[width=\textwidth]{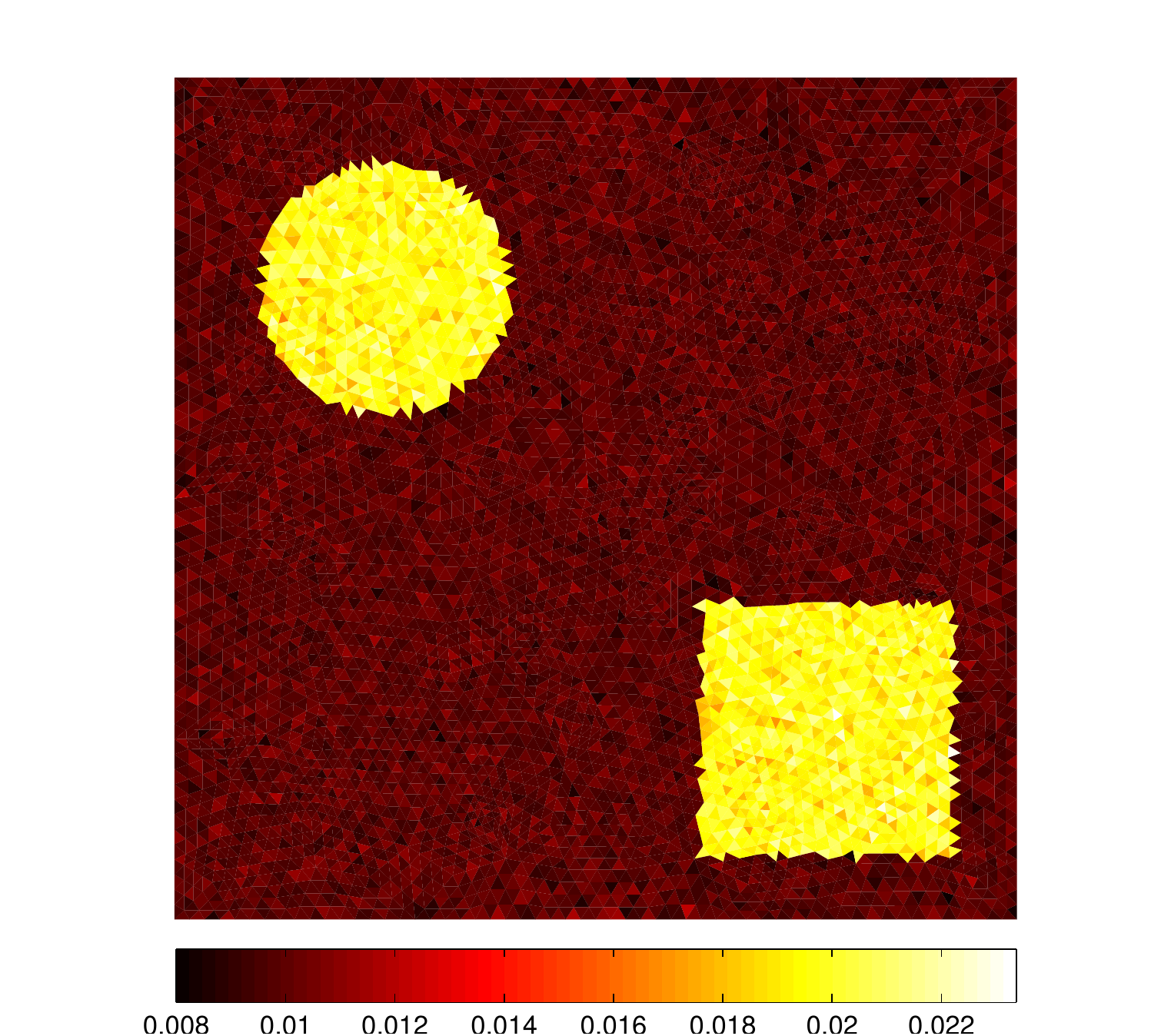}
\end{minipage}
\begin{minipage}{0.24\linewidth}
  \includegraphics[width=\textwidth]{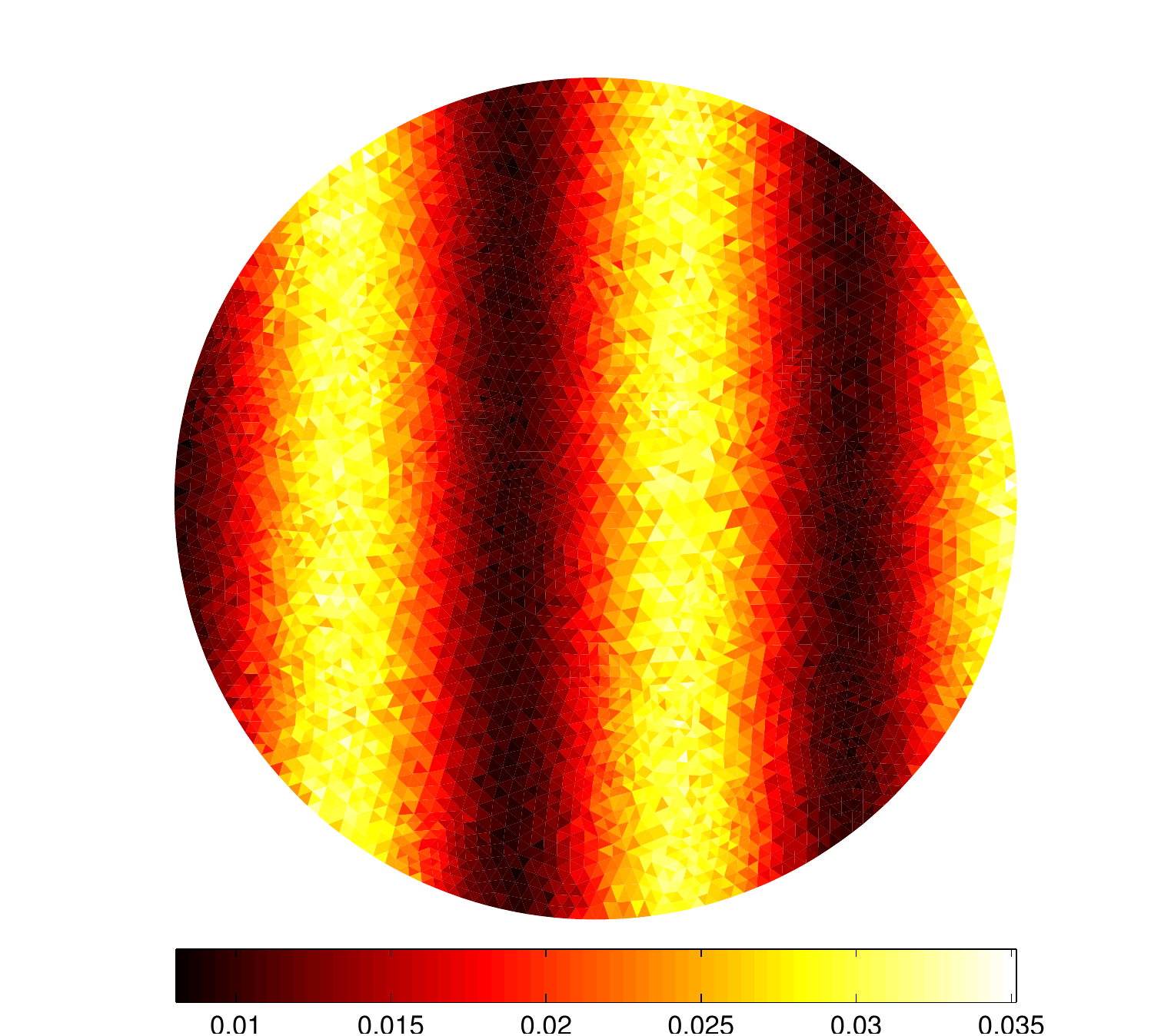}
\end{minipage}
\begin{minipage}{0.24\linewidth}
  \includegraphics[width=\textwidth]{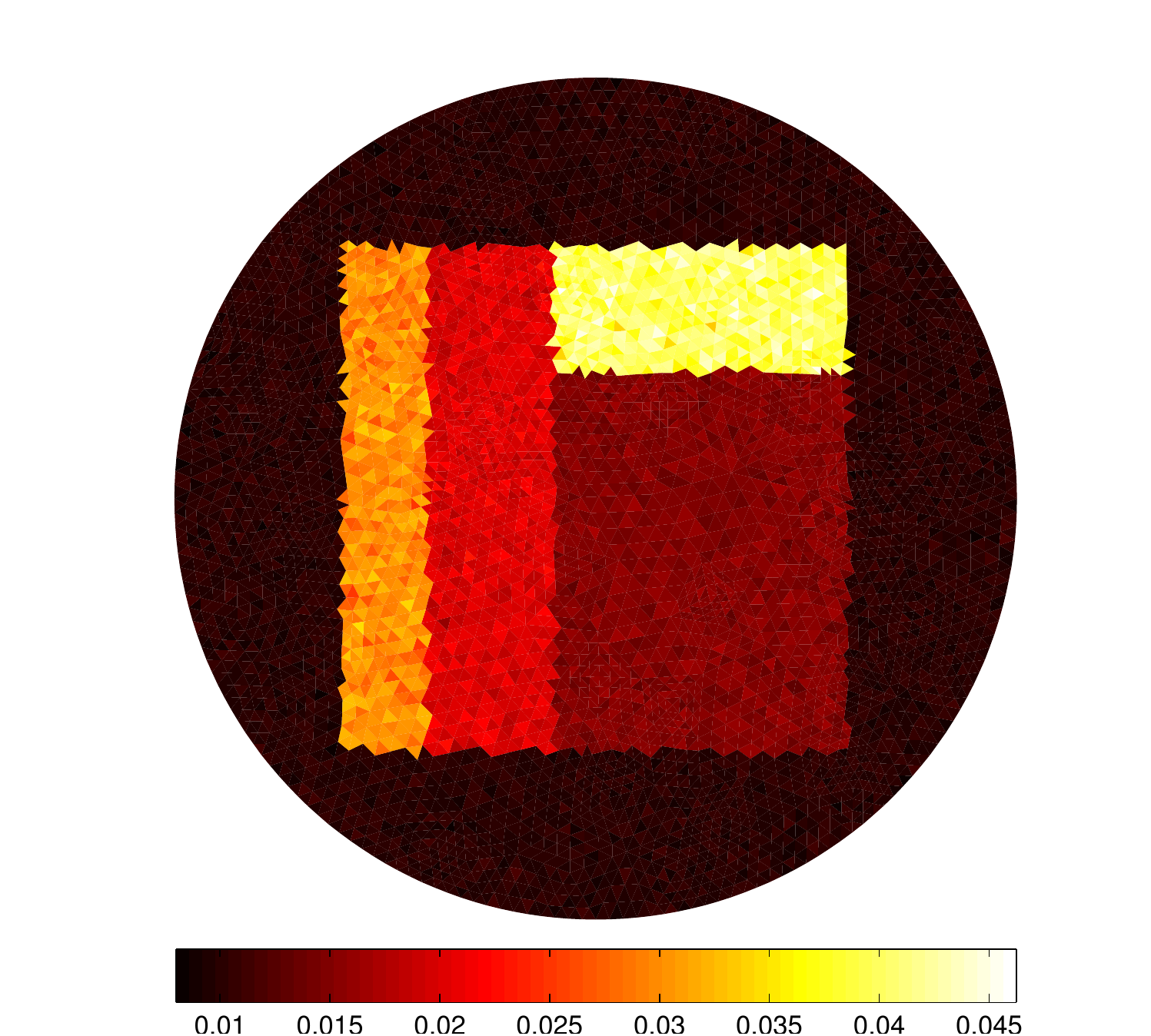}
\end{minipage}
\begin{minipage}{0.24\linewidth}
  \includegraphics[width=\textwidth]{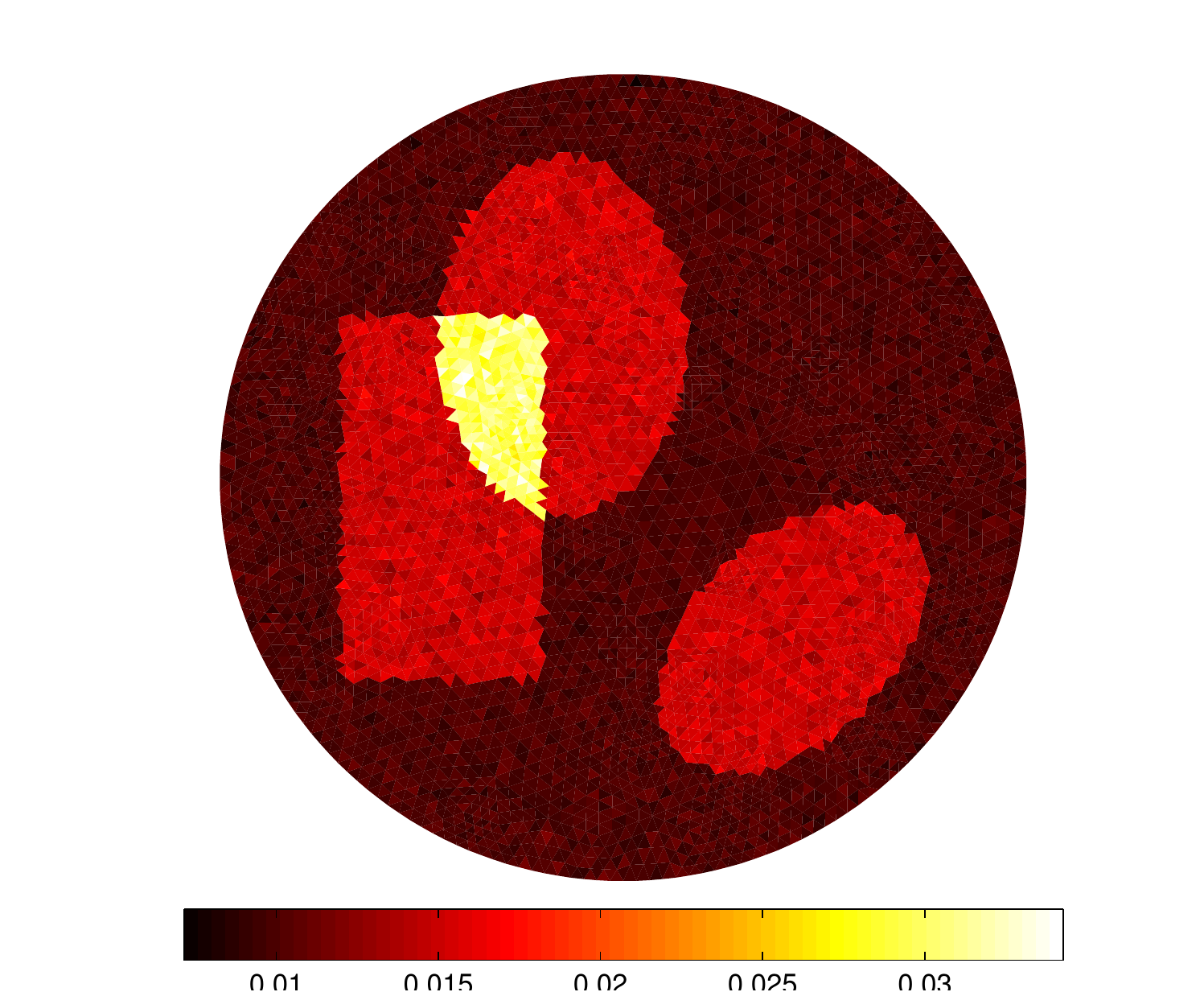}
\end{minipage}
 \caption{\label{fig:3n}Given scattering coefficient, reconstructions of absorption coefficient  from data added by 5\% Gaussian noise. Top row: BB method. Bottom row: improved fixed-point iteration.}
\end{figure}

\begin{figure}[H]
    \centering
\begin{minipage}{0.48\linewidth}
\includegraphics[width=\textwidth]{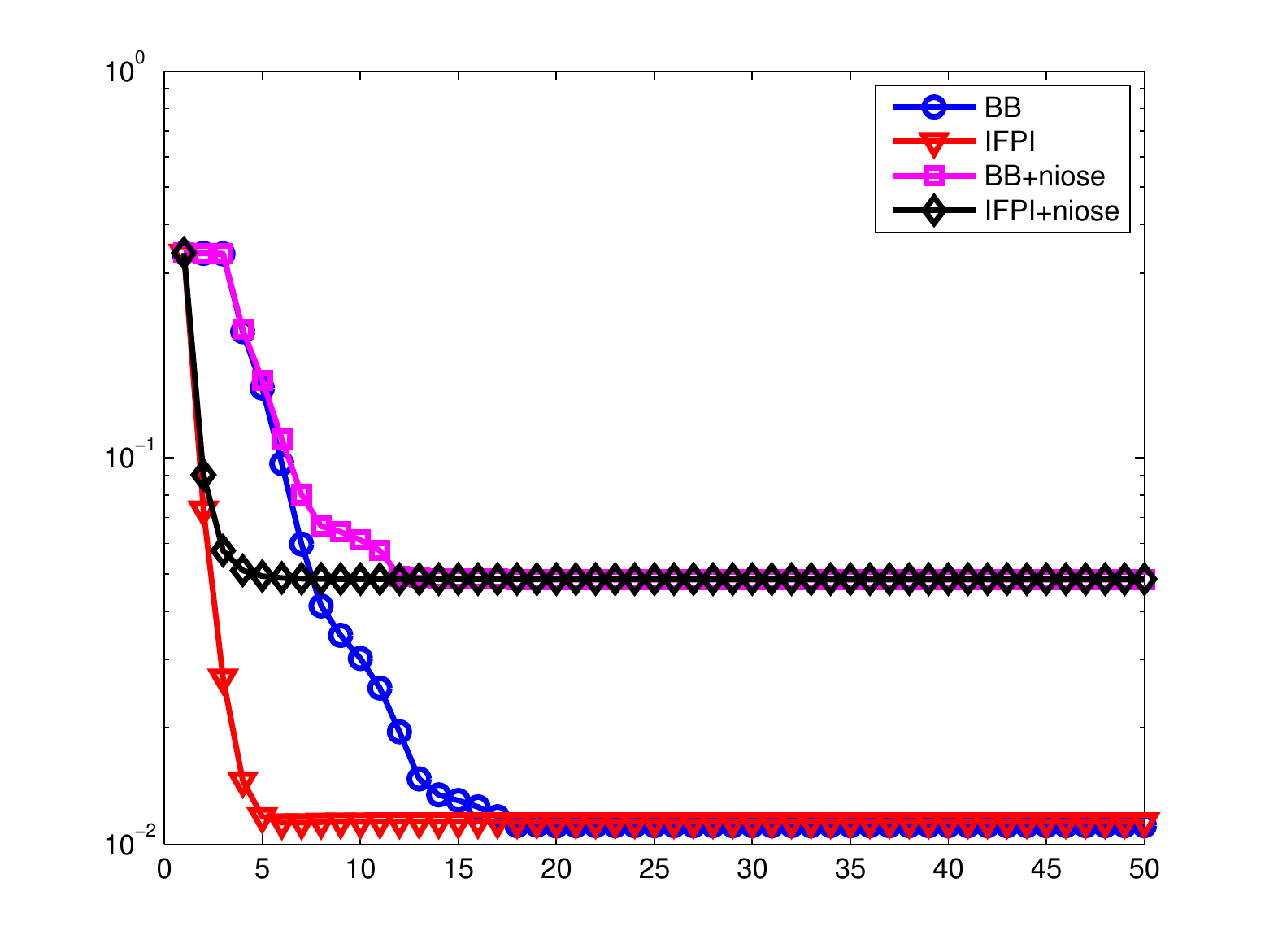}
\end{minipage}
\begin{minipage}{0.48\linewidth}
  \includegraphics[width=\textwidth]{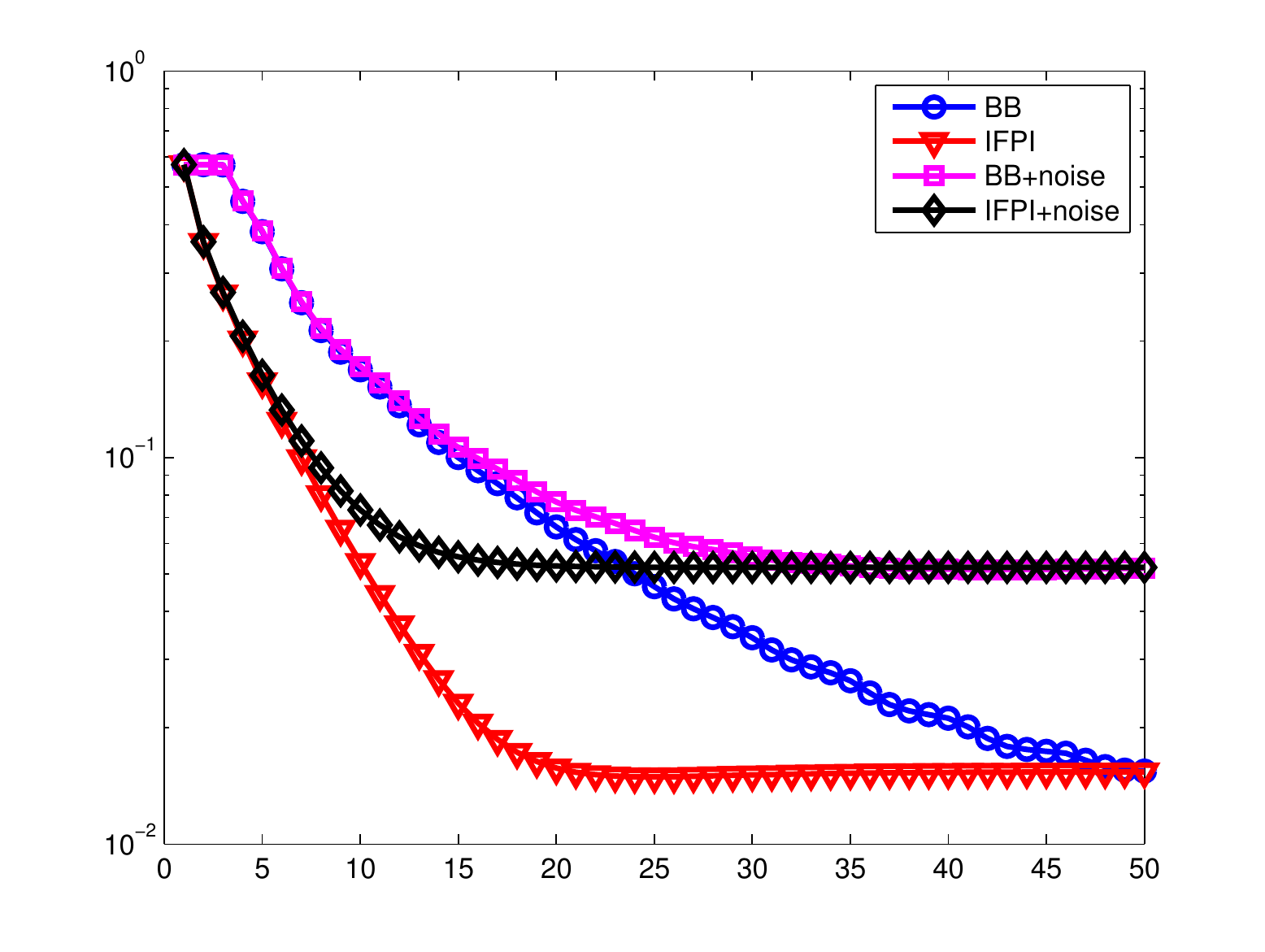}
\end{minipage}
\\
\begin{minipage}{0.48\linewidth}
  \includegraphics[width=\textwidth]{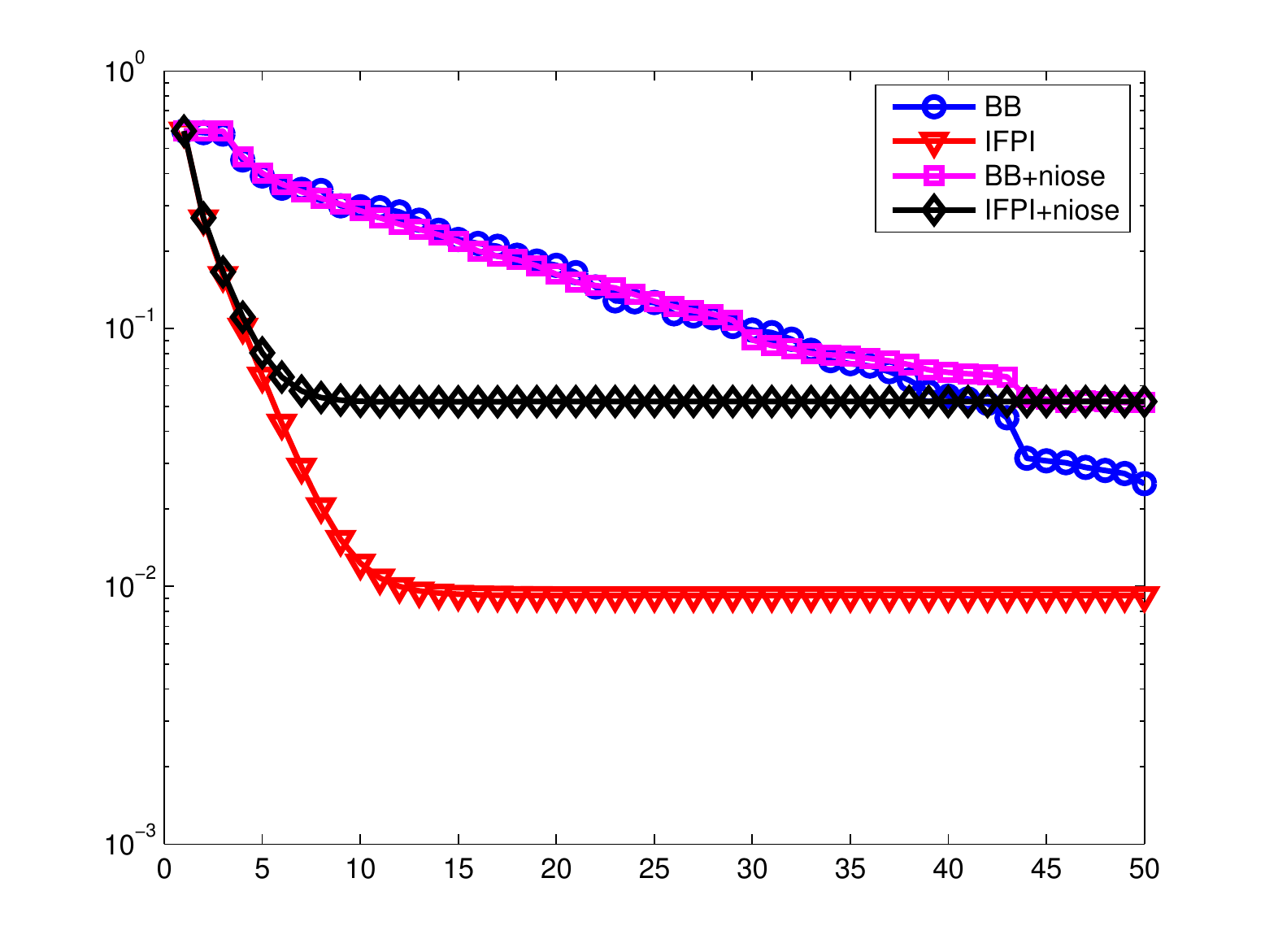}
\end{minipage}
\begin{minipage}{0.48\linewidth}
  \includegraphics[width=\textwidth]{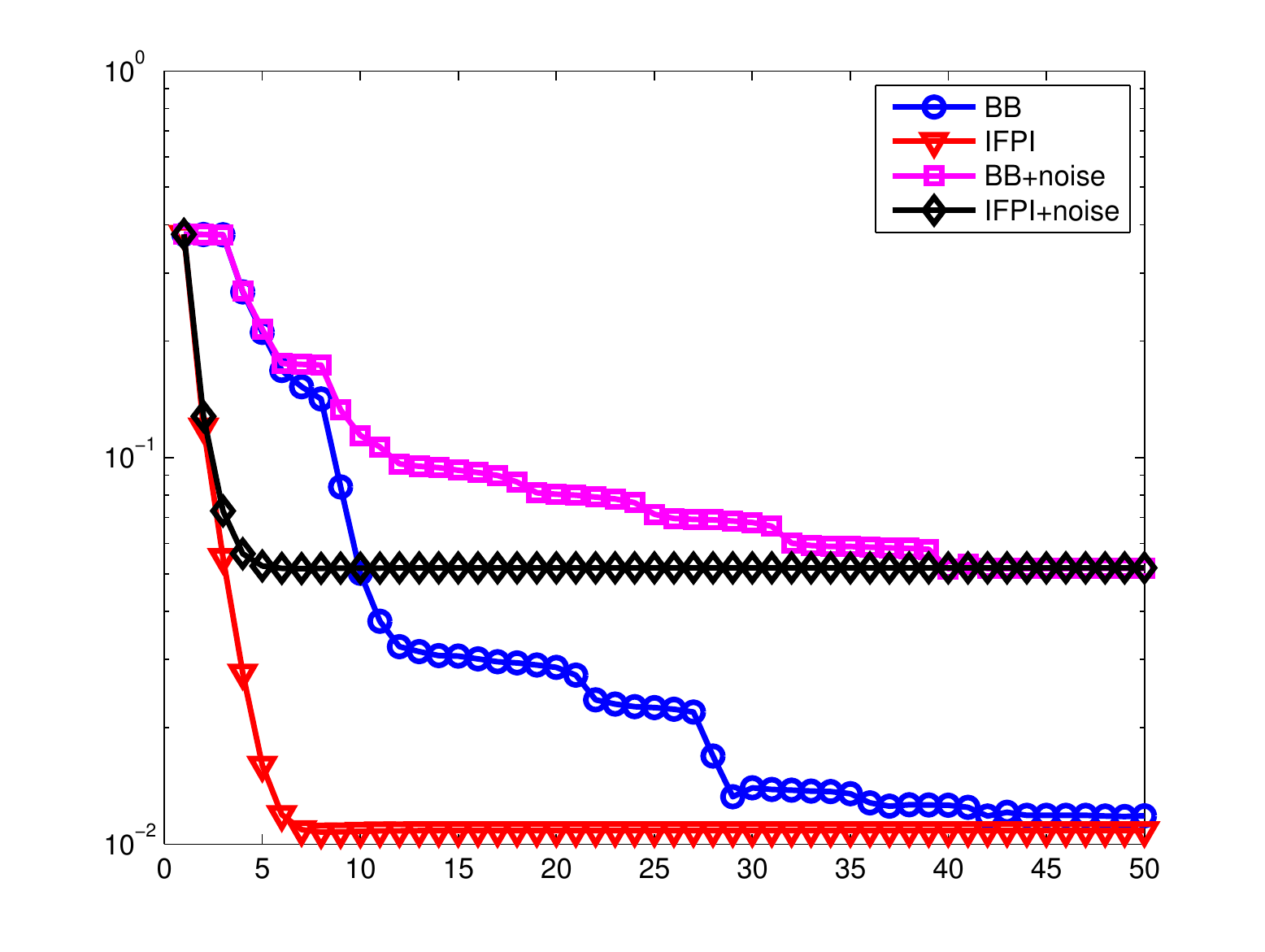}
\end{minipage}
 \caption{\label{fig:4}Specific iterative relative errors $\mathcal{\epsilon}_{\mu_a}$of $\mu_a$, where "BB " and "IFPI" mean the results of BB method and improved fixed-point iteration respectively.}
\end{figure}

\subsubsection{Reconstruction of $\mu_a$ and $\mu_s$ simultaneously}
\label{sec:4.3.2}
We apply BB method to reconstruct absorption and scattering coefficients simultaneously from four measurements. The boundary point sources are placed in the top, bottom, left and right sides respectively. The reconstruction results from noiseless data are showed in Figure~\ref{fig:5}. Corresponding reconstruction results from data added by 5\% Gaussian noise are showed in Figure~\ref{fig:7}. The relative $L^2$ errors of reconstructions of optical coefficients of four templates are tabulated in Table~\ref{tab:1}. From Figure~\ref{fig:5} and Figure~\ref{fig:7}, we can see that BB method retrieves absorption coefficient accurately in piecewise constant and smooth cases, and even is stable to noisy data. In the first, third and fourth templates, some jagged borders of inclusions are reasonable, which are caused by the interpolation of data from fine mesh to coarse mesh. Nevertheless, the reconstruction of scattering coefficient is not satisfactory. In piecewise constant case, the edges between different pieces are blur, which can be seen from the first, third, and fourth template in Figure~\ref{fig:5} and Figure~\ref{fig:7}. This is due to the insensitive scattering coefficient. After all, there is no explicit $\mu_s$ in the formula \eqref{eq:5}.

\begin{table}[H]
 \begin{center}
\begin{tabular}{|c|c|c|c|c|}
\hline
    &  \multicolumn{2}{c|}{noise-free data}             &   \multicolumn{2}{c|}{noisy data}             \\
\hline
    &  $\epsilon_{\mu_a}$          &  $\epsilon_{\mu_s}$  &  $\epsilon_{\mu_a}$     &  $\epsilon_{\mu_s}$  \\
\hline
 1  &  4.61e-2       &  1.52e-1  &  1.09e-1  &  1.81e-1  \\
 2  &  2.65e-2       &  1.53e-1  &   7.71e-2  &  1.64e-1  \\
 3  &  3.17e-2        &  1.09e-1  &   1.29e-1  &  1.37e-1  \\
 4  &  5.32e-2        &  1.26e-1  &   1.28e-1  &  1.49e-1  \\
\hline
\end{tabular}
\end{center}
  \caption{Relative $L^2$ errors on reconstructions by BB method.}
  \label{tab:1}
\end{table}

\begin{figure}[H]
    \centering
\begin{minipage}{0.24\linewidth}
\includegraphics[width=\textwidth]{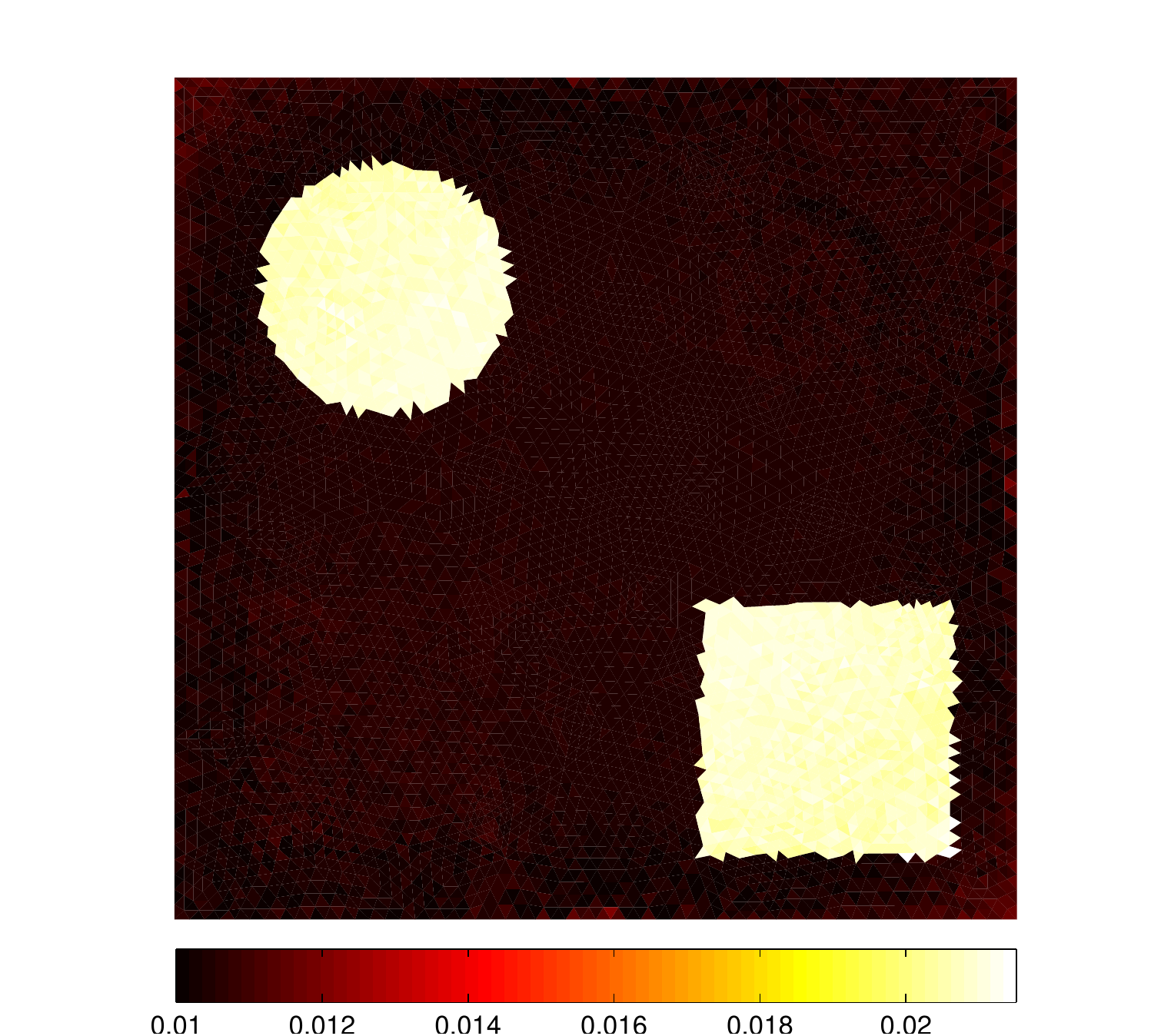}
\end{minipage}
\begin{minipage}{0.24\linewidth}
  \includegraphics[width=\textwidth]{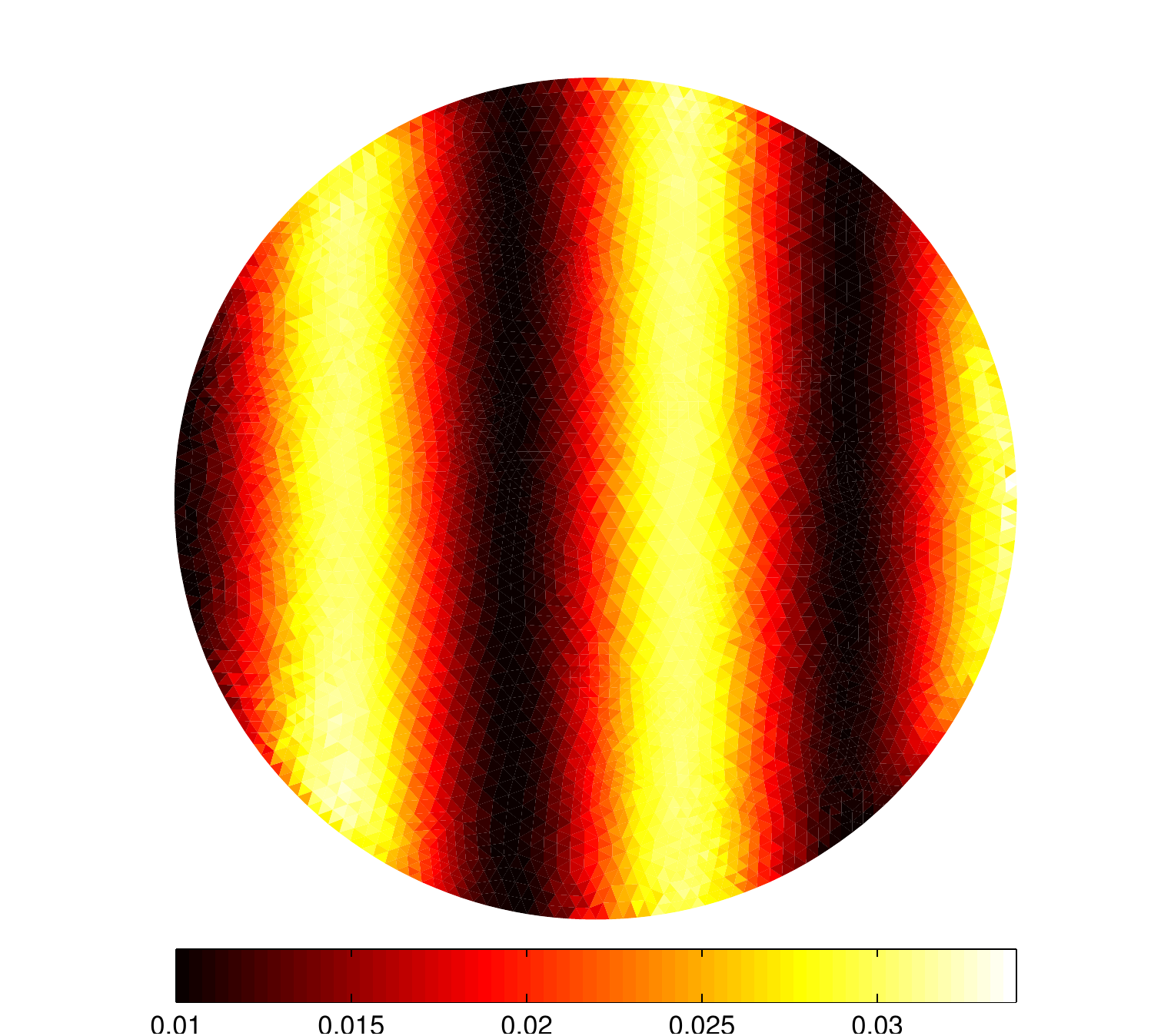}
\end{minipage}
\begin{minipage}{0.24\linewidth}
  \includegraphics[width=\textwidth]{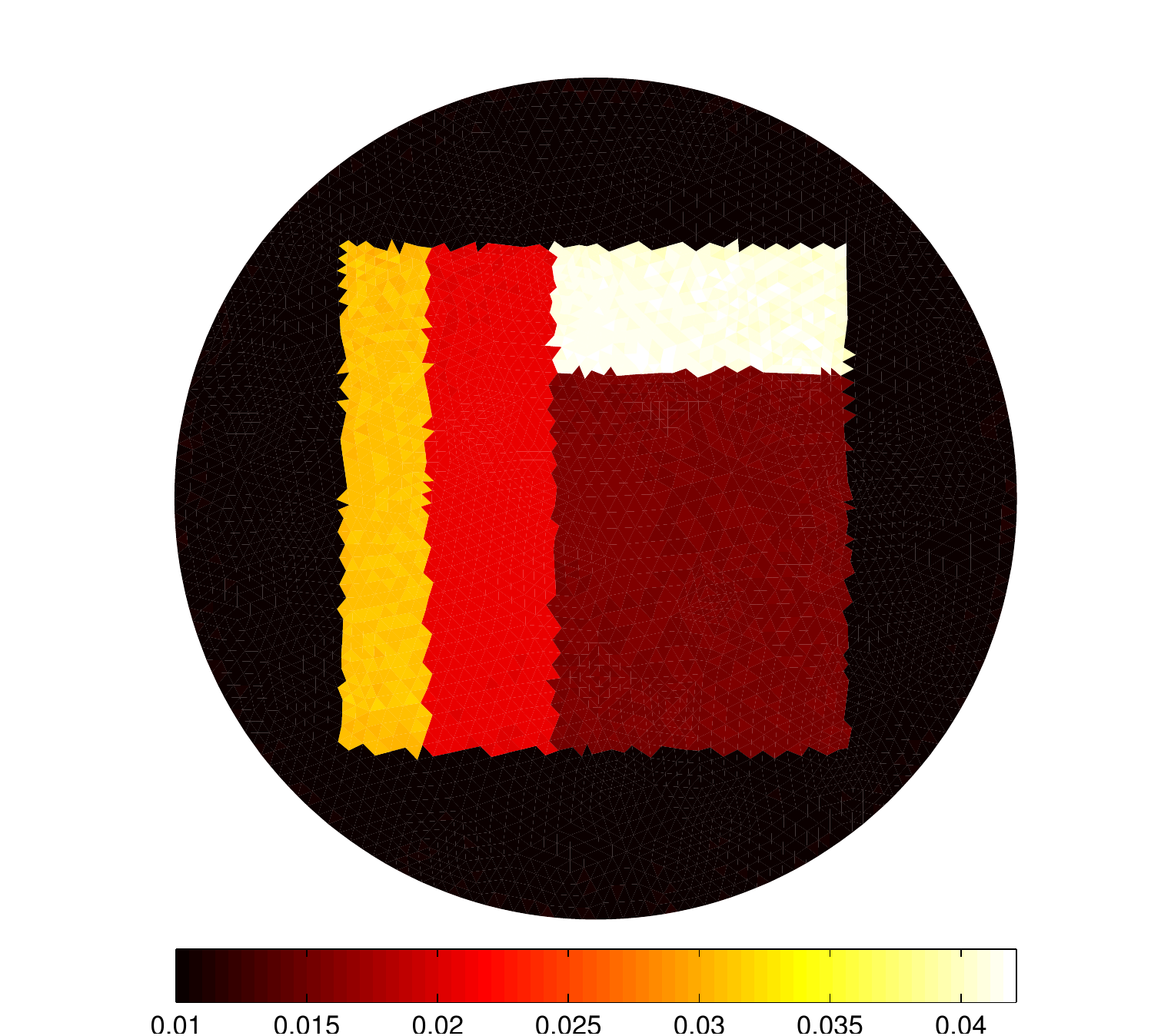}
\end{minipage}
\begin{minipage}{0.24\linewidth}
  \includegraphics[width=\textwidth]{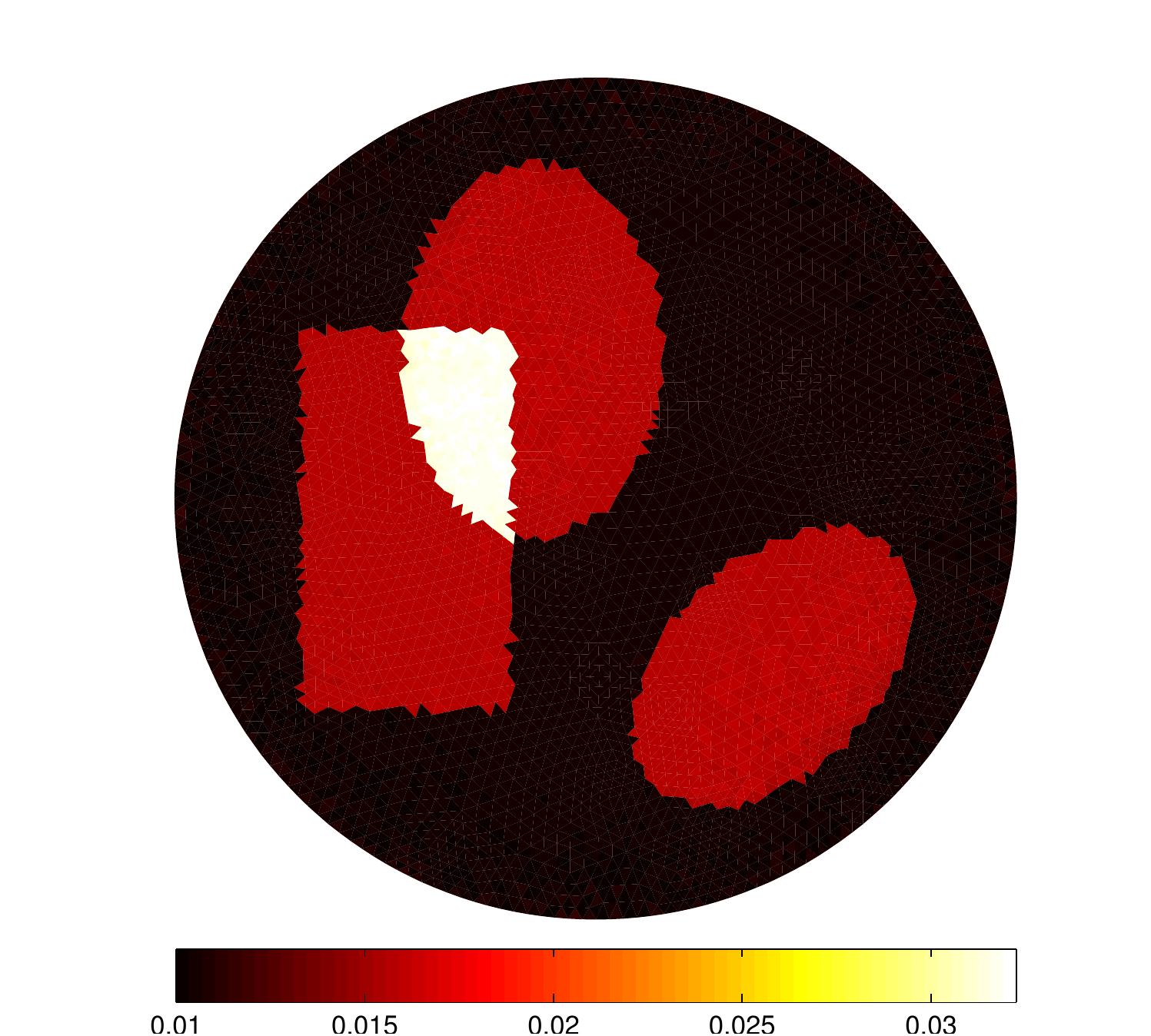}
\end{minipage}
\\
\begin{minipage}{0.24\linewidth}
  \includegraphics[width=\textwidth]{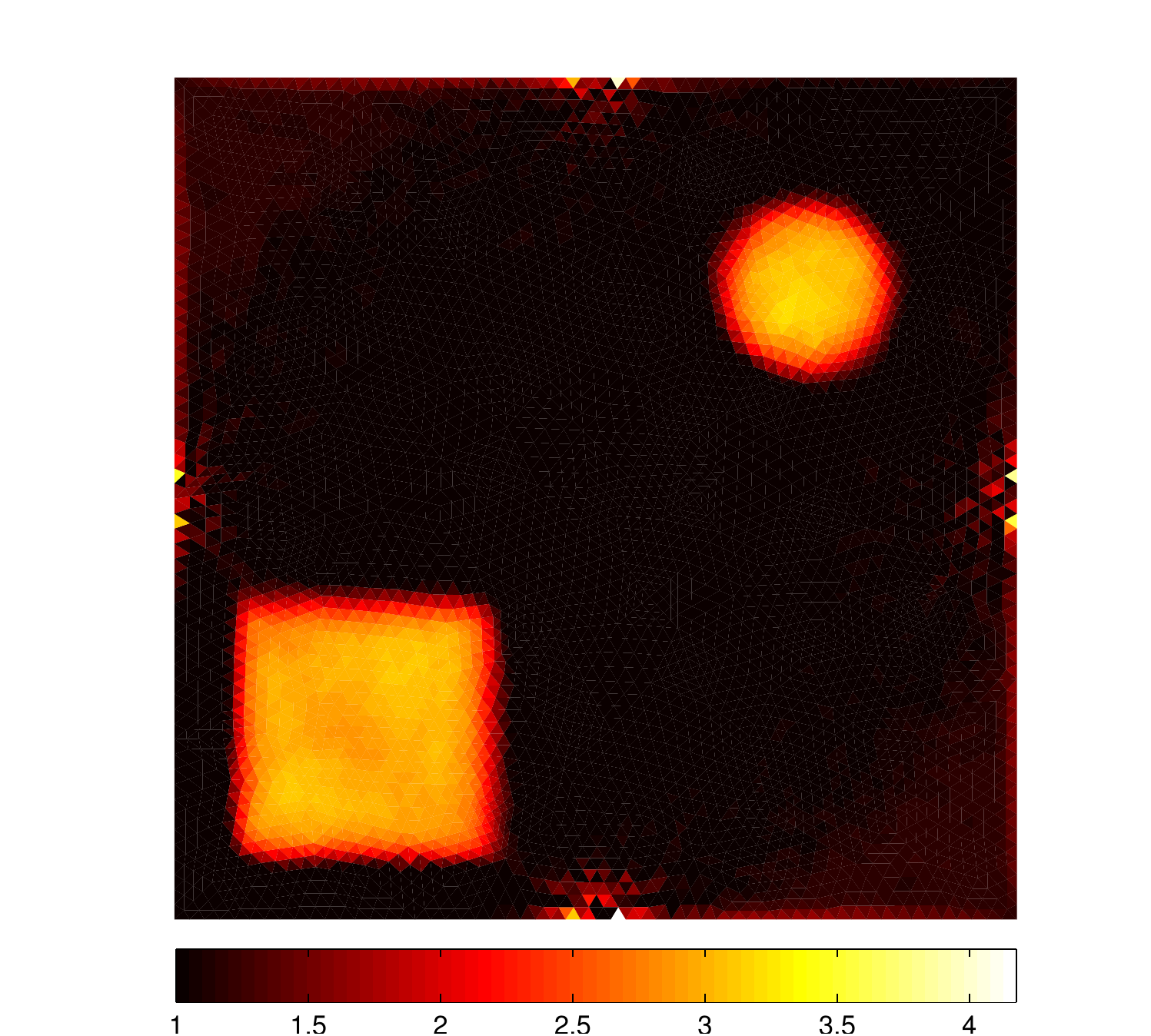}
\end{minipage}
\begin{minipage}{0.24\linewidth}
  \includegraphics[width=\textwidth]{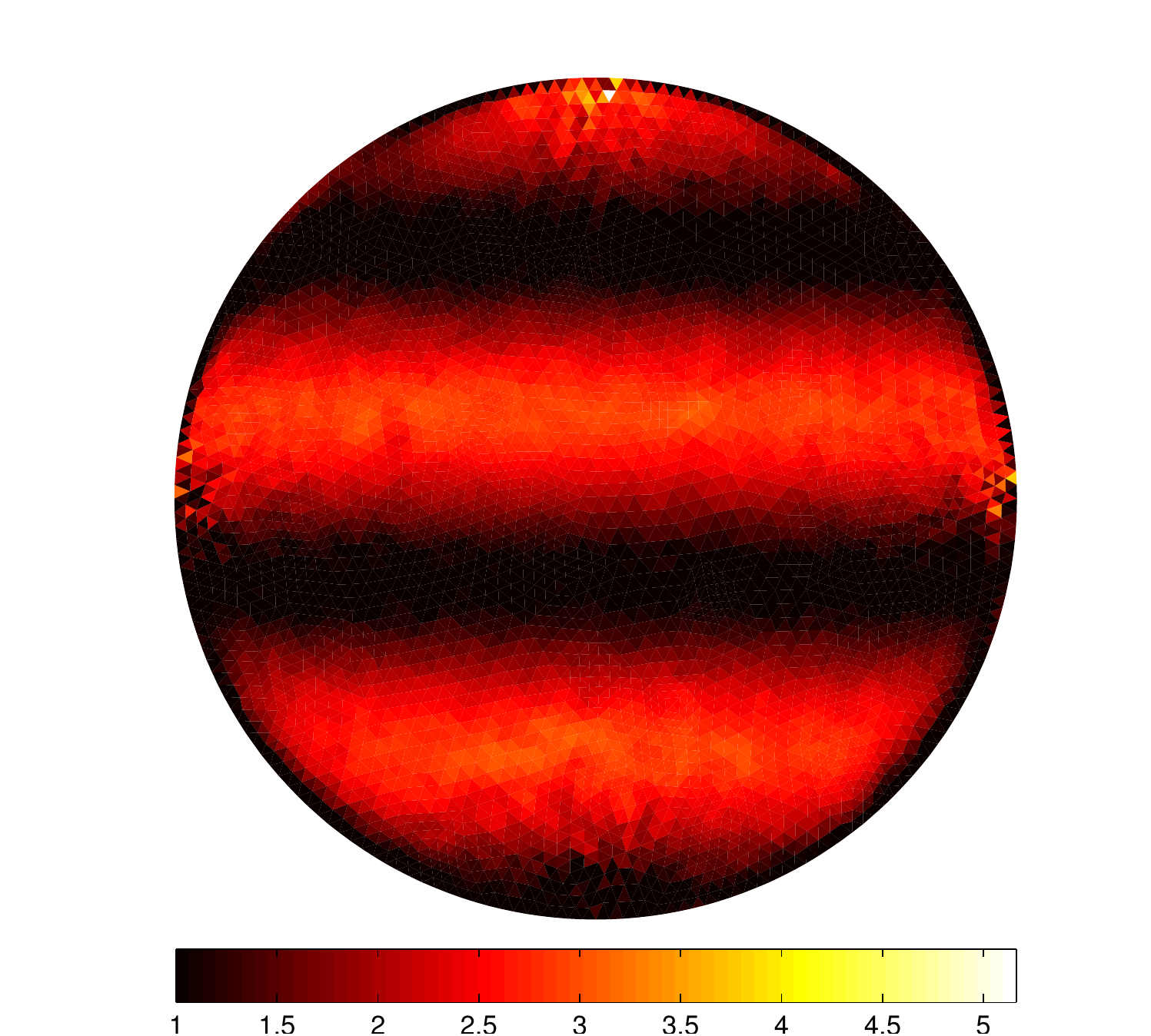}
\end{minipage}
\begin{minipage}{0.24\linewidth}
  \includegraphics[width=\textwidth]{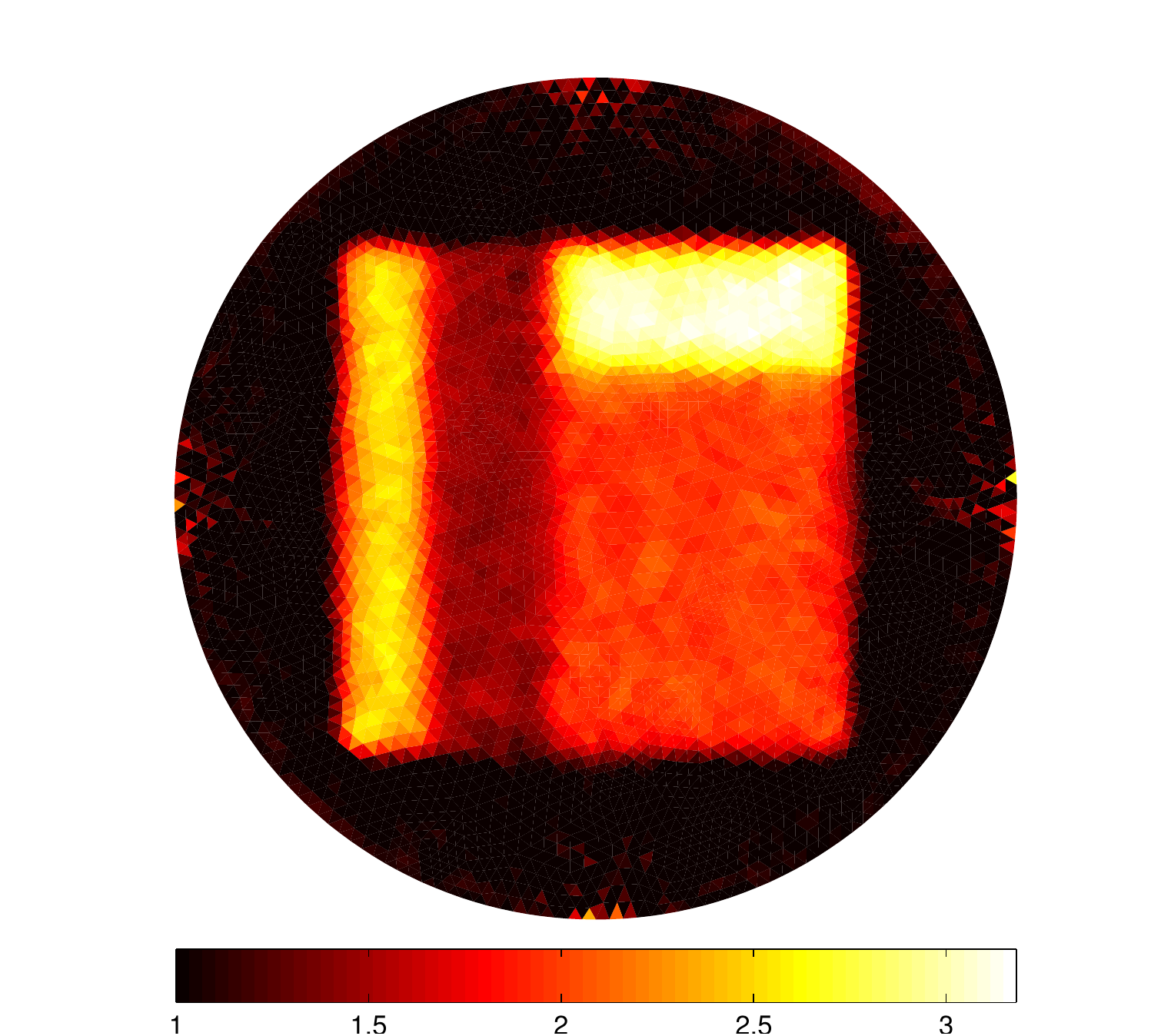}
\end{minipage}
\begin{minipage}{0.24\linewidth}
  \includegraphics[width=\textwidth]{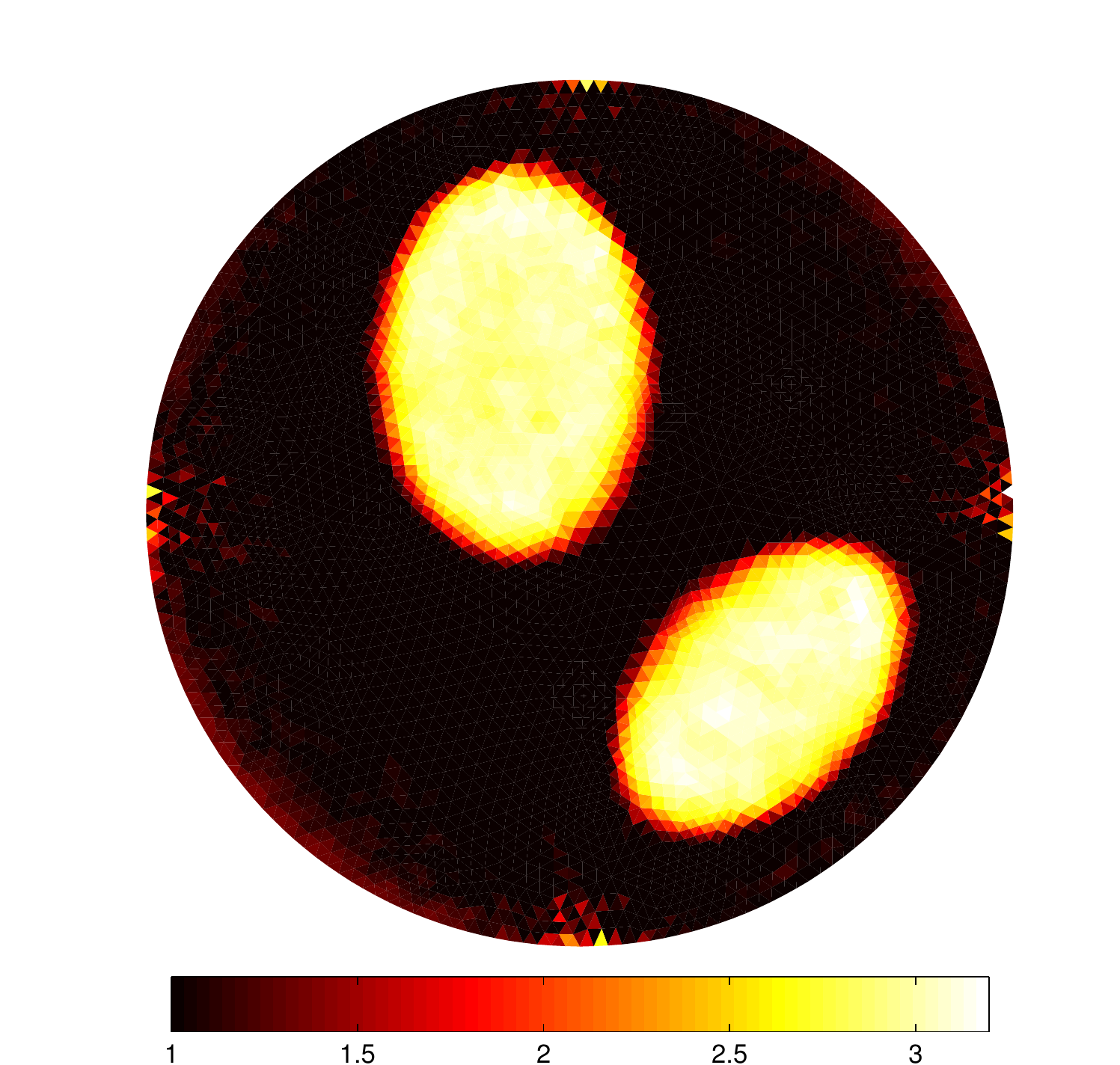}
\end{minipage}
 \caption{\label{fig:5}Reconstructions of optical coefficients by BB method from noiseless data. Top row: reconstruction of $\mu_a$. Bottom row: reconstruction of $\mu_s$.}
\end{figure}

\begin{figure}[H]
    \centering
\begin{minipage}{0.24\linewidth}
\includegraphics[width=\textwidth]{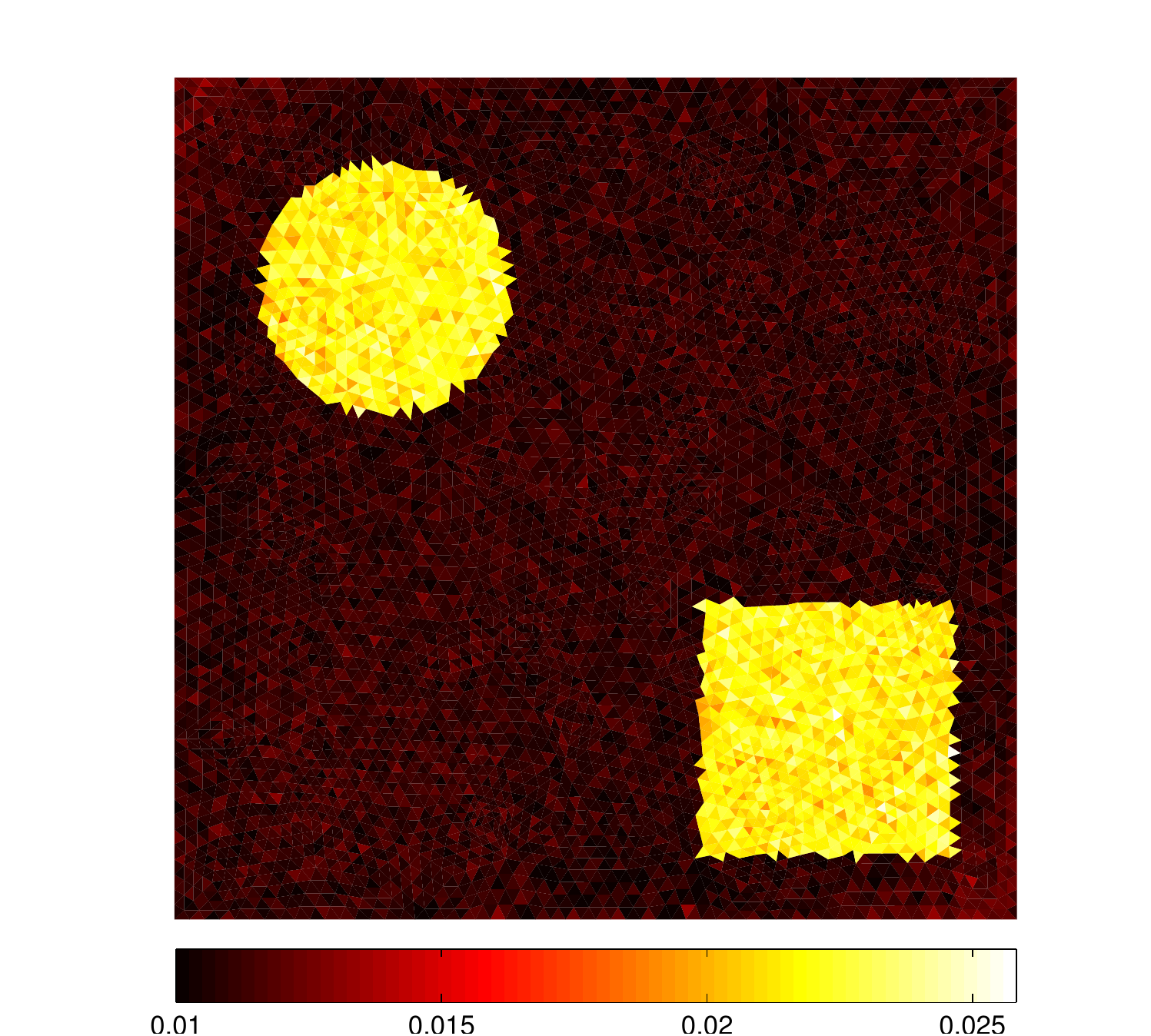}
\end{minipage}
\begin{minipage}{0.24\linewidth}
  \includegraphics[width=\textwidth]{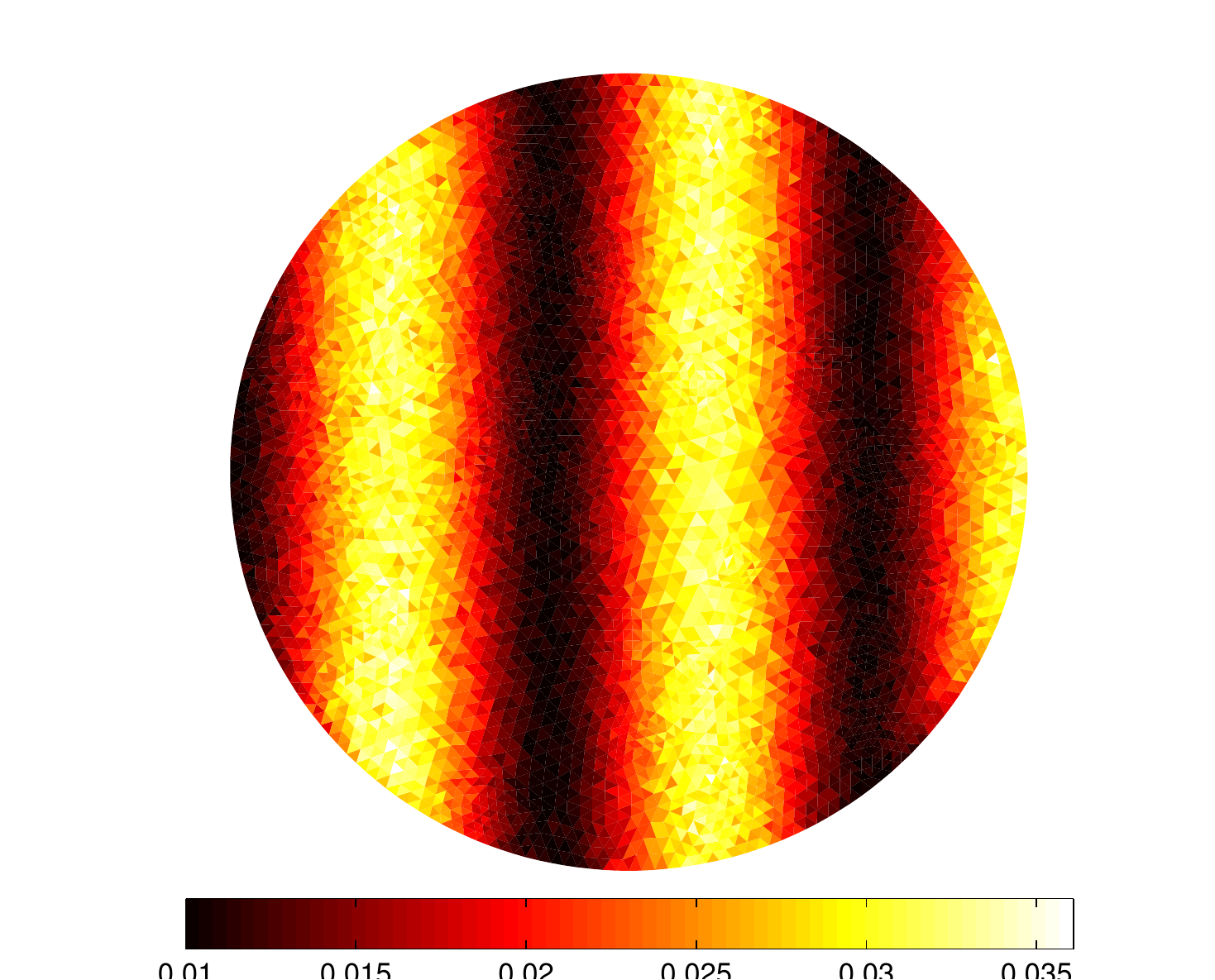}
\end{minipage}
\begin{minipage}{0.24\linewidth}
  \includegraphics[width=\textwidth]{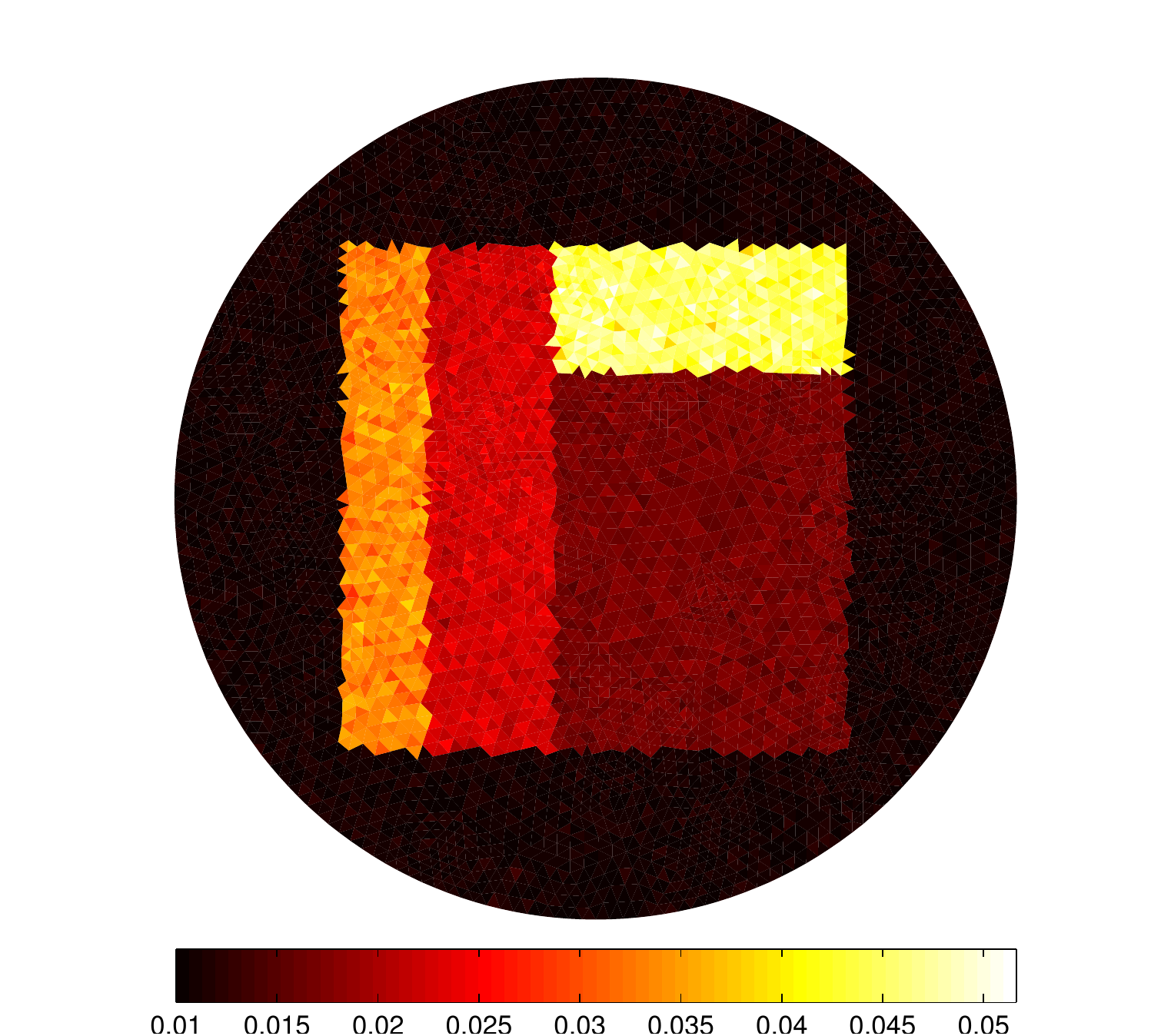}
\end{minipage}
\begin{minipage}{0.24\linewidth}
-  \includegraphics[width=\textwidth]{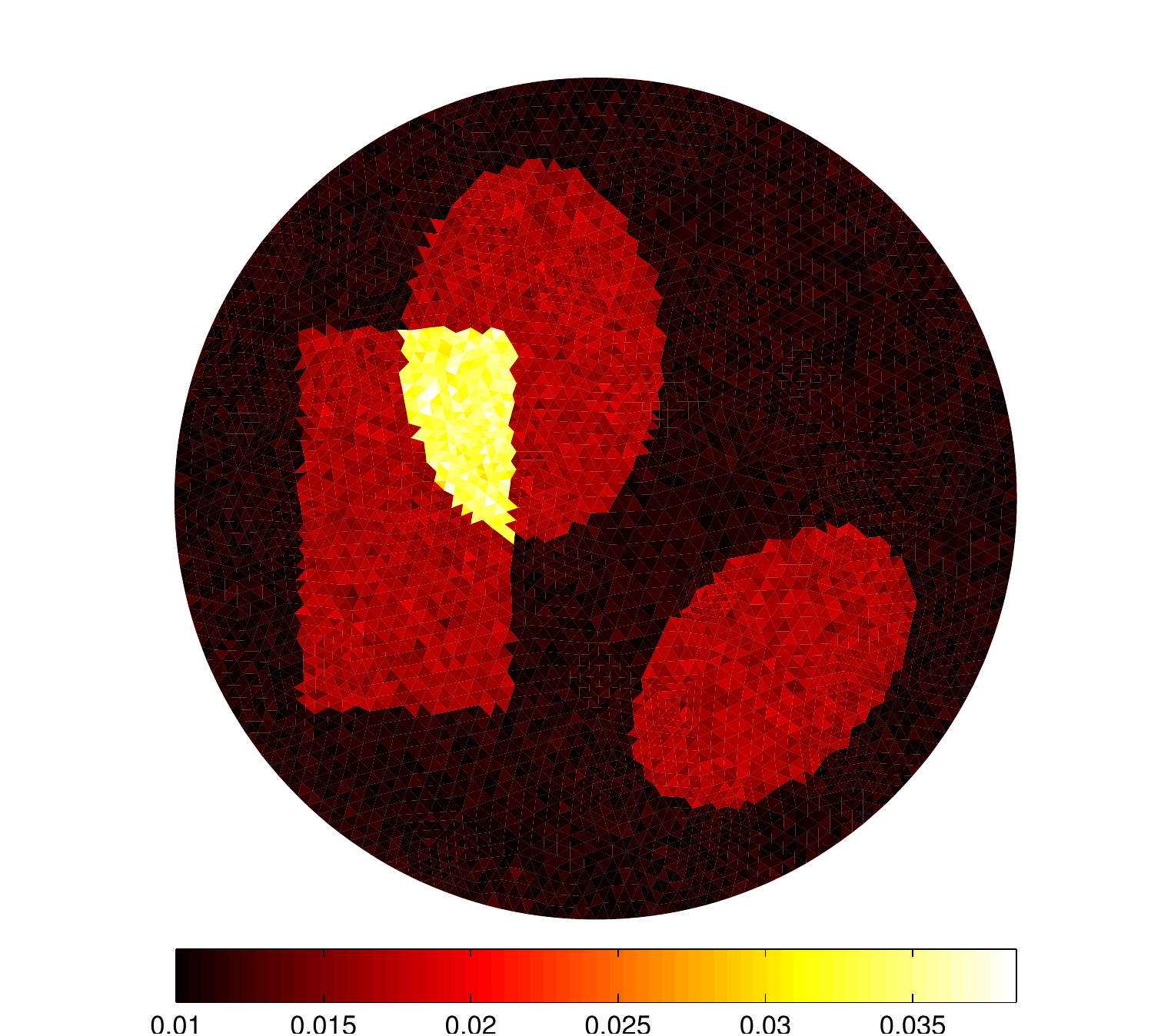}
\end{minipage}
\\
\begin{minipage}{0.24\linewidth}
  \includegraphics[width=\textwidth]{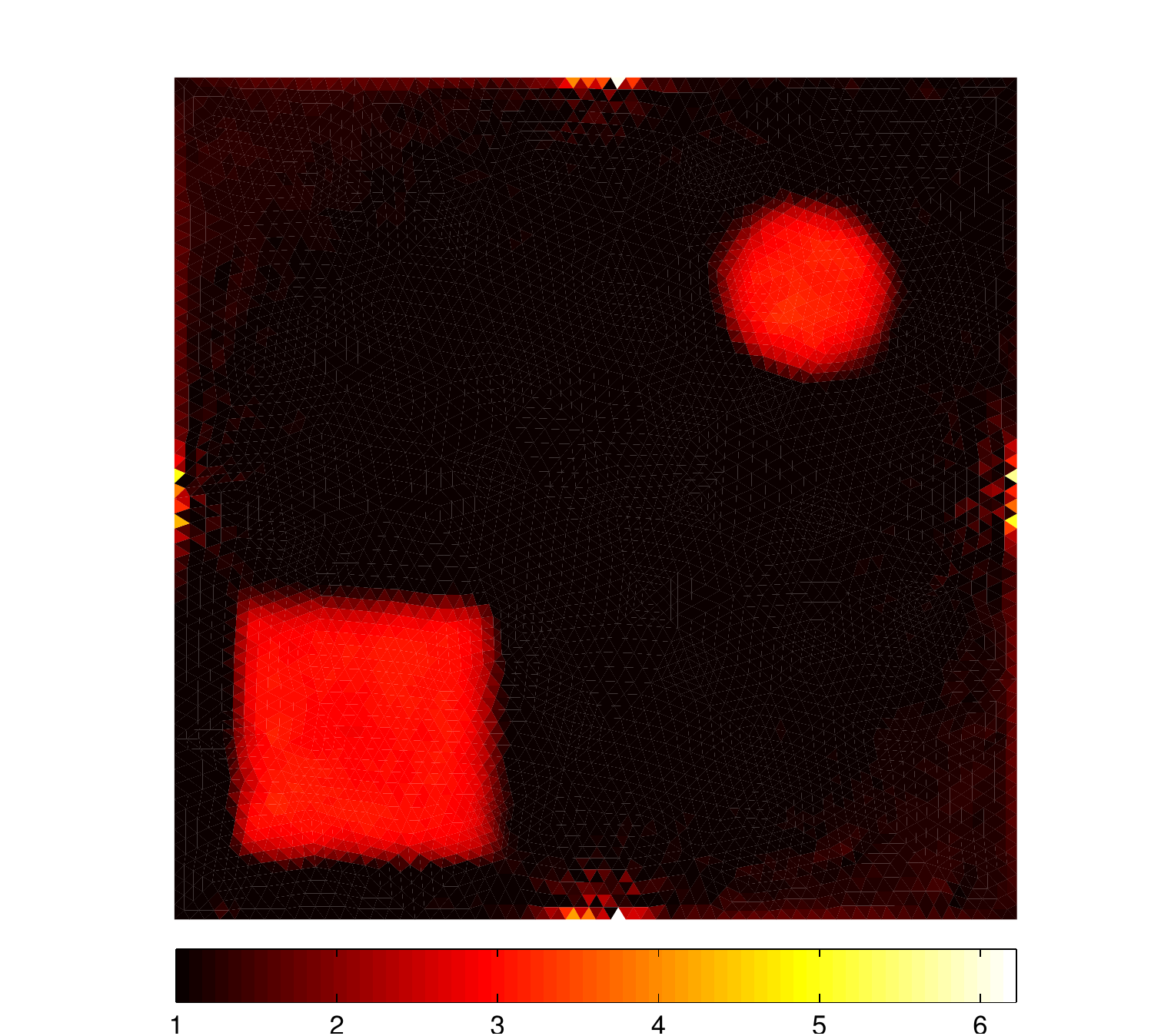}
\end{minipage}
\begin{minipage}{0.24\linewidth}
  \includegraphics[width=\textwidth]{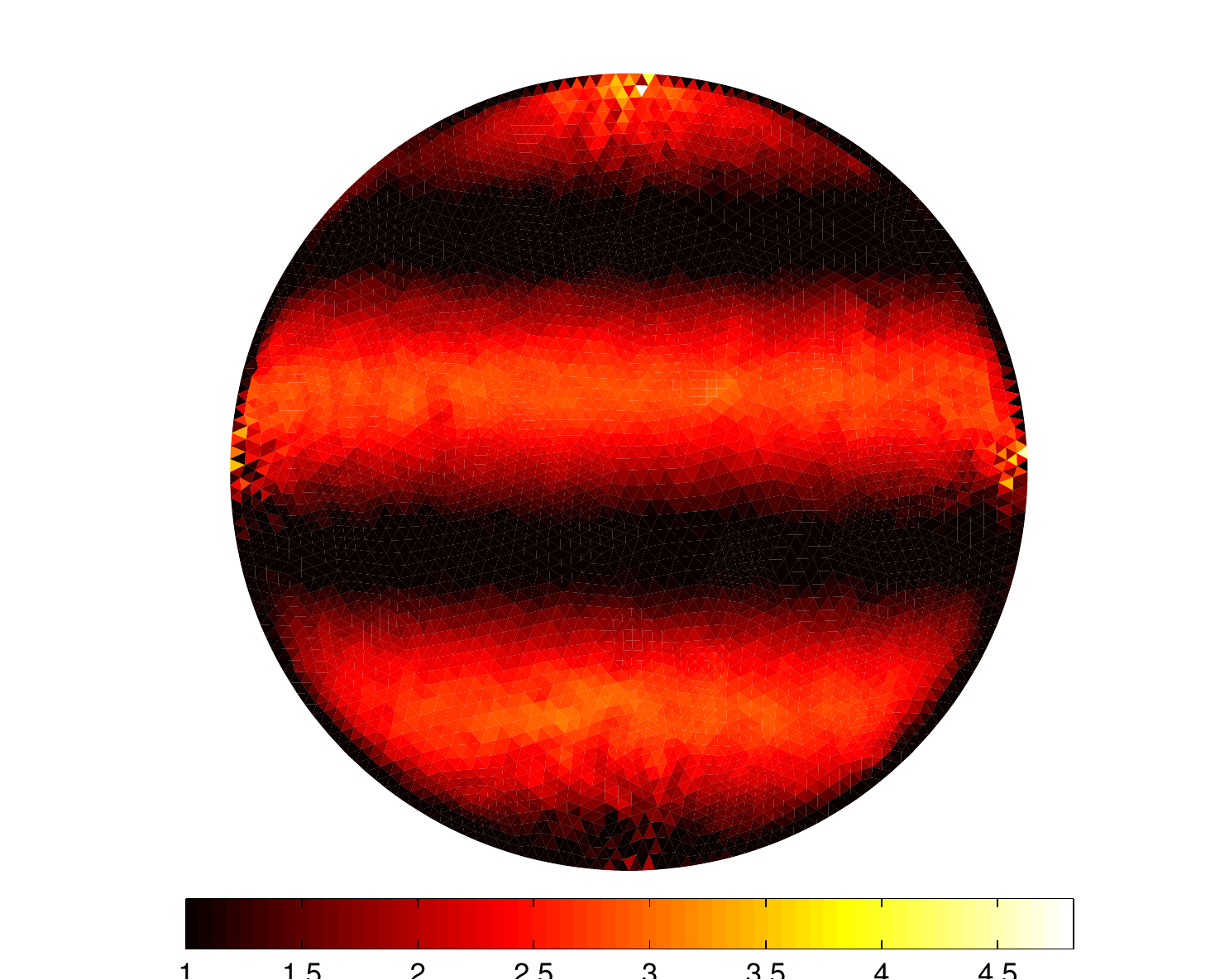}
\end{minipage}
\begin{minipage}{0.24\linewidth}
  \includegraphics[width=\textwidth]{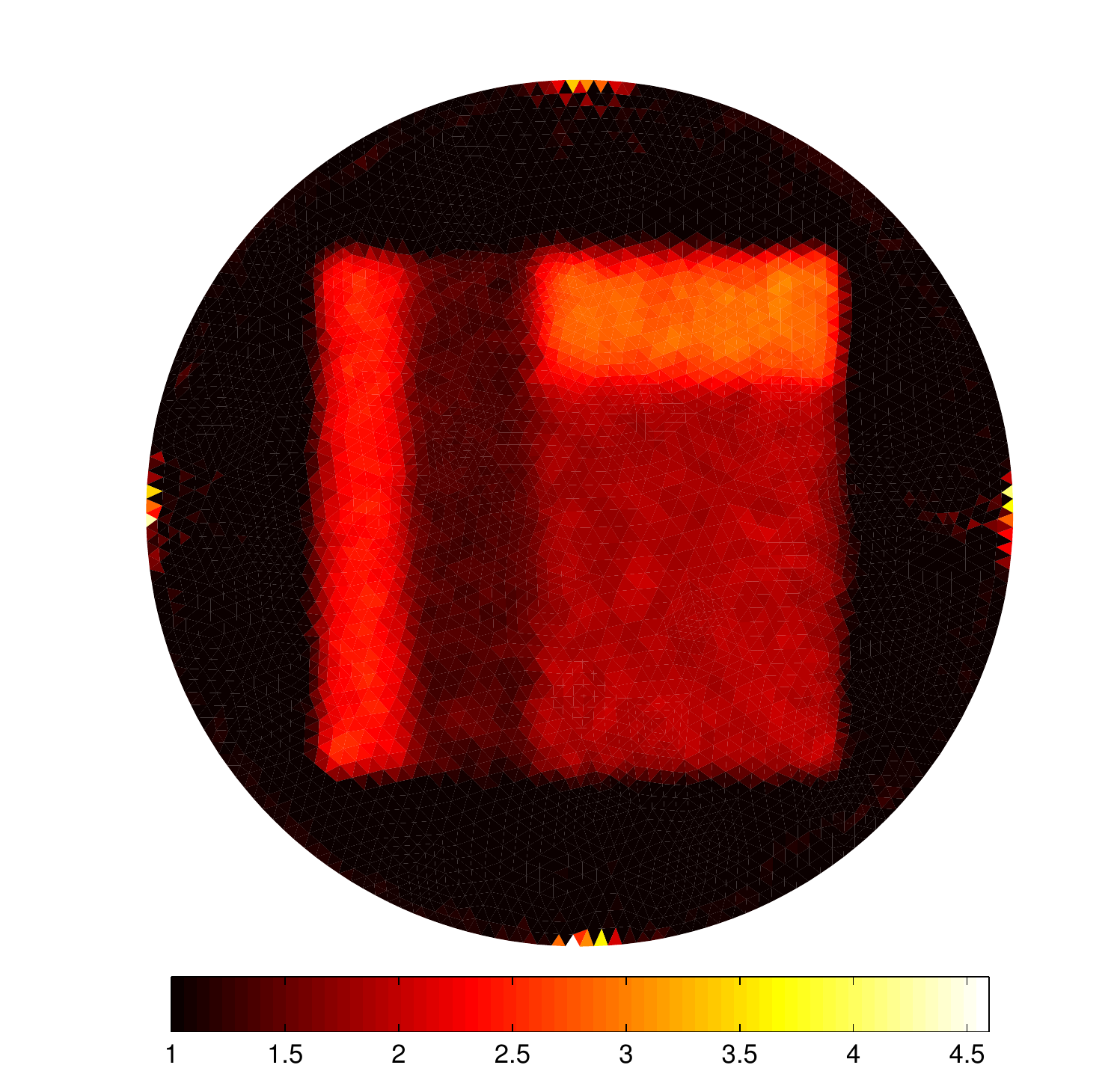}
\end{minipage}
\begin{minipage}{0.24\linewidth}
  \includegraphics[width=\textwidth]{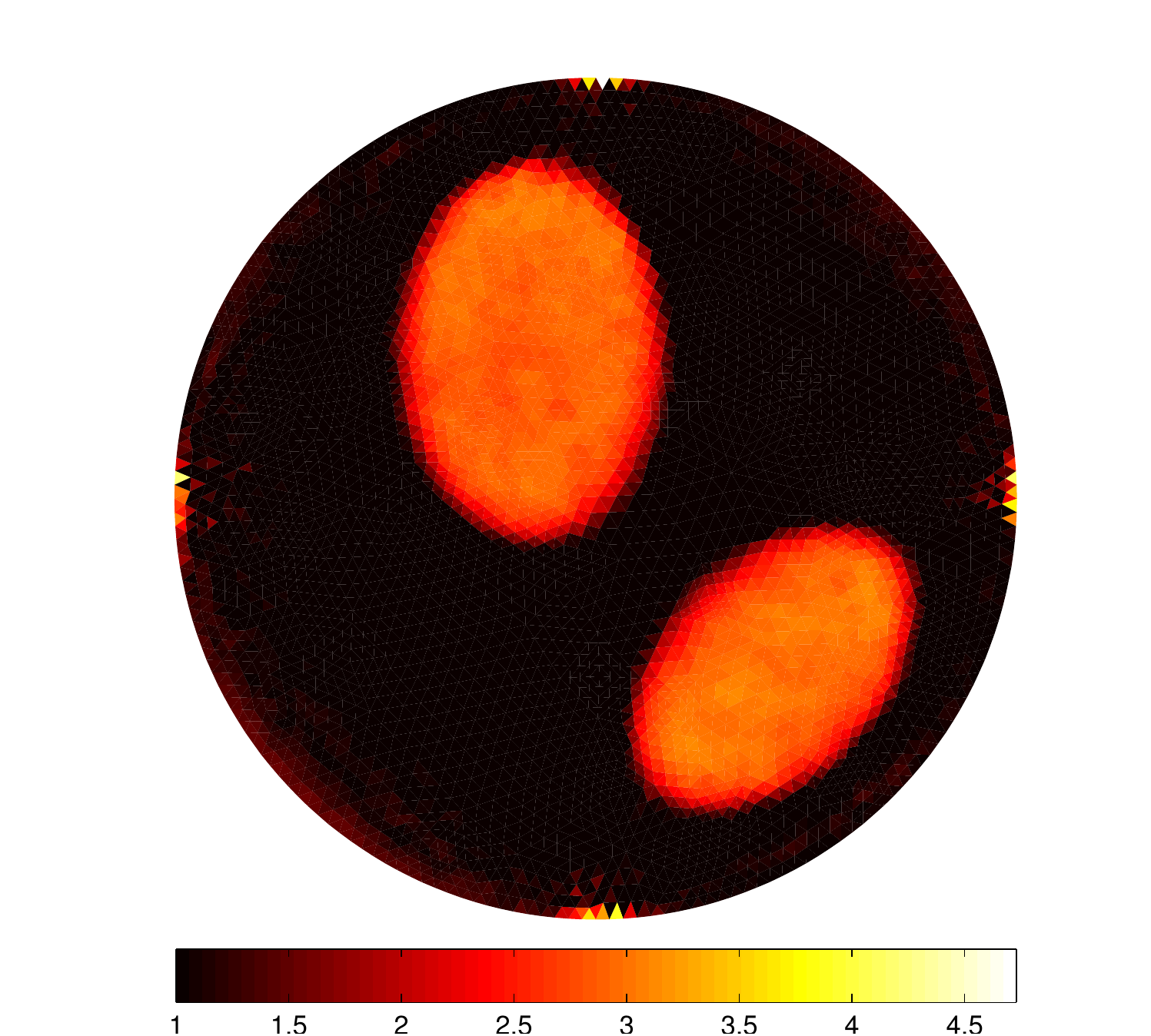}
\end{minipage}
 \caption{\label{fig:7}Reconstructions of optical coefficients by BB method from data added by 5\% Gaussian noise. Top row: reconstruction of $\mu_a$. Bottom row: reconstruction of $\mu_s$.}
\end{figure}

\section{Conclusion}
\label{sec:5}

In this paper, we investigate the reconstruction of absorption and scattering coefficients in QPAT in two cases. Given scattering coefficient, we propose an improved fixed-point iterative method to reconstruct the absorption coefficient and prove its convergence. The advantage of this reconstruction algorithm is that its fast convergence and it does not require the initial guess close to the exact solution. Meanwhile, it does not need to solve adjoint RTE in the optimization approach. For the simultaneous reconstruction of the two coefficients, we apply a state of art BB method, which dose not need linesearch to compute stepsize. Indeed, linesearch involves solving RTE for several times, which is expensive. BB method only use the values of previous two steps to update current estimation. Numerical results show that improved fixed-point iteration can achieve almost the same accuracy with BB method. It takes about quarter the computational time of the BB method. Moreover, the algorithm is stable to noise. In the unknown scattering coefficient case, numerical results show that BB method can obtain quite accurate absorption estimate.  However, the result of the estimation of scattering coefficient is not satisfactory. There is no explicit $\mu_s$ in~\eqref{eq:5}, so it is insensitive to the objective function. There is obvious blur near the borders in piecewise constant templates.

In future, we will design a better error functional to improve the reconstruction of $\mu_s$. Although BB method can recover numerically absorption and scattering coefficients simultaneously, its convergence remains to be studied. Besides, for more practical application, we will do some realistic simulations, such as 3-D templates or real physical QPAT data.

\section{Acknowledgments}
\label{sec:6}

We would like to thank the anonymous referees for their useful comments that help us improve the quality of the paper. We would like to thank Prof.\ Markus Haltmeier (Department of Mathematics, University of Innsbruck Technikestra{\ss}e) and Dr.\ Hao Gao (Wallace H. Coulter Department of Biomedical Engineering, Georgia Institute of Technology) for their helpful advice on RTE solver. We also thank Mr.\ Ji Li (School of Mathematical Sciences, Peking University) for helpful comments on this manuscript draft. This work was supported by NSF grants of China (61421062, 11471024).

\appendix

\section{Convergence of the solver for original RTE}
We use the norm $\|\phi\|:=\sqrt{\sum_{k=0}^{P-1}\omega_k\oint_\Omega\phi_k^2}$ to estimate the accuracy of numerical solution, where $\omega_k$ is defined by $\omega_k=\oint_{\mathcal{S}^{n-1}}L_k(\theta)\rd \theta$ with Lagrangian function $L_k(\theta)$ in $\mathcal{S}^{n-1}$. We denote the solution of original RTE \eqref{eq:2} and angular discretized equation \eqref{eq:34} by $\phi(x,\theta)$ and $[\phi]:=(\phi_k(x))_{k=0}^{P-1}$ respectively, and assume the angular and spatial mesh size are $h_a$ and $h$ respectively. Then there are some convergence results.

\begin{theorem}[\cite{Gao2013}]
  \label{a1}
Assume $\phi(x,\theta)\in C^2(\Omega\times\mathcal{S}^{n-1})$, with vaccum boundary condition and sufficiently fine angular mesh
\begin{equation}
  \label{eq:a1}
  \|\phi-[\phi]\|=\sqrt{\sum_{k=0}^{P-1}\omega_k\oint_{\Omega}(\phi_k(x)-\phi(x,\theta_k))^2\rd x}=O(h_a^2).
\end{equation}
\end{theorem}

The convergence of spatial discrete scheme \eqref{eq:35} is discussed next. Assume triangulation $T_h$ with $h=\sup_{S\in T_h} diam(S)$, $V_{h}^d:=\{v:v|_{S\in T_h}\in P^d(S),v|_{\Omega\backslash S}=0\}$, where $P^d(S)$ is the space of $d$-degree polynomials. Let $(u,v)_S:=\int_S uv \rd x$, $\Gamma_h:=(\cup_{S\in T_h}\partial S)\backslash \partial \Omega$ and $\partial \Omega^k_{-(+)}:=\{x\in \partial \Omega:\nu\cdot \theta_k\leq (>)0\}$. Then the convergence result is as follows.

\begin{theorem}[\cite{Gao2013}]
  With vaccum boundary condition, if $\phi_k^h\in V^d_h$ satisfies
  \begin{equation}
    \label{eq:a2}
    A_k(\phi^h_k,v)=(q,v)+\left<q_b,v\right>_{\partial \Omega_-^k},\ \forall v\in V^d_h,\ k=0,1,\dots,P-1,
  \end{equation}
where
\begin{equation*}
  A_k(\phi_k^h,v)=\sum_S(\theta_k\cdot\nabla \phi_k^h+(\mu_a+\mu_s)\phi_k^h-\mu_s\sum_{k'}\omega_{kk'}\phi_{k'}^h,v)_S+\left<\phi_k^{h^+}-\phi_k^{h^-},v^+\right>_{\Gamma_h^k}+\left<\phi_k^h,v\right>_{\partial \Omega^k_-},
\end{equation*}
with $\phi_k^{h^\pm}=\lim_{\epsilon\rightarrow 0^\pm}\phi_k^h(x+\epsilon \theta_k)$, $\left<u,v\right>_{\Gamma_h^k}=\int_{\Gamma_h}uv|\theta_k\cdot \nu|$, and $\left<u,v\right>_{\partial \Omega^k_-}=\int_{\partial \Omega^k_-}uv|\theta_k\cdot \nu|$, then $\phi^h:=(\phi_k^h)\in (V_h^d)^P$ satisfies
\begin{equation}
  \label{eq:a3}
  \norm{[\phi]-\phi^h}\leq Ch^{d+1/2}|\phi|_{d+1},
\end{equation}
where the $|\phi|^2_{d+1}:=\sum_{k=0}^{P-1}\omega_k|\phi_k|^2_{H^{d+1}}$.
\end{theorem}
Combining \eqref{eq:a1} and \eqref{eq:a3}, when $d=1$ we can obtain the estimate
\begin{equation}
  \label{eq:a03}
  \norm{\phi-\phi^h}=O(h_a^2)+O(h^{3/2}).
\end{equation}

Obviously, the corresponding update scheme of \eqref{eq:a2} is \eqref{eq:35} when $d=1$. Therefore the convergence proof of \eqref{eq:35} is done.

\section{Convergence of the solver for adjoint RTE}
As for adjoint RTE \eqref{eq:31}, we denote the solution of it and its angular discretized equation
\begin{equation}
  \label{eq:a4}
  -\theta_k\cdot \nabla \tilde{\phi}_k+(\mu_a+\mu_s)\tilde{\phi}_k=\mu_s\sum_{k'=1}^P\omega_{kk'}\tilde{\phi}_{k'}+q_k, 1\leq k\leq P
\end{equation}
by $\tilde{\phi}(x,\theta)$ and $[\tilde{\phi}]:=(\tilde{\phi}_k(x))_{k=0}^{P-1}$ respectively. Through similar discussion, we obtain following results.

\begin{theorem}
  \label{a3}
Assume $\tilde{\phi}(x,\theta)\in \mathcal{C}^2(\Omega\times\mathcal{S}^{n-1})$, with vaccum boundary condition and sufficiently fine angular mesh
\begin{equation}
  \label{eq:a5}
  \|[\tilde{\phi}]-\tilde{\phi}\|=\sqrt{\sum_{k=0}^{P-1}\omega_k\oint_{\Omega}(\tilde{\phi}_k(x)-\tilde{\phi}(x,\theta_k))^2\rd x}=O(h_a^2).
\end{equation}
\end{theorem}

\begin{theorem}
  With vaccum boundary condition, if $\tilde{\phi}_k^h\in V^d_h$ satisfies
  \begin{equation}
    \label{eq:a5}
    \tilde{A}_k(\tilde{\phi}^h_k,v)=(q,v)+\left<q_b,v\right>_{\partial \Omega_+^k},\ \forall v\in V^d_h,\ k=0,1,\dots,P-1,
  \end{equation}
where
\begin{equation*}
 \tilde{A}_k(\tilde{\phi}_k^h,v)=\sum_S(-\theta_k\cdot\nabla \tilde{\phi}_k^h+(\mu_a+\mu_s)\tilde{\phi}_k^h-\mu_s\sum_{k'}\omega_{kk'}\tilde{\phi}_{k'}^h,v)_S+\left<\tilde{\phi}_k^{h^-}-\tilde{\phi}_k^{h^+},v^-\right>_{\Gamma_h^k}+\left<\tilde{\phi}_k^h,v\right>_{\partial \Omega^k_+},
\end{equation*}
 then $\tilde{\phi^h}:=(\tilde{\phi}_k^h)\in (V^d_h)^P$ satisfies
\begin{equation}
  \label{eq:a6}
  \norm{[\tilde{\phi}]-\tilde{\phi}^h}\leq Ch^{d+1/2}|\tilde{\phi}|_{d+1},
\end{equation}
where the $|\tilde{\phi}|^2_{d+1}:=\sum_{k=0}^{P-1}\omega_k|\tilde{\phi_k}|^2_{H^{d+1}}$.
\end{theorem}
Combining \eqref{eq:a5} and \eqref{eq:a6}, when $d=1$ we can obtain the estimate
\begin{equation}
  \label{eq:a7}
  \norm{\tilde{\phi}-\tilde{\phi}^h}=O(h_a^2)+O(h^{3/2}).
\end{equation}

Obviously, the corresponding update scheme of \eqref{eq:a5} is \eqref{eq:40} when $d=1$. Therefore the convergence proof of \eqref{eq:40} is done.

\section{Comments on the log-type error functional}
According to \cite{Tarvainen2012}, taking log accelerates the convergence of minimization method significantly. It is due to taking log changes the shape of contours of error functional. From Figure 2 in \cite{Tarvainen2012}, the contours of error functional become less narrow after taking log. In numerical simulations, we find log-type error functional decreasing more fast. Indeed, we can explain it by a one-dimension example. Let $f(x)$ be a continuous function on bounded region $[a,b]$, which satisfies $0< f(x)<1$. Given $h^*=f(x^*)\ (\forall x\in[a,b])$, we can apply least square method to recover $x^*$. Define error functional
\begin{equation}
  \label{eq:a8}
\begin{aligned}
  \mathcal{J}_1(x)&=\frac{1}{2}\norm{f(x)-h^*}_2^2,\\
\mathcal{J}_2(x)&=\frac{1}{2}\norm{\log(f(x))-\log(h^*)}_2^2.
\end{aligned}
\end{equation}
Then the gradients are
\begin{equation}
  \label{eq:a9}
\begin{aligned}
  \nabla \mathcal{J}_1(x)&=f'(x)(f(x)-h^*),\\
\nabla \mathcal{J}_2(x)&=\frac{f'(x)}{f(x)}(\log(f(x))-\log(h^*)).
\end{aligned}
\end{equation}
\begin{theorem}
\label{log}
  For error functionals defined in \eqref{eq:a8}, we claim that there exists $\delta>0$ such that if $x\in (x^*-\delta,x^*+\delta)$, $|\nabla \mathcal{J}_1(x)|\leq |\nabla \mathcal{J}_2(x)|$.
\end{theorem}
\begin{proof}
Assume there exists some $\delta_0$ such that if $x\in (x^*-\delta_0,x^*+\delta_0)\backslash\{x^*\}$, $f(x)\neq h^*$. Otherwise, $|\nabla \mathcal{J}_1|=|\nabla \mathcal{J}_2|$ in $(x^*-\delta_0,x^*+\delta_0)$.

  Since $f(x)$ is continuous, $f(x)\rightarrow f(x^*)\ (x\rightarrow x^*)$, then
  \begin{equation*}
    \lim_{x\rightarrow x^*}\log\left(1+\frac{f(x)-f(x^*)}{f(x^*)}\right)^{\frac{f(x^*)}{|f(x)-f(x^*)|}}=1.
  \end{equation*}
Therefore, for arbitrary $\epsilon>0$, there exists some $\delta_1>0$ such that if $x\in (x^*-\delta_1,x^*+\delta_1)\backslash\{x^*\}$,
\begin{equation}
\label{eq:a10}
  1-\epsilon<\left|\log\left(1+\frac{f(x)-f(x^*)}{f(x^*)}\right)^{\frac{f(x^*)}{|f(x)-f(x^*)|}}\right|<1+\epsilon.
\end{equation}
Dividing \eqref{eq:a10} by $f(x^*)$, we can obtain
\begin{equation*}
  \frac{1-\epsilon}{f(x^*)}<\left|\log\left(1+\frac{f(x)-f(x^*)}{f(x^*)}\right)^{\frac{1}{|f(x)-f(x^*)|}}\right|<\frac{1+\epsilon}{f(x^*)}.
\end{equation*}
Let $0<\epsilon\leq 1-f(x^*)$, we can get
\begin{equation*}
  f(x)<1<\frac{1}{|f(x)-f(x^*)|}\left |\log\left(1+\frac{f(x)-f(x^*)}{f(x^*)}\right)\right|.
\end{equation*}
Clearly, it follows that
\begin{equation}
  \label{eq:a11}
  \left|f'(x)(f(x)-f(x^*))\right|\leq \left|\frac{f'(x)}{f(x)}(\log(f(x))-\log(f(x^*)))\right|.
\end{equation}
Let $\delta=\min\{\delta_0,\delta_1\}$, then $|\nabla \mathcal{J}_1|\leq |\nabla \mathcal{J}_2|$ in $(x^*-\delta,x^*+\delta)$.
\end{proof}

\begin{theorem}
  \label{log2}
For error functionals defined in \eqref{eq:a8}, we claim that there exists $\delta>0$ such that functional
\begin{equation*}
  \mathcal{D}(x):=\mathcal{J}_2-\mathcal{J}_1
\end{equation*}
is monotonically increasing with respect to $\alpha(x):=|x-x^*|$ in domain $[x^*-\delta,x^*+\delta]$.
\end{theorem}

\begin{proof}
Assume there exists $\delta_0$ such that $f(x)\neq h^*$ in $[x^*-\delta_0,x^*+\delta_0]$. Since $f$ is continuous, there exists some $\delta_1<\delta_0$ such that $|f(x)-h^*|$ is increasing with respect to $\alpha(x)$, that is, in interval $[x^*-\delta_1,x^*+\delta_1]$, the closer between $x$ and $x^*$ then the smaller $|f(x)-h^*|$.
  
For arbitrary $x=x^*+\Delta x\in [x^*-\delta_1,x^*+\delta_1]$, let $f(x)=h^*+\Delta f$. Then we have 
\begin{equation*}
\begin{aligned}
  \mathcal{D}(x)&=\frac{1}{2}\norm{\log f(x)-\log h^*}_2^2-\frac{1}{2}\norm{f(x)-h^*}_2^2\\
&=\frac{1}{2}(\log f+f-\log h^* -h^*)(\log f - f -\log h^*+h^*)\\
&=\frac{1}{2}(\log(1+\frac{\Delta f}{h^*})+\Delta f)(\log(1+\frac{\Delta f}{h^*})-\Delta f).
\end{aligned}
\end{equation*}
Clearly, $\log(1+\frac{\Delta f}{h^*})+\Delta f$ and $\log(1+\frac{\Delta f}{h^*})-\Delta f$ are monotonically increasing with respect to $\Delta f$. Combing $\log(1+\frac{\Delta f}{h^*})+\Delta f=\log(1+\frac{\Delta f}{h^*})-\Delta f=0$ when $\Delta f=0$, we obtain $\mathcal{D}(x)$ is non-negative and monotonically increasing with respect to $|\Delta f|$. Therefore, $\mathcal{D}$ is monotonically increasing with respect to $\alpha(x)$. 

Let $\delta=\min\{\delta_0,\delta_1\}$, then this completes the proof.
\end{proof}

From Theorem \ref{log} and Theorem \ref{log2}, we see if $\sup f(x)<1$, taking log not only makes the error functional steeper near the minimizer for some fixed direction, but also the farther between $x$ and $x^*$, the change from $\mathcal{J}_1(x)$ to $\mathcal{J}_2(x)$ is more significant. Then, for multivariate function, taking log makes the contours not so narrow. Thence, log-type error functional improves the convergence of minimization method. As for the general case of $\sup f(x)>1$, we can replace $f(x)$ by $f(x)/c$ for some constant $c>\sup f(x)$.

\section*{References}
 % \bibliography{qpat}

\end{document}